\documentclass[11pt]{article}
\usepackage[letterpaper,margin=1in]{geometry}

\usepackage{setspace}
\usepackage[T1]{fontenc}
\usepackage{amsmath, amsthm, amssymb}
\usepackage[table]{xcolor}
\usepackage{tabularx, booktabs} 
\usepackage{bm}

\usepackage{mathtools}
\usepackage{exscale,relsize}
\usepackage{newpxmath}
\usepackage{inconsolata, stmaryrd}
\usepackage{tikz, graphicx, subfig, float, tcolorbox, framed, venndiagram, listings, color, enumerate, enumitem, pdfpages}  
\usepackage{sidecap, multirow}
\usetikzlibrary{calc}

\usepackage{hyperref}
\usepackage{cleveref}

\definecolor{Header}{RGB}{230,249,255} %
\definecolor{Body}{RGB}{230,242,252}    

\usepackage[boxruled,vlined,commentsnumbered,titlenotnumbered,linesnumbered]{algorithm2e}

\hypersetup{colorlinks,breaklinks,urlcolor=[rgb]{0,0,0},linkcolor=[rgb]{0,0,0.8},citecolor=[rgb]{0,0,0.8}}

\tcbset{
  colframe=black, 
  colback=white, 
  boxrule=0.5pt, 
  arc=2mm, 
  outer arc=2mm, 
  boxsep=2mm, 
}

\DeclareMathOperator{\poly}{poly}

\renewcommand{\le}{\leqslant}
\renewcommand{\ge}{\geqslant}
\renewcommand{\leq}{\leqslant}
\renewcommand{\geq}{\geqslant}
\newcommand{\A}{\mathcal{A}}
\newcommand{\B}{\mathcal{B}}
\newcommand{\E}{\mathcal{E}}
\newcommand{\F}{\mathbb{F}}
\newcommand{\G}{\mathcal{G}}
\newcommand{\I}{\mathcal{I}}
\renewcommand{\L}{\mathcal{L}}

\newcommand{\U}{\mathcal{U}}
\newcommand{\V}{\mathcal{V}}
\newcommand{\X}{\mathcal{X}}

\newcommand{\cost}{{\bm c}}
\newcommand{\wt}{{\bm wt}}
\newcommand{\NMC}{{\textsc{NMC}}}

\newcommand{\rank}{{\mathrm{rank}}}

\newcommand{\rev}{{\mathrm{rev}}}

\newcommand{\val}{{\mathrm{val}}}

\SetNlSty{bfseries}{\color{black}}{}
\newtheorem{theorem}{Theorem}[section] 
\newtheorem{definition}[theorem]{Definition}

\newtheorem{lemma}[theorem]{Lemma}
\newtheorem{corollary}[theorem]{Corollary}
\newtheorem{observation}[theorem]{Observation}
\newtheorem{remark}[theorem]{Remark}

\newcommand{\head}{{\mathrm{head}}}
\newcommand{\tail}{{\mathrm{tail}}}
\newcommand{\trail}{P}

\newcommand{\BMF}{Bounded Max-Flow}

\date{}
\author{
  Keerti Choudhary\thanks{Department of Computer Science and Engineering, IIT Delhi, India. Email: keerti@cse.iitd.ac.in} \and
  Amit Kumar\thanks{Department of Computer Science and Engineering, IIT Delhi, India. Email: amitk@cse.iitd.ac.in} \and
  Lakshay Saggi\thanks{Department of Computer Science and Engineering, IIT Delhi, India. Email: csz228231@cse.iitd.ac.in}
}
\title{Sensitivity Oracles for Matroid Packing, Matroid Covering, and\\ Matching Problems with Applications}

\begin{document}
\maketitle
\thispagestyle{empty}

\begin{abstract}
Sensitivity oracles preprocess a graph so that queries can be answered after any $f$ edge insertions and deletions, without recomputing from scratch. For structural optimization problems the known landscape is limited: for flows and cuts, all known compact oracles handle only $f\le2$ failures; existing oracles for $s$- and global min-cut apply only to undirected graphs; and for matchings, arborescence and spanning-tree packings, and arboricity, no efficient oracle is known for $f>1$. We present a unified algebraic framework based on sensitivity oracles for matroid packing, covering, and parity of sparse linear matroids, yielding the first oracles supporting an arbitrary number $f$ of updates across all of these problems (all constructions randomized Monte-Carlo).

Concretely, we obtain efficient oracles for exact $(s,t)$-max-flow/min-cut, resolving an open problem of Baswana, Bhanja, and Pandey (ICALP'22) with near-optimal space; for all-pairs $k$-bounded flow, generalizing the near-optimal reachability oracle of Brand and Saranurak (FOCS'19, the case $k=1$); the first oracles for any $f$ for directed $s$- and global min-cut; oracles for $k$-disjoint arborescences, $k$-disjoint spanning trees, colorful spanning trees, and arboricity; and oracles for the existence of an $\alpha$-factor, with perfect matching as the case $\alpha=1$.

We further introduce the \emph{subset sensitivity model}, in which updates are confined to a susceptible edge set of size $\sigma$ fixed during preprocessing. Here we decouple updates from the matroid representation and eliminate the dependence on $k$ and the matroid density altogether: all of the above are supported with $\widetilde O(f^\omega)$ query time and $O(f\sigma^2)$ space. We also prove a matching $\Omega(\min\{\sigma^2,n^2\})$-bit lower bound when $f\ge2$, establishing optimality.

\end{abstract}

\newpage
\setcounter{tocdepth}{2}
\tableofcontents
\thispagestyle{empty}

\newpage
\section{Introduction}
\label{sec:introduction}
\setcounter{page}{1}

Real-world networks, including communication, transportation, and infrastructure networks, are inherently subject to failures and transient perturbations. Links may fail, capacities may fluctuate, and new connections may be introduced. 

This motivates the study of \emph{sensitivity oracles} (also called fault-tolerant oracles): compact data structures that preprocess a graph (or a more general combinatorial object) so that, for any small set of $f$ updates (edge insertions, deletions, or weight changes), queries about the updated object can be answered without recomputation. Over the past two decades, the sensitivity setting has been studied extensively for shortest paths and distances~\cite{DTCR08, DuanP09, WeimannY13, BrandS19, ChechikC20, GrandoniVW19, GuR21,KarczmarzS23}, reachability~\cite{BaswanaCR16, Choudhary16, BrandS19}, and connectivity~\cite{PatrascuT:07,DuanP09, BaswanaCR16, LongS22}.

In contrast, the landscape is starkly different for \emph{structural optimization} problems such as max-flows and min-cuts, arborescence and spanning-tree packings, matchings and degree-constrained subgraphs, and matroid intersection or parity. Here, essentially all known oracles handle only $f \le 2$ failures, and the underlying combinatorial techniques appear inherently tied to small values of $f$:
\begin{itemize}[leftmargin=2em, itemsep=1pt, topsep=2pt]
\item For a fixed pair $(s,t)$, 
Baswana, Bhanja, and Pandey~\cite{BaswanaBP22} designed an $(s,t)$-min-cut oracle for $f\le 2$ edge failures with $O(n^2)$ space and $O(1)$ query time, using the submodularity structure of minimum and near-minimum cuts, and posed the case $f>2$ as an open problem. Recently, the authors in~\cite{AhiCPPS26} handled general $f$ via directed flow augmentation~\cite{KimKPW22}, but with space and query time scaling as $2^{O_f(\lambda^4\log k)}$ in the flow value $\lambda$, which is quite expensive for graphs with large cuts or multiple failures. For $f\geq 3$ the relevant minimum cuts interact in increasingly complex ways, and no compact oracle avoiding this exponential dependence is known.
\item For all-pairs min-cut, the only known result is the single-failure oracle of Baswana and Pandey \cite{BaswanaP22}; for $s$-min-cut and global min-cut, oracles are known only for $f \le 2$ and only in undirected graphs~\cite{DinitzN95, Bhanja25}; in directed graphs nothing non-trivial is known.
\item For matchings, degree-constrained subgraphs ($\alpha$-factors), disjoint arborescences, colorful spanning trees, arboricity, and matroid packing/covering/parity, no sensitivity oracles were known for any $f>1$, even though the \emph{dynamic} versions of some of these problems have received attention in the past~\cite{BlikstadMNT23, VosG26, GuptaP13, BernsteinS15, BernsteinS16, VosC26, BhattacharyaKSW23, Kiss23, AzarmehrBR24, Behnezhad23}, along with results about fault-tolerant subgraphs~\cite{AssadiB19, BanerjeeGRS22, BansalCDW24}.
\end{itemize}
Even seemingly simple decision questions like \emph{does the graph remain strongly connected after $f$ updates? does it still admit a perfect matching?} had no efficient answer. This raises the following central question of this paper:

\begin{tcolorbox}
\emph{
Is there a unifying technique that yields compact, efficient sensitivity oracles, for an arbitrary number $f$ of updates, across a vast class of structural optimization problems?
}
\end{tcolorbox}

We answer this question affirmatively, and our route is \emph{algebraic}. Algebraic techniques involving maintaining determinants and inverses of symbolic matrices under low-rank updates have recently proven remarkably effective 
for dynamic and fault-tolerant variants of shortest paths and reachability~\cite{Sankowski05, WeimannY13, BrandS19, GuR21, KarczmarzS23}, dynamic matchings~\cite{BrandNS19, BrandKZ26}, dynamic SCC~\cite{KarczmarzSmule23} along with many other applications~\cite{Brand21}. 
Yet no general framework 
for \emph{sensitivity} oracles for matroid problems has previously been studied. We bridge this gap: rather than solving each graph problem separately, we design sensitivity oracles for the fundamental matroid 
primitives that underlie them all: \emph{matroid packing}, \emph{matroid covering}, and \emph{matroid matching (parity)}.

\subsection{Our Contributions}
Let $M_1, M_2$ be two linear matroids of rank $r$ on a common ground set $E$ of size $m$, with integer element weights in $[-W,W]$. Given $k \ge 1$, the \emph{$k$-packing} problem asks whether there is a set $S \subseteq E$ of size $kr$ that can be partitioned into $k$ pairwise disjoint common bases of $M_1$ and $M_2$ (the two partitions may differ); the weighted version asks for the minimum weight of such a set. The \emph{$k$-covering} problem asks whether the entire ground set can be partitioned into $k$ independent sets of a given matroid. In the \emph{general sensitivity model}, the oracle must support up to $f$ arbitrary element updates (insertions, deletions, or weight changes) anywhere in the ground set.

For \emph{sparse} linear matroids (constant many non-zeros per column of the linear representation), we obtain the following.

\begin{table}[htbp]
\centering
\resizebox{0.92\linewidth}{!}{
\renewcommand{\arraystretch}{1.2}
\begin{tabular}{@{}llccc@{}}
\toprule
\textbf{Problem} & \textbf{Ref} & \textbf{Preprocessing} & \textbf{Space} & \textbf{Query Time} \\
\midrule
$k$-Packing (decision) & Thm~\ref{thm:sensitivity_packing} & $O(k^2 m + (kr)^\omega)$ & $O(km+(kr)^2 \log(kr))$ & $O((f k^2)^\omega)$ \\
$k$-Packing (weighted) & Thm~\ref{thm:sensitivity_packing} & $\widetilde{O}(k^2 m + W (kr)^3)$ & $O(km+W (kr)^3 \log(kr))$ & $\widetilde{O}(W kr (f k^2)^\omega)$ \\
$k$-Covering (decision) & Thm~\ref{thm:sensitivity_covering} & $O((kr)^\omega)$ & $O((kr)^2 \log(kr))$ & $O((f k)^\omega)$ \\
\bottomrule
\end{tabular}}
\vspace{-1mm}
\caption{Sensitivity oracles for matroid $k$-packing and $k$-covering.}
\label{tab:main}
\end{table}

Instantiating this framework, and combining it with new problem-specific reductions, yields the first sensitivity oracles supporting an arbitrary number $f$ of updates for a broad range of problems.

\paragraph{All-pairs Bounded Flow.} For directed unit-capacity graphs with costs in $[-W,W]$ and a flow bound $k$, we obtain an oracle that, for any pair $(s,t)$ and any $f$ edge updates, report $\min(k,\lambda_{G+U}(s,t))$ in $O((fk^2)^\omega\log k)$ time, where $\lambda_{G+U}(s,t)$ denotes the value of $(s,t)$-max-flow in the updated graph. Further, the oracle identifies the nearest $(s,t)$ min-cut and computes the minimum cost of the corresponding flow, when $\lambda_{G+U}(s,t) < k$ (see Table~\ref{table:gen-flow}). Our result extends the reachability and distance sensitivity oracles of Brand and Saranurak~\cite{BrandS19} (the special case $k=1$), which are known to be near-optimal for constant $f$; for $k>1$ no comparable oracles were known even for undirected graphs, where the state of the art was the single-failure all-pairs min-cut oracle of~\cite{BaswanaP22}.

\paragraph{Fixed-pair $(s,t)$-max-flow.} Our most notable application is for a fixed source-sink pair with \emph{arbitrary} flow value. Applying the all-pairs oracle with $k=n$ would be unnecessarily expensive; instead we present a flow-rerouting reduction whose effective capacity bound is $k = O(f)$, thereby eliminating the blowup associated with large flow values. After any $f$ updates the oracle outputs the updated flow value in just $O(f^{3\omega})$ time (see~\Cref{table:gen-flow}).

This fully resolves the open problem of Baswana, Bhanja, and Pandey~\cite{BaswanaBP22} for every $f$. Moreover, the space is essentially optimal for any constant $f\geq 2$ due to $\Omega(n^2)$ space lower bound by Bhanja~\cite{Bhanja24}. In contrast with~\cite{AhiCPPS26}, whose bounds degrade as $2^{O_f(k^4\log k)}$, our bounds are polynomial in $f$ and independent of $k$. 

\begin{table}[htbp]
\centering
\resizebox{0.82\linewidth}{!}{
\renewcommand{\arraystretch}{1.2}
\begin{tabular}{@{}lllll@{}}
\midrule
\textbf{Query pair} & \textbf{\#Failures} & \textbf{Space} & \textbf{Query Time} & \textbf{Ref.} \\
\midrule
\multicolumn{5}{@{}l}{\emph{Prior work}} \\
fixed $(s,t)$ & $f \le 2$ & $O(n^2)$ & $O(1)$  & \cite{BaswanaBP22} \\
fixed $(s,t)$ & $f \le 2$ & $O(\lambda n)$ & $O(1)$  & \cite{AhiCPPS26} \\
fixed $(s,t)$ & any $f$ & $2^{O_f(\lambda^4\log\lambda)}\, n\log n$ & $2^{O_f(\lambda^4\log\lambda)}\log n$  & \cite{AhiCPPS26} \\[1mm]
\midrule
\multicolumn{5}{@{}l}{\emph{This paper}} \\
$V \times V$, $k$-bounded & any $f$ & $\widetilde O((kn)^2)$ & $\widetilde O((fk^2)^\omega)$ & Thm~\ref{theorem:FT-all-pairs-flow} \\
\quad nearest min-cut & any $f$ & $\widetilde O((kn)^2)$ & $O((fk^2)^\omega\, n)$ & Thm~\ref{theorem:FT-all-pairs-flow} \\
\quad min-cost flow & any $f$ & $\widetilde O(W(kn)^3)$ & $\widetilde{O}((fk^2)^\omega\, kWn)$ & Thm~\ref{theorem:FT-all-pairs-flow} \\
fixed $(s,t)$, exact & any $f$ & $\widetilde O((fn)^2)$ & $\widetilde{O}(f^{3\omega})$ & Thm~\ref{theorem:st-flow} \\
\quad nearest min-cut & any $f$ & $\widetilde O((fn)^2)$ & $\widetilde{O}(f^{3\omega}\, n)$ & Thm~\ref{theorem:st-flow} \\
\quad min-cost max-flow & any $f$ & $\widetilde O(W(fn)^3)$ & $\widetilde{O}(f^{3\omega+1}\, Wn)$ & Thm~\ref{theorem:st-flow} \\
\midrule
\end{tabular}
}
\vspace{-1mm}
\caption{Sensitivity oracles for $(s,t)$- max-flow and min-cut in directed graphs. Here $\lambda$ is the value of $(s,t)$-max-flow in $G$.}
\vspace{-2mm}
\label{table:gen-flow}
\end{table}

\paragraph{$s$-min-cut and global min-cut in directed graphs.} No non-trivial sensitivity oracles were known for these problems in directed graphs; undirected results~\cite{DinitzN95, Bhanja25} existed for $f\le 2$. We present the first directed $s$-min-cut and global min-cut oracles for arbitrary $f$ with $O(f^2 n^3\log n)$ space and $\widetilde{O}(f^{3\omega+1})$ query time. Additionally, the oracle reports the corresponding cut in $\widetilde{O}(f^{3\omega}n)$ time. Further, we provide $\Omega(n^2)$ lower bound for these problems, for any constant $f\geq 2$. 
We also present an alternate oracle for $s$-min-cut and global min-cut via arborescence packings (see \Cref{cor:s-global-mincut-via-arb}) that reports the cut \emph{value} in just $\widetilde{O}(f^\omega)$ time, at the cost of $O(fm^2)$ space, exhibiting a complementary space-query trade-off.

\paragraph{Spanning structures.}
Our framework also yields sensitivity oracles for a broad class of spanning structure problems. As a basic example, consider the problem of deciding whether $G$ admits a reachability tree rooted at $s$ after $f$ updates. 
Existing pairwise reachability oracle of Brand and Saranurak~\cite{BrandS19} require $\Omega(n)$ queries\footnote{Note that $\text{poly}(f)$ query time is possible using work of \cite{BrandS19} when updates are restricted to edge deletions, but the general setting of both insertions and deletions is more challenging.}
to certify this property, and with query time depending only on $f$ was previously not known.
We instead obtain an oracle for deciding the existence of $k$ edge-disjoint arborescences rooted at a \emph{query} root $s$ in $O((fk^2)^\omega)$ time. Instantiating $k=1$ yields an $O(f^\omega)$-time oracle for rooted reachability trees. 

As an immediate consequence, since $k$-strong connectivity amounts to existence of $k$-arborescence packing in both directions from any root, we decide $k$-strong connectivity of $G+U$ in $O((fk^2)^\omega)$ time with $O((kn)^2\log n)$ space (\Cref{cor:scc}). For comparison, for $k=1$, the oracle of~\cite{BaswanaCR19scc} reports the strongly connected components after $f$ failures using $O(2^f n^2)$ space and $\widetilde{O}(2^f n)$ query time; for the decision problem we replace both exponential dependencies by $\mathrm{poly}(f)$.

The same guarantees extend to deciding the existence of $k$ edge-disjoint spanning trees. We also obtain the first sensitivity oracle for colorful (rainbow) spanning trees, with $O(f^\omega)$ query time and quadratic space, and additionally provide a matching lower bound on the oracle size. Finally, we give the first sensitivity oracle for deciding whether the updated graph has arboricity (the minimum number of forests needed to cover edges of a graph) at most $k$, answering queries in $O((fk)^\omega)$ time (see \Cref{table:gen-spaning-struct}).

\begin{table}[!ht]
\centering
\resizebox{0.72\linewidth}{!}{
\renewcommand{\arraystretch}{1.2}
\begin{tabular}{@{}lccc@{}}
\midrule
\textbf{Problem} & 
\textbf{Space} & 
\textbf{Query Time} &
\textbf{Ref}\\ 
\midrule
$k$-Disjoint Arborescences  
& 
$O((kn)^2 \log n)$ & 
$O((fk^2)^\omega)$ &
Thm~\ref{thm:sensitivity_arborescences} \\

$k$-Disjoint Spanning Trees~ \  
& 
$O((kn)^2 \log n)$ & 
$O((fk^2)^\omega)$ &
Thm~\ref{thm:sensitivity_k_spanning_trees} \\

Colorful Spanning Tree  
& 
$O(n^2 \log n)$ &
$O(f^\omega)$ &
Thm~\ref{thm:sensitivity_colorful} \\

$k$-Forest Covering  
& 
~ \ $O((kn)^2 \log n)$ \ & 
$O((fk)^\omega)$ &
Thm~\ref{thm:sensitivity_k_forests} \\
\midrule
\end{tabular}
}
\caption{Sensitivity oracles (decision version) for spanning structure problems. All results have corresponding weighted versions with additional $(Wkn)$-multiplicative factors in space and query.}
\label{table:gen-spaning-struct}
\end{table}

\paragraph{Matroid matching, matchings, and $\alpha$-factors.}
Finally, our framework extends to \emph{matroid matching (parity)}: for a $c$-sparse linear matroid of rank $r$ on $m$ pairs, we obtain sensitivity oracles deciding the existence of a parity basis after $f$ pair updates in $O((c^2f)^\omega)$ time using $O(r^2\log r)$ space. Since matroid parity generalizes both matroid intersection and non-bipartite matching, these oracles have immediate implications for matchings and $\alpha$-factors.
(Given a function $\alpha: V \to \mathbb{Z}^+$, an \emph{$\alpha$-factor} is a subgraph of $G$ in which every vertex $v$ has degree exactly $\alpha(v)$. This generalizes the notion of perfect matching where $\alpha$ is identically one.)

For degree bounds $\alpha(v)\le k$, this yields an $\alpha$-factor existence oracle with $O((kn)^2\log n)$ space and $O((fk^2)^\omega)$ query time. As a special case, setting $\alpha= 1$ gives a perfect matching oracle with $O(n^2\log n)$ space and $O(f^\omega)$ query time. We complement this with an $\Omega(n^2)$ space lower bound for perfect matching oracles for $f\ge2$ updates, establishing space optimality up to logarithmic factors. The standard reduction from maximum to perfect matching further gives maximum matching size within the same bounds up to polylogarithmic factors. Handling \emph{arbitrary} degree bounds requires a different approach: the standard expansion is prohibitively large, and unlike maximum flow, $\alpha$-factors admit no natural residual graph. We develop a residual-graph analogue based on closed alternating trails in the symmetric difference of two $\alpha$-factors, yielding an oracle with $O((fn)^2\log n)$ space and $O(f^{3\omega})$ query time, independent of $\max_v\alpha(v)$.
\medskip

Note that all decision oracles above achieve $\widetilde{O}_{f,k}(n^2)$ or $\widetilde{O}_{f,k}(n^3)$ space, work for sparse matroids, and have $\mathrm{poly}(fk^2)$ query time. The dependence on $k$ arises from our generic reduction, which maps $f$ updates in the original instance to $\Theta(fk^2)$ updates in the auxiliary matrix representation on which our framework operates. This leads to a natural question: can one eliminate the dependence on $k$ and obtain query time depending only on the number $f$ of updates?

\begin{tcolorbox}
\emph{Can the dependence on $k$ (and on the matroid density) be removed from the query time altogether, possibly at the cost of more but still polynomial space?}
\end{tcolorbox}

\paragraph{Removing the dependence on $k$, and Subset Sensitivity.}
We answer this question affirmatively. The $k$-dependence in the general model arises because every update modifies the oracle's internal matroid representation, and the cost of each such modification grows with $k$. Our key idea is to \emph{decouple updates from the representation}. The price of decoupling is space, which now scales with the square of the set of elements that may ever be added or deleted, which in the worst case is $O(m^2)$. 

Additionally, we observe that if the set of edge updates is restricted to a smaller set of size $\sigma$, then the space of our oracle has only a quadratic dependence on $\sigma$.
We therefore introduce the \emph{subset sensitivity model}, where, a set $E_Q$ of \emph{susceptible} elements of size $\sigma$, possibly including elements not initially present, is specified at preprocessing time, and all updates are restricted to $E_Q$. This captures applications where changes are localized,
such as a known set of vulnerable links or candidate edges, and reduces the space requirement to $O(\sigma^2)$. 

Applying this primitive throughout our reductions yields subset sensitivity oracles for $(s,t)$-max-flow, $s$- and global min-cut, all-pairs bounded flow, arborescences, spanning trees, colorful spanning trees, arboricity, matroid parity, and arbitrary-degree $\alpha$-factors with $O(f^\omega)$ query time for decision variants (see \Cref{table:subset}). This trade-off is inherent: we prove that $\Omega(\min\{\sigma^2,n^2\})$ bits are necessary for flows, cuts, colorful spanning trees, and matching in this model, already for $f=2$ (see \Cref{sec:lb}), establishing the optimality of our $O(\sigma^2)$ space bounds.

As a direct application with $k=1$, we obtain subset sensitivity oracles (see \Cref{cor:subset}) for reachability with $O(\sigma^2)$ space and $O(f^\omega)$ query time, and distances with $\widetilde{O}(Wn\sigma^2)$ space and $\widetilde{O}(Wnf^\omega)$ query time.
This matches the query time of the general-model reachability and distance oracle of~\cite{BrandS19}, while replacing the $n^2$ space dependence by $\sigma^2$.

\begin{table}[htbp]
\centering
\begin{center}
\resizebox{0.76\linewidth}{!}{
\renewcommand{\arraystretch}{1.2}
\begin{tabular}{@{}lccc@{}}
\midrule
\textbf{Problem} & \textbf{Space} & \textbf{Query Time} & \textbf{Ref} \\
\midrule
$V_0\times V_0$ Bounded Flow $(V_0\subseteq V)$ & $O((\sigma + k|V_0|)^2)$ & $O((f + k)^\omega)$ & Thm~\ref{theorem:subset-allpairs-flow} \\
$(s,t)$-Max-Flow & $O(\sigma^2)$ & $\widetilde{O}(f^\omega)$ & Thm~\ref{theorem:subset-st-flow} \\
$s$-Min-Cut Value & $O(f\sigma^2)$ & $\widetilde{O}(f^{\omega })$ & Cor~\ref{cor:s-global-mincut-via-arb} \\
Global Min-Cut Value & $O(f\sigma^2)$ & $\widetilde{O}(f^{\omega })$ & Cor~\ref{cor:s-global-mincut-via-arb} \\
$k$-Arborescences & $O(\sigma^2)$ & $O((f + k)^\omega)$ & Thm~\ref{theorem:subset-arborescences} \\
$k$-Spanning Trees & $O(\sigma^2)$ & $O(f^\omega)$ & Thm~\ref{thm:subset_sensitivity_k_spanning_trees} \\
Colorful Spanning Tree & $O(\sigma^2)$ & $O(f^\omega)$ & Thm~\ref{thm:subset_sensitivity_colorful} \\
$k$-Forest Covering & $O(\sigma^2)$ & $O(f^\omega)$ & Thm~\ref{thm:subset_sensitivity_k_forests} \\
$r$-rank Matroid Parity & $O(\sigma^2)$ & $O(f^\omega)$ & Thm~\ref{thm:parity_subset_sensitivity} \\
Arbitrary degree $\alpha$-factor & $O(\sigma^2)$ & $O(f^\omega)$ & Thm~\ref{thm:subset-arbitrary-alpha} \\
\midrule
\end{tabular}
}
\end{center}
\vspace{-3mm}
\caption{Subset sensitivity oracles (decision variants).}
\label{table:subset}
\end{table}

\subsection{Technical Overview}
\label{sec:techniques}
We now describe the main ideas underlying our framework. To illustrate the key techniques, we develop our algebraic approach through the lens of the $(s,t)$-max-flow sensitivity oracle. We first restrict attention to the decision (unweighted) versions; the weighted versions follow naturally by introducing an indeterminate $Y$ whose exponents encode weights. Throughout, $\omega$ is the matrix multiplication exponent and $\F$ is a sufficiently large finite field.

\subsubsection{From Reachability to $k$-Bounded Flow}
Consider the problem of preprocessing a directed graph $G$ such that for any query consisting of a subset  of $f$ edge updates, the updated $(s,t)$-max-flow value can be quickly reported. Prior sensitivity oracles for cuts and flows have been predominantly combinatorial~\cite{BaswanaBP22, AhiCPPS26, BaswanaP22}. These methods analyze the structure of minimum and near-minimum $(s,t)$-cuts and how failures interact with them. While successful for small $f$, this analysis becomes increasingly intricate as the number of failures grows, causing query times or space to blow up exponentially with the flow value~\cite{AhiCPPS26}.

We instead take an algebraic approach. Our starting point is the work of Brand and Saranurak~\cite{BrandS19}, who obtained reachability and distance sensitivity oracles for any number of failures. Since reachability is the special case of deciding whether the $(s,t)$-max-flow is at least $1$, a natural question is:

\begin{quote}
\emph{Can this algebraic paradigm be lifted from reachability ($k=1$) to $k$-bounded flow---deciding whether at least $k$ edge-disjoint $(s,t)$-paths survive $f$ edge updates?}
\end{quote}

We answer this question by building a general algebraic framework based on matroids.

\paragraph{Modeling Flow as Matroid Intersection.}
Given a directed graph $H=(V_H, E_H)$ with designated terminals $s$ and $t$, we can encode the existence of exactly $\lambda$ edge-disjoint $(s, t)$-paths as a matroid intersection problem. We first augment $H$ by adding exactly $k$ zero-cost self-loops at every vertex $v \in V_H \setminus \{s, t\}$, and $k-\lambda$ zero-cost self-loops at $s$ and $t$. We then define two partition matroids on this augmented edge set:
\begin{itemize}
    \item $M_{\mathrm{out}}$: A subset $S \subseteq E_H$ is independent if every vertex $v \in V_H \setminus \{t\}$ has at most one outgoing edge in $S$. 
    \item $M_{\mathrm{in}}$: A subset $S \subseteq E_H$ is independent if every vertex $v \in V_H \setminus \{s\}$ has at most one incoming edge in $S$.
\end{itemize}
A common basis of $M_{\mathrm{out}}$ and $M_{\mathrm{in}}$ is a subset $S$ that satisfies both degree constraints tightly. Combinatorially, such a set $S$ forces the existence of exactly $\lambda$ edge-disjoint $(s, t)$-paths, with the self-loops absorbing the remaining degree budget at each vertex. Furthermore, a minimum-weight common basis corresponds exactly to the minimum-cost routing of these paths.
\medskip

To test for $k$-bounded flow, we must check if the graph can support the packing of $k$ such structures. This is equivalent to checking if the ground set $E_H$ contains a subset that can be partitioned into $k$ common bases of $M_{\mathrm{out}}$ and $M_{\mathrm{in}}$.

Our resolution to this is to pass to \emph{$k$-fold unions}. The $k$-fold union $M^{\vee k}$ of a matroid $M$ is a matroid whose independent sets are the union of $k$ independent sets of $M$. Because partition matroids are strongly base-orderable, checking for a $k$-packing of common bases is perfectly equivalent to checking for a \emph{single} common basis of the $k$-fold unions $M_{\mathrm{out}}^{\vee k}$ and $M_{\mathrm{in}}^{\vee k}$.

To represent these unions algebraically, we use the randomized construction of Narayanan, Saran and Vazirani~\cite{NarayananSV94}. If a linear matroid $M$ of rank $r$ has a representation matrix $A \in \F^{r \times m}$, its $k$-fold union $M^{\vee k}$ can be represented by stacking $k$ copies of $A$, each scaled column-wise by a random diagonal matrix $C_i$:
$$
U = \begin{pmatrix} A C_1 ; A C_2 ; \ldots ; A C_k \end{pmatrix} \in \F^{kr \times m}
$$
Let $U_{\mathrm{out}}$ and $U_{\mathrm{in}}$ be the $kr \times m$ union representation matrices for $M_{\mathrm{out}}^{\vee k}$ and $M_{\mathrm{in}}^{\vee k}$. We associate a formal variable $Y$ and an independent random scalar $x_e$ with each element $e$, forming a diagonal weight matrix $D \in \F[Y]^{m \times m}$ where $D_{e,e} = x_e Y^{W + \wt(e)}$.

Then the \emph{mixed Laplacian} is defined as:
\[
\L_{\mathrm{pack}} \;=\; U_{\mathrm{out}} \, D \, U_{\mathrm{in}}^T \;\in\; \F^{kr\times kr}
\]
By a $k$-fold Cauchy-Binet expansion, the determinant of $\L_{\mathrm{pack}}$  enumerates the common bases of the unions. \emph{Thus, $k$-bounded flow is reduced entirely to evaluating whether $\det(\L_{\mathrm{pack}}) \neq 0$ (for existence) and isolating its minimum $Y$-degree (for minimum cost).}

\paragraph{Sensitivity via Sparse Determinant Updates.}
The power of this formulation in the sensitivity setting lies in sparsity. An update to an edge $e$ in the graph only modifies the $e$-th rank-one term in the Laplacian sum $\L_{\mathrm{pack}} = \sum_{e\in E} D_{e,e} (u_{\mathrm{out}})_e (u_{\mathrm{in}})_e^T$. Because the base partition matroids are extremely sparse ($O(1)$ non-zeros per column), the union columns $(u_{\mathrm{out}})_e$ and $(u_{\mathrm{in}})_e$ have at most $O(k)$ non-zeros. Consequently:
\begin{quote}
\em Any $f$ edge updates in the graph translate to a sparse additive update to $\L_{\mathrm{pack}}$ containing only $O(fk^2)$ non-zero entries.
\end{quote}

To evaluate this perturbed determinant, we utilize the Matrix Determinant Lemma, which allows us to compute the determinant of a matrix under a low-rank update in time proportional to the update size, provided we have precomputed its inverse. We leverage the data structure of Brand and Saranurak~\cite{BrandS19} to achieve this.

\emph{Remark on Singular Matrices:} The original tool of \cite{BrandS19} required the base matrix $\L_{\mathrm{pack}}$ to be invertible (i.e., the base graph must already possess the required flow). We remove this limitation by introducing a random low-rank numeric pad $UV^T$ during preprocessing that forces invertibility. When a query arrives, we fold the numeric pad back into the query update as a structured low-rank perturbation (\Cref{lemma:BrandS-det-General-Q}), allowing the determinant lemma to work seamlessly even if the base matrix is identically singular.

With this tool, an $f$-edge update can be evaluated in $O((fk^2)^\omega)$ time for the decision version, and $\widetilde{O}(Wkr(fk^2)^\omega)$ time for the weighted version. 

\paragraph{Extension to All-Pairs queries via Parallel Elements.}
To extend this to dynamic all-pairs ($V \times V$) queries without blowing up the matrix dimension, we introduce a static dummy source $s^*$ and dummy sink $t^*$ during preprocessing. At query time, to test $(s,t)$-flow, we dynamically insert $k$ parallel edges from $s^*$ to $s$ and from $t$ to $t^*$. 
Crucially, parallel copies of an edge share the exact same sparse column support in the base matroid representations. Thus, adding $k$ parallel edges only perturbs an $O(k)\times O(k)$ block in the Laplacian-contributing $O(k^2)$ modified entries, exactly the same cost as a \emph{single} edge update. Therefore, an all-pairs query with $f$ updates costs $O((fk^2 + k^2)^\omega) = O((fk^2)^\omega)$ time, rather than a naive multiplicative blowup of $O(((f+k)k^2)^\omega)$.

\subsubsection{Two Strategies for Exact $(s,t)$-Max-Flow}
The bounded flow oracle performs well for small flow values, whereas extending it to exact, unbounded $(s,t)$-max-flow remains challenging. In particular, setting $k=n$ in the all-pairs oracle leads to extremely slow query time.

We develop two distinct strategies to resolve this: one focusing on algorithmic rerouting, and the other on structurally decoupling the matrix updates.

\paragraph{Strategy 1: Residual Graph Reduction}
Here we consider the question if recomputing the exact flow value after $f$ edge updates can be reduced to All-pairs bounded flow problem. Our core insight is that \emph{sensitivity does not require computing the entire flow, only repairing it.}

During preprocessing, we compute an exact max-flow $h$ of value $\lambda_G(s,t)$. Consider a set of edge failures $F$, of which $d \le f$ edges actually carried flow in $h$. Deleting these $d$ edges destroys exactly $d$ units of flow, leaving $d$ units of excess at the severed tail vertices and $d$ units of deficit at the head vertices.

To repair the flow, we must re-route these $d$ units of excess to the deficits. This is equivalent to finding $d$ edge-disjoint paths in the \emph{residual graph} $(G-F)_h$. By attaching a dummy source to the excess vertices and a dummy sink to the deficit vertices, repairing the exact max-flow is reduced completely to an all-pairs bounded-flow query on the residual graph with an effective capacity bound of  $k = d \le f$.

To support insertions and accurately track the new flow value, we pad the base graph $G$ with $2f$ zero-cost parallel $(s,t)$-buffer edges, saturating $f$ of them in $h$. Insertions and deletions simply force the repaired flow to saturate fewer or more of these buffer edges, allowing us to obtain the new max-flow value $\lambda_{G+U}(s,t)$ directly off the buffer usage via binary search (\Cref{lemma:reroute-1}).

By plugging this residual graph formulation back into our all-pairs bounded flow oracle with $k=O(f)$, we drop the global capacity $n$ from the query time entirely. This resolves the open problem of~\cite{BaswanaBP22} by yielding an exact max-flow oracle bounded purely by $\widetilde{O}(f^{3\omega})$ query time and $O((fn)^2 \log n)$ space. Costs are handled seamlessly by tracking the cost of the rerouting paths in the residual graph (\Cref{lemma:reroute-2}).

\paragraph{Strategy 2: The Subset Sensitivity Model (Decoupling)}

While the residual rerouting strategy successfully removed $k$ from the query time, the query time still turns out to be $O(f^{3\omega})$. This leads to a second question:

\begin{quote}
\emph{Can we remove the dependence on $k$ from the algebraic update time directly, maintaining an $O(f^\omega)$ query cost while simultaneously shrinking the space requirement?}
\end{quote}

The algebraic bottleneck in the general model occurs because every physical edge update forces the oracle to explicitly rewrite $O(k^2)$ entries in the dense union representation. 

We resolve this by \emph{decoupling the updates from the representations.} 
We \emph{freeze} the representation matrices $A := U_{\mathrm{out}}$ and $B := U_{\mathrm{in}}$ at preprocessing. Instead of updating the matroids directly, we realize every network update as a simple \emph{diagonal switch} in the weight matrix $D$. An element is either toggled on, off, or reweighted. For an update set $S \subseteq Q$ of size $f$, the query matrix becomes:
\[
L' \;=\; A D' B^T \;=\; L + A_S\, \Delta_S\, B_S^T
\]
and by \Cref{lemma:matrix_det},
\[
\det(L') \;=\; \det(L)\cdot\det\bigl(I_f + B_S^T L^{-1} A_S\;\Delta_S\bigr).
\]
The only data a query needs is the $f\times f$ matrix $B_S^T L^{-1} A_S$. Interestingly, we observe that if the locus of edge updates is confined to a \emph{susceptible set} $Q$ of size $\sigma$ declared at preprocessing, then it suffices to compute and store the $\sigma\times \sigma$ principal submatrix 
\[
U_Q \;:=\; B_Q^T\, L^{-1}\, A_Q.
\]
This allows us to also introduce the \textbf{Subset Sensitivity Model}, where a susceptible set of elements $Q$ (of size $\sigma$) is designated at preprocessing time. We thus obtain the following result.

\begin{lemma}[Subset determinant oracle; see \Cref{lemma:det}]
\label{lem:ov-det}
Let $A, B \in \F_d[Y]^{r\times m}$, let $D$ be diagonal, $L = ADB^T$ (possibly singular), and let $Q\subseteq[m]$, $|Q|=\sigma$, be fixed at preprocessing. There is a data structure that, for any $D'$ differing from $D$ in $\le f$ positions within $Q$, computes $\det(AD'B^T)$:
for $d=0$, with $O(\sigma^2 r^{\omega-2} + r^\omega)$ preprocessing, $O(\sigma^2)$ space, and $O(f^\omega)$ query time.
\end{lemma}

By leveraging this subset determinant lemma on our fixed $(s,t)$ flow formulations, we remove the dependence on matrix dimensions and packing parameters from the query time, yielding an exact $(s,t)$-max flow oracle with an optimal $\widetilde{O}(f^\omega)$ query time and strict $O(\sigma^2)$ space.

\subsubsection{Other Applications of Matroid Intersection}
Our algebraic framework extends far beyond $(s,t)$-flows. We discuss some of these applications below.

\vspace{-2mm}

\paragraph{$s$-min-cut and global min-cut.}
The $s$-min-cut is $\min_{v\ne s}\lambda(s,v)$; storing an $(s,v)$ oracle for every $v$ makes a query $n$ evaluations, which is too slow. For \emph{deletions only} there is a simple fix: if the $s$-min-cut decreases, some deleted edge $(x_i,y_i)$ crosses the new minimum cut, hence the new cut separates $s$ from some $y_i$; so it suffices to query the $\le f$ sinks $\{y_i\}$, plus compare against the old value. \emph{Insertions break this}: the new min-cut may avoid all updated edges yet still be smaller than the old cut \emph{would have been after insertions}, the baseline itself moves. We restore the argument with a preprocessing trick based on Frank's augmentation approach~\cite{Frank90}. We compute an \emph{optimal} set $I^*$ of $f$ edges whose insertion maximizes the $s$-min-cut, and let $\lambda^* = \mu_{G+I^*}(s)$. For any update set $U$ (with $|U|\le f$ insertions), $\mu_{G+U}(s) \le \lambda^*$ by optimality of $I^*$; and if the minimum cut of $G+U$ is not crossed by any edge of $I^*\cup U$, its capacity is identical in $G+U$ and $G+I^*$, hence at least $\lambda^*$. Therefore, we have (\Cref{lemma:insertion-deletion-cut})
\[
\mu_{G+U}(s) \;=\; \min\Bigl(\lambda^*,\; \min_{y\in Y,\,y\neq s}\lambda_{G+U}(s,y)\Bigr),
\qquad Y = \text{heads of } I^*\cup U,\ |Y|\le 2f,
\]
reducing the query to $O(f)$ fixed-pair max-flow queries against precomputed $(s,v)$-oracles. Global min-cut follows by symmetry: it equals $\min(\mu_G(s), \mu_{G^{\mathrm{rev}}}(s))$ for any fixed $s$.

\paragraph{Spanning structures and arboricity}
Arborescence packing is the second application of the same framework: by Edmonds' theorem, $k$ disjoint $z$-rooted arborescences correspond to a $k$-packing of the graphic matroid against the in-degree partition matroid. To support a \emph{query-time root} $s$, we keep an isolated dummy root $z$ at preprocessing and, at query time, insert $k$ parallel edges $(z,s)$; by the cheap-parallel-elements refinement this bundle costs only $O(k^2)$ additional Laplacian entries.

As in $(s,t)$-max-flow scenario, the parallel copies of one element share the same sparse column support, so all $k$ copies together touch only an $O(k)\times O(k)$ submatrix ($O(k^2)$ entries total), the same as a \emph{single} update (\Cref{rem:parallel-edges-update-packing-laplacian}). Thus a query consisting of $f$ genuine updates plus a bundle of $k$ parallel query edges costs $O((fk^2 + k^2)^\omega) = O((fk^2)^\omega)$, not $O(((f+k)k^2)^\omega)$. This observation is small but load-bearing: it is what makes \emph{dynamic-root} arborescence queries and \emph{all-pairs} flow queries as cheap as static-root ones.
So a dynamic-root query is exactly as cheap as a static one with query time $O((fk^2)^\omega)$, and $O(f^\omega)$ for $k=1$, which is the single-source reachability-tree oracle highlighted in the introduction. Spanning-tree packing ($M_1 = M_2 = $ graphic), colorful spanning trees (graphic vs.\ color-partition, $k=1$), and arboricity ($k$-covering of the graphic matroid, using the random-augmentation covering oracle) follow the same recipe with $\delta = f$ column updates.

Arboricity (partitioning a graph into $k$ forests) is a \emph{covering} problem. Here the goal is to determine if the entire edge set is independent in the $k$-fold graphic union, i.e., whether the $kr \times m$ (where $r=n-1$) representation union matrix $U$ has full column rank. We convert rank to determinant by squaring it up, appending $kr-m$ uniformly random columns $R$. A Laplace-expansion argument proves that the square matrix is non-singular, $\det([U \mid R]) \neq 0$, w.h.p.\ exactly when the rank is full (\Cref{lemma:aug-linear-rep}) and the graph is $k$-coverable. Since the random columns never change and each edge update touches $O(k)$ entries, we apply our determinant oracles directly, yielding a query time of $O((fk)^\omega)$ (\Cref{thm:sensitivity_covering,thm:sensitivity_k_forests}).

\subsubsection{Matroid Parity and $\alpha$-Factors}

\paragraph{Sparse linear matroid parity.}
The parity (matroid matching) problem admits an analogous algebraic formulation: the determinant of the skew-symmetric matrix $ADB^T - BDA^T$ is the square of a Pfaffian enumerating parity bases (\Cref{thm:parity_enumeration}). For $c$-sparse matroids, $f$ pair updates modify $O(c^2 f)$ entries, and this yields the parity oracle (\Cref{thm:parity_general_sensitivity}), and in turn the bounded-degree $\alpha$-factor and perfect matching oracles.

\paragraph{$\alpha$-factors with arbitrary degree bounds.}
For arbitrary $\alpha$, the parity instance has dimension $\sum_v\alpha(v)$, which can be $\Theta(nm)$. To tackle this, we fix a min-cost $\alpha$-factor $M$ at preprocessing; for any updated optimum $M'$, the symmetric difference $M\triangle M'$ decomposes into closed alternating trails. We prove two facts (i) there exists optimal $M'$ such that the difference decomposes into at most $f$ trails; and (ii) a trail with no closed alternating sub-trail meets each vertex in at most two $M$-edges and two $M'$-edges (else a sub-trail could be extracted, an argument inspired by the bicoloured-cycle technique of \cite{FleischnerS05}). Thus, an optimal $\alpha$-factor of $G+U$ drops at most $2f$ edges of $M$ at every vertex. This local bound lets us build an auxiliary two-layer graph (a drop layer for $M$-edges, an add layer for non-$M$-edges, coupled by parallel edges and unit budget gadgets), on which factors of $G+U$ correspond exactly, with matching costs, to factors of the auxiliary graph (\Cref{thm:bijection}). The bounded-degree oracle with $k=O(f)$ then gives $O((fn)^2\log n)$ space and $O(f^{3\omega})$ query time, independent of $\max_v\alpha(v)$ (\Cref{thm:sensitivity-arbitrary-alpha}).

\subsection{Related Works}

\vspace{2mm}
\noindent
\emph{Matroid Algorithms.~}
Polynomial-time algorithms for matroid union and intersection date back to
Edmonds~\cite{Edmonds2003} and Lawler~\cite{Lawler75}, followed by a long
line of work improving exact and approximate algorithms for matroid
intersection, union, packing, and covering~\cite{Cunningham86, Harvey09,
Quanrud24, Terao25}. More recently, there has been growing interest in
matroid algorithms over uncertain inputs. Blikstad, Mukhopadhyay, Nanongkai,
and Tu~\cite{BlikstadMNT23} maintain a minimum-weight basis in the
decremental setting using $\widetilde{O}(\sqrt{\rank(M)})$ rank queries per
update. Vos and Grilnberger~\cite{VosG26} gave a deterministic dynamic
$(1\pm\epsilon)$-approximation for the packing number $\Phi$ and covering
number $\beta$ of a matroid, using $\Phi \cdot \mathrm{poly}(\log n,
1/\epsilon)$ (resp.\ $\beta \cdot \mathrm{poly}(\log n, 1/\epsilon)$) rank
queries per update. Sparsification techniques of Huang and
Sellier~\cite{HuangS24} led to progress in streaming matroid intersection;
Chandrasekaran, Chekuri, and Zhu~\cite{ChandrasekaranCZ25} initiated the
study of base packing in online matroids, and Buchbinder, Gupta, Hathcock,
Karlin, and Sarkar~\cite{BuchbinderGHKS24} gave an
$\widetilde{O}(\mathrm{OPT})$-competitive algorithm for intersecting a
matroid with an online partition matroid. Fault-tolerant aspects of matroids
have also been studied recently~\cite{BentertFGM25, DeyK26}.

\vspace{4mm}
\noindent
\emph{Static Algebraic Algorithms.~}
Cheung, Lau, and Leung~\cite{CheungLL11} pioneered the use of fast matrix
multiplication for graph connectivity, solving All-Pairs Connectivity in
$\widetilde{O}(m^\omega)$ time. Georgiadis, Italiano, Karanasiou,
Papadopoulos, and Parotsidis~\cite{GeorgiadisIKPP17} presented a
deterministic algorithm for all-pairs $2$-bounded edge (vertex)
connectivity in $\widetilde{O}(n^\omega)$ time.
Authors in~\cite{AbboudGIKPTUW19} extended this to general $k$ in the
vertex-capacitated setting, with a randomized $O((nk)^\omega)$-time
algorithm; they further showed that on DAGs, all-pairs $k$-bounded max-flow
can be solved in $O((k\log n)^{4^k+o(k)} \cdot n^\omega)$ time, and in
$O(2^{O(k^2)} \cdot mn)$ time for sparse DAGs. The framework
of~\cite{CheungLL11} was lifted to the $k$-bounded setting ($k$-APC) by
Akmal and Jin~\cite{AkmalJ23} in $\widetilde{O}((kn)^\omega)$ time, and
refined by Akmal~\cite{Akmal24} via enumerating polynomials. Analogous
algebraic algorithms exist for linear matroid
intersection~\cite{Harvey09}, linear matroid parity~\cite{CheungLL14}, and
weighted $\alpha$-factors~\cite{GabowS21_I, GabowS21_II}.

\vspace{4mm}
\noindent
\emph{Min-Cut Sensitivity Oracles.~}
Results of Picard and Queyranne~\cite{PicardQ82} and Dinitz and
Vainshtein~\cite{DinitzV00} imply that a single edge deletion's effect on
the $(s,t)$-max-flow value can be detected from the structure of the flow
itself. For a fixed pair $(s,t)$, Baswana, Bhanja, and
Pandey~\cite{BaswanaBP22} designed the first dual-failure sensitivity oracle
for $(s,t)$-min-cut, with $O(n^2)$ space and $O(1)$ query time, and posed
the design of compact oracles for more than two failures as an open problem.
Very recently, Ahi, Choudhary, Pande, Pushpraj, and Saggi~\cite{AhiCPPS26}
constructed a max-flow sensitivity oracle for one and two failures with
$O(\lambda n)$ space and $O(n)$ query time, together with an $\Omega(n)$
lower bound for $f \ge 2$. For all-pairs single-failure min-cut, Baswana and
Pandey~\cite{BaswanaP22} obtained an $O(n^2)$-space, $O(1)$-query oracle.
These combinatorial approaches exploit the structure of near-minimum cuts
and flow decompositions, and appear inherently limited to one or two
failures. Our framework resolves the open problem of~\cite{BaswanaBP22},
providing the first $(s,t)$-min-cut sensitivity oracle for any $f$, with
$O((fn)^2 \log n)$ space, $\widetilde{O}(f^{3\omega})$ query time for the
value, and $\widetilde{O}(f^{3\omega} n)$ time for reporting the nearest
min-cut.

\vspace{4mm}
\noindent
\emph{Fault-Tolerant Reachability and Distances.~}
Algebraic matrix tools have become a powerhouse for dynamic and
fault-tolerant reachability and distances.
Sankowski~\cite{Sankowski04} introduced dynamic reachability
algorithms based on maintaining the inverse of a symbolic adjacency matrix.
For weighted graphs, Weimann and Yuster~\cite{WeimannY13} employed matrix
adjoint computation for distance sensitivity, constructing an $f$-DSO with
$\widetilde{O}(n^{3-\alpha})$ space and $\widetilde{O}(n^{2-2(1-\alpha)/f})$
query time. Brand and Saranurak~\cite{BrandS19} substantially improved the
preprocessing and query time for multiple failures, obtaining
$\widetilde{O}(Wn^3)$ preprocessing and $\widetilde{O}(Wnf^\omega)$ query
time. Karczmarz and
Sankowski~\cite{KarczmarzS23} recently improved multiple-failure DSOs via
generic matrices and the Frobenius normal form. On the combinatorial side, \cite{DTCR08,BernsteinK09,DuanP:10,ChechikLPR:12,Bilo16-stacs,Choudhary16,BaswanaCR:17,ChechikC:17,DuanP:17,GIP17,GeorgiadisDIKP17,GuptaS18,BiloCFS21,ItalianoKP21,DuanR22,BiloCCCFKS23,Dey024}
studied the problem of designing fault-tolerant oracles for reachability, distances, and strong connectivity.

\vspace{4mm}
\noindent
\emph{Dynamic Algorithms.~}
In the fully dynamic regime, min-cut algorithms range from Thorup's
tree-packing approach~\cite{Thorup07} to recent expander-decomposition
methods~\cite{GoranciHNSTW23, JinST24}, sub-polynomial update times for
all-pairs min-cuts~\cite{JinS22}, and dynamic flow
advances~\cite{BrandCKLMGS24, BrandCKLPPSS24, ChenKLMP24, Karczmarz24}. 
Several recent works study dynamic graph algorithms with predictions~\cite{PengR23,
HuKP24, HenzingerSSY24, BrandFNP24, LiuS24, McCauleyMNS24, McCauleyMNNS25},
a model related in spirit to our subset sensitivity setting.

\subsection{Organization}
The remainder of this paper is organized as follows. We establish notation and algebraic primitives in Section~\ref{sec:preliminaries}. We next present our packing and covering oracles in Section~\ref{section:packing}. These primitives are applied to flows and cuts in Section~\ref{sec:flows}, and to spanning structures in Section~\ref{sec:trees}. We extend the framework to linear matroid parity and $\alpha$-factors in Sections~\ref{sec:matroid_parity} and~\ref{sec:alpha-factors} respectively. Our subset sensitivity results are presented in  Section~\ref{sec:subset}, and matching space lower bounds are presented in Section~\ref{sec:lb}.

\section{Preliminaries}
\label{sec:preliminaries}
Let $G = (V, E)$ be an uncapacitated directed graph with a cost function $\cost : E \to \{-W,\dots,W\}$. A \emph{trail} $\trail$ from edge $e_0$ to edge $e_\ell$ in $G$ is a sequence
of distinct edges $(e_0, e_1, \dots, e_\ell)$ with
$\head(e_i) = \tail(e_{i+1})$ for each $0 \le i < \ell$ (vertices may
repeat). The cost of a trail is the sum of its edge costs. A trail is a
\emph{path} if no vertex repeats.

For a subset of $U = (I, F)$ of edge updates consisting of a set $I$ of edge insertions and a set $F \subseteq E$ of edge deletions, with $|U| = |I| + |F| \le f$, we denote by $G+U := (V, (E \setminus F) \cup I)$ the graph obtained by applying the updates in $U$ to $G$. 

For a pair $s,t \in V$, we use $\lambda_G(s,t)$ to denote the maximum value of an $(s,t)$-flow in $G$. An \emph{$(s,t)$-cut} $C$ partitions the vertex set into $(A_C, B_C)$ with $s \in A_C$ and $t \in B_C$; its capacity is the number of edges directed from $A_C$ to $B_C$, and an $(s,t)$-min-cut is one with minimum capacity, of value $\lambda_G(s,t)$ by the max-flow min-cut theorem. For a fixed source $s \in V$, the \emph{$s$-min-cut} value is $\mu_G(s) := \min_{v \in V \setminus \{s\}} \lambda_G(s,v)$, and the \emph{global min-cut} value is $\mu_G := \min_{s \in V} \mu_G(s)$. Given an $(s,t)$-flow $h$ in $G$, the \emph{residual graph} $G_h$ is obtained from $G$ by reversing every edge that carries flow under $h$, i.e., replacing each edge $e = (u,v)$ with $h(e) = 1$ by its reverse $e^{\rev} = (v,u)$, where $\cost(e^{\rev}) = -\cost(e)$. We use $\val(h)$ to denote the value of the flow $h$.

An $(s,t)$-min-cut $C^*$ is said to be the {\em nearest min-cut} if for each $(s,t)$-min-cut $C$ we have $A_{C^*}\subseteq A_C$. 
We denote such a $C^*$ with $\NMC_G(s,t)$.
Ford and Fulkerson~\cite{FulkersonF:62} gave an algorithm for constructing the (unique) nearest $(s,t)$-min-cut. For the sake of completeness we state the following characterization of NMCs from~\cite{FulkersonF:62}.

\begin{lemma}
\label{lemma:nearest-mincut}
Let $C=\NMC_G(s,t)$ with partition $(A,B)$. Then $A$ consists precisely of those
$w\in V\setminus\{t\}$ for which $G_w:=G+(w,t)$ satisfies
$\lambda_{G_w}(s,t)=\lambda_{G}(s,t)+1$.
\end{lemma}

\paragraph{Algebraic Primitives}
We now state the algebraic tools that we use to compute determinants of low-degree polynomial matrices in (near) matrix-multiplication time. We shall use $\F$ to denote a field. Recall  that $\F[Y]$ is  used to denote the ring of polynomials in $Y$ with coefficients from $\F$. Given a non-negative integer $d$, let $\F_d[Y]$ be the subset of polynomials in $\F[Y]$ of degree at most $d$. Finally, we shall use  $\F_d[Y]^{N \times N}$ to denote the set of $N \times N$ matrices whose entries belong to $\F_d[Y]$. 

\begin{lemma}
\label{lemma:det-adj-row}
For any matrix $M \in \F_d[Y]^{N \times N}$, one can compute $\det(M)$ in $\widetilde{O}(d\,N^{\omega})$ field operations.
\end{lemma}

\begin{lemma}[Matrix Determinant Lemma]
\label{lemma:matrix_det}
Let $M \in \F_d^{N \times N}$ be an invertible matrix and let $U, V \in \F_d^{N \times k}$. Then
\[
    \det(M + U V^T) = \det(M) \cdot \det(I_k + V^T M^{-1} U).
\]
\end{lemma}

Our results also employ the  oracle by Brand and Saranurak~\cite{BrandS19} that preprocesses a matrix  $H$ to efficiently answer determinant queries following a limited number of updates to  $H$. 

\begin{lemma}[{\cite[Corollary of Theorem 4.1]{BrandS19}}]
\label{lemma:BrandS-det}
Let $M \in \F_d[Y]^{N \times N}$ be a matrix such that $\det(M) \neq 0$. Then, there is an algorithm that preprocesses $M$ using $\widetilde{O}(dN^{3})$ operations and creates a data structure occupying $O(d N^{3} \log N)$ space with the following guarantee:  given $C \in \F_d[Y]^{N \times N}$ with at most $\delta$ nonzero entries and satisfying $\det(M+C) \neq 0$, one can query $\det(M+C)$ in
$\widetilde{O}(dN \delta^{\omega})$ 
operations.

Moreover, if $d=0$, then the preprocessing, space and query time are $O(N^\omega), \ O(N^2 \log N)$ and $O(\delta^\omega)$ respectively.
\end{lemma}

While the above lemma assumes a non-singular base matrix, we show in Appendix that this result can be naturally extended to general, possibly singular matrices with same asymptotic complexity bounds.

We next state two determinant identities that will be used 
to expand determinants of products of rectangular matrices and to 
decompose determinants of block-structured representation matrices.

\begin{lemma}[Cauchy-Binet]
\label{lemma:cauchy-binet}
Let $A$ and $B$ be two $r \times m$ matrices 
with $m \geq r$. Then
$$
\det(AB^T) \;=\; \sum_{\substack{S \subseteq [m] \\ |S| = r}} 
\det(A_S)\,\det(B_S),
$$
where the sum ranges over all $\binom{m}{r}$ subsets $S$ of 
size $r$, and $A_S$, $B_S$ denote the corresponding 
$r \times r$ submatrices.
\end{lemma}

\begin{lemma}[Generalized Laplace expansion]
\label{lemma:generalised-laplace-expansion}
Let $A$ be an $n \times n$ matrix and let its rows be partitioned 
into $k$ contiguous blocks $R_1, \dots, R_k$ of sizes 
$r_1, \dots, r_k$ with $\sum_i r_i = n$. Then
$$
\det(A) \;=\; \sum_{\substack{S_1 \uplus \dots \uplus S_k = [n] \\ 
|S_i| = r_i}} 
(-1)^{\sigma(S_1,\dots,S_k)}\, 
\prod_{i=1}^k \det(A_{R_i, S_i}),
$$
where the sum is taken over all ordered partitions of the columns 
$[n]$ into blocks $S_1 \uplus \dots \uplus S_k$ with 
$|S_i| = r_i$, and $(-1)^{\sigma(S_1,\dots,S_k)}$ is the sign 
of the permutation that reorders $[n]$ into $(S_1, \dots, S_k)$.
\end{lemma}

We finally state the Schwartz-Zippel lemma which will be useful in providing high probability success bounds for our algorithms.

\begin{lemma}[Schwartz-Zippel lemma]
Let $P(x_1, x_2, \dots, x_N)$ be a non-zero polynomial of degree $d$ over a finite field $\F$. If a point is selected uniformly at random from $\F^N$, then the probability that $P$ evaluates to a non-zero value is at least $1-d/{|\F|}$.
\label{lemma:SZ}
\end{lemma}

In this paper, we shall work with finite fields. Thus $|\F|$ will be $p^r$, for some prime $p$ and integer $r\geqslant 1$. We will be working with fields representable in $O(\log n)$ bits. This would support addition and multiplication over $\F$ in constant time in the Word-RAM model.

\paragraph{Matroid Terminology}
We next briefly review standard matroid terminology. 

\begin{definition}[Matroid]
A \emph{matroid} $M = (E, \I)$ consists of a finite ground set $E$ and a collection $\I$ of subsets of $E$, called \emph{independent sets}, satisfying three properties:
\begin{enumerate}
\item 
$\emptyset \in \I$; 
\item if $I \in \I$ and $I' \subseteq I$, then $I' \in \I$; and 
\item if $I_1, I_2 \in \I$ with $|I_1| < |I_2|$, there exists an element $e \in I_2 \setminus I_1$ such that $I_1 \cup \{e\} \in \I$. 
\end{enumerate}
\end{definition}

\begin{definition}[Basis]
A \emph{basis} of $M$ is a maximal independent set. It is well known that all bases of $M$ have the same cardinality, known as the \emph{rank} of the matroid, denoted $r(M)$ (or simply $r$ if $M$ is clear from the context). We denote the set of all bases of $M$ by $\B(M)$.
\end{definition}

\begin{definition}[Linear Matroid]
A matroid $M = (E, \I)$  on a ground set $E$ of size $m$ is said to be \emph{linear} if there exists a $d \times m$ matrix $A$ over $\F$ such that $S \subseteq E$ is independent iff the sub-matrix $A_S$ (columns of $A$ indexed by $S$) has full column rank in $\F$. If the rank of $M$ is $r$, we can assume without loss of generality that $d = r$. 
\end{definition}

Note that for any linear matroid $M$ of size $m$ and rank $r$ it is possible to produce a linear representation matrix $A$ of size $r\times m$. Clearly, for such a representation $A$, any basis $S \in \B(M)$ satisfies $\det(A_S) \neq 0$.

\begin{definition}[Sparse Linear Matroid]
\label{definition:sparse-matroid}
A linear matroid $M$  is said to be $c$-\emph{sparse} (or $c$-\emph{column-sparse}) if there exists a representation matrix $A$ for $M$ over a field $\F$ such that every column of $A$ contains at most $c$ non-zero entries. For convenience, we refer to the matrix as just \emph{sparse} if $c = O(1)$.
\end{definition}

Matroids arising from graphs (such as graphic or partition matroids) naturally admit sparse representations, as their elements map to edges with constant-sized incidence bounds.
We next define the notion of strongly base-orderable matroids.

\begin{definition}[Strongly base-orderable]
\label{definition:sbo}
A matroid $M$ on $E$ is \emph{strongly base-orderable} (SBO) if, for any two bases $B_1$ and $B_2$ of $M$, there exists a bijection $\pi \colon B_1 \to B_2$ such that for every $X \subseteq B_1$, both
\[
(B_1 \setminus X) \cup \pi(X) \quad \text{and} \quad (B_2 \setminus \pi(X)) \cup X
\]
are bases of $M$.
\end{definition}

Important families of SBO matroids include  partition matroids, transversal matroids, gammoids, graphic matroids on $K_4$-minor-free graphs, etc. 

While packing common basis of two arbitrary matroids is an NP-Hard problem under the rank oracle model (NP-Hard even for two arbitrary linear matroids with given representation)~\cite{BercziS21}, the following theorem by Davies and McDiarmid provides a nice characterisation for packing common basis in SBO matroids. 

\begin{theorem}[Davies-McDiarmid~\cite{DaviesMcDiarmid1976}]
\label{theorem:SBO-ind}
Let $M_1$ and $M_2$ be strongly base-orderable matroids on the same ground set $E$.
Then for any $S\subseteq E$,  $S$ can be partitioned into $k$ common independent sets of $M_1$ and $M_2$ if and only if $S$ can be partitioned into $k$ independent sets of $M_1$ and into $k$ independent sets of $M_2$.
\end{theorem}

\section{Matroid Packing and Matroid Covering}
\label{section:packing}

In this section, we present an algebraic framework for addressing {\em matroid intersection packing} and {\em matroid covering} problems in the general sensitivity setting as well as \emph{subset} sensitivity setting.

\subsection{Sensitivity Oracle for Matroid Intersection Packing}
Let $M_1 = (E, \I_1)$ and $M_2 = (E, \I_2)$ be two linear matroids of rank $r$ defined on a common ground set $E$ of size $m$, and let $\wt: E \to [-W, W]$ be a weight function.
Given an integer $k$, the \emph{decision version} of the $k$-packing problem asks whether there exists a subset $S$ of $E$ such that $S$ can be partitioned into $k$ pairwise disjoint common bases of $M_1$ as well as $M_2$ (observe that the partitioning of $S$ may differ in the two cases). The \emph{weighted version} asks for the  minimum total weight of such a $k$-packing-set $S$, provided one exists. We approach this problem by analyzing the intersection of the $k$-fold unions of $M_1$ and $M_2$.

\paragraph{$k$-fold Matroid Union}
Let $M$ be a linear matroid defined on a ground set $E$ of size $m$. Suppose $M$ has rank $r$ and is represented by a $r \times m$ matrix $A$ over a finite field $\F$. The $k$-fold union of $M$, denoted $M^{\vee k}$, is a matroid whose independent sets are those subsets $S \subseteq E$ that can be partitioned into $k$ disjoint sets $S_1 \uplus S_2 \uplus \dots \uplus S_k$ such that each $S_i$ is an independent set in $M$. Note that the rank of $M^{\vee k}$ may be much smaller than $kr$.

\begin{definition}[Randomized Union Representation]
\label{def:union_rep}
Let $M^{\vee k}$ be the $k$-fold union of a matroid $M = (E, \I)$ of rank $r$ on $m$ elements. A random matrix $U \in \F^{kr \times m}$ is a randomized union representation of $M^{\vee k}$ if it satisfies the following two properties:
\begin{enumerate}
\item
If a set $S \subseteq E$ is independent in $M^{\vee k}$, then $U_S$ has full column rank with probability at least $1 - \frac{1}{(kr)^{\Omega(1)}}$.
\item
If a set $S \subseteq E$ is dependent in $M^{\vee k}$, then $U_S$ does not have full column rank (this holds deterministically).
\end{enumerate}
\end{definition}

The following lemma shows how to compute a randomized union representation of $M^{\vee k}$ from a linear representation of $M$. For the sake of completeness, we provide a proof in the appendix.

\begin{lemma}[\cite{NarayananSV94}]
\label{lemma:union_rep}
Let $M$ be a linear matroid of rank $r$ on $m$ elements, represented by an $r \times m$ matrix $A$ over a finite field $\F$ with $|\F| \ge \poly(kr)$. For each $i \in [k]$, let $C_i$ be an $m \times m$ diagonal matrix whose diagonal entries are independent, uniformly random elements from $\F$. Then the $rk \times m$ matrix
$$U := \begin{pmatrix} A C_1 \\ A C_2 \\ \vdots \\ A C_k \end{pmatrix}$$
is a randomized union representation of $M^{\vee k}$.
\end{lemma}

\begin{remark}[Parallel Elements in $k$-Fold Union]
\label{rem:parallel-elements-k-fold-union}
Observe that it is possible that an independent set in $M^{\vee k}$ contains $k$ identical copies of an element, i.e.,  the columns of $A$ corresponding to this independent set may contain $k$ identical columns. This is consistent with the definition of $M^{\vee k}$,
since an independent set $S \subseteq E$ in $M^{\vee k}$ is partitioned into $k$ blocks $S_1 \uplus \cdots \uplus S_k$, each being an independent set in $M$; thus each such block may contain at most one copy of such an element. The representation $U$ remains valid, as the diagonal scalings $C_1, \dots, C_k$ assign independent random multipliers to each (element, copy) pair.
\end{remark}

\paragraph{Intersection of two $k$-Fold Matroid Unions}
Given two linear matroids $M_1$ and $M_2$ of rank $r$ on a common ground set $E$,
we consider the problem of identifying existence of a common basis in $M_1^{\vee k}$ and $M_2^{\vee k}$ of size $kr$.
Let $U_1$ and $U_2$ be the $kr \times m$ union representation matrices for $M_1^{\vee k}$ and $M_2^{\vee k}$, respectively, constructed using independent random diagonal matrices as given in \Cref{lemma:union_rep}.
With each element $e \in E$,
we associate  a value $x_e$ drawn uniformly from a large  enough finite field $\F$. Define a diagonal matrix $D\in \F[Y]^{m\times m}$ with an indeterminate $Y$ as follows:
$$
D_{e,e} := x_e Y^{W + \wt(e)} \quad \text{for all } e \in E,
$$
and $D_{e,e'} = 0$ for $e \neq e'$.

We next define the \emph{Mixed Laplacian} $\L_\cap$ of size $kr \times kr$ as:
$$
\L_\cap := U_1 D U_2^T.
$$

The following theorem extends the classical Cauchy-Binet approach for matroid intersection~\cite{Harvey09, CheungLL14} to the $k$-fold union setting.

\begin{theorem}[Enumerating Common Bases]
\label{thm:k-fold-matroid_intersection}
Suppose both $M_1^{\vee k}$ and $M_2^{\vee k}$ have rank $kr$, and $|\F| \ge \poly(kr)$. Then
$$
\det(\L_\cap) =
Y^{krW}
\sum_{\substack{S \in \B(M_1^{\vee k}) \cap \B(M_2^{\vee k})}}
\alpha_S  \Bigl( \prod_{e \in S} x_e \Bigr)
Y^{\sum_{e \in S} \wt(e)},
$$
where $\alpha_S := \det\bigl((U_1)_S\bigr)\det\bigl((U_2)_S\bigr)$ for any common basis $S$ is non-zero with probability $1 - \frac{1}{(kr)^{\Omega(1)}}$.
Furthermore, with the same probability bound, the minimum degree of $Y$ in $\det(\L_\cap)$ minus $krW$ equals the weight of a minimum-weight common basis.
\end{theorem}

\begin{proof}
We evaluate the determinant of the Matroid Laplacian $\L_\cap = U_1 D U_2^T$ using the Cauchy-Binet formula (see \Cref{lemma:cauchy-binet}) and substitute $\det(D_S) = \prod_{e \in S} x_e \, Y^{W + \wt(e)}$ and $\alpha_S$ to get:
$$
\det(\L_\cap) = \det\bigl(U_1 (D U_2^T)\bigr)
= \sum_{\substack{S \subseteq E \\ |S|=kr}} \det\bigl((U_1)_S\bigr) \det\bigl((U_2)_S\bigr) \det(D_S)
= Y^{krW} \sum_{\substack{S \subseteq E \\ |S|=kr}} \alpha_S \Bigl(\prod_{e \in S} x_e\Bigr) Y^{\wt(S)}.
$$
By \Cref{lemma:union_rep}, if $S$ is dependent in either $M_j^{\vee k}$, then $(U_j)_S$ does not have full column rank, and so $\alpha_S = 0$. This  restricts the sum to common bases. Next, if $S$ is a common basis, then by the probabilistic claim in \Cref{lemma:union_rep}, each $(U_j)_S$ has full column rank with probability at least $1 - 1/(kr)^{\Omega(1)}$, and so $\alpha_S \neq 0$ with at least this probability.

For the final claim, let $\wt_{\min}$ denote the minimum weight of a common basis of $M_1^{\vee k}$ and $M_2^{\vee k}$. Then, the coefficient of $Y^{krW + \wt_{\min}}$
is $\sum_{S^*} \alpha_{S^*} \prod_{e \in S^*} x_e$, where the sum is over minimum-weight common bases.
The above is a nonzero polynomial of degree $kr$ since distinct bases contribute distinct  monomials in $\{x_e\}$, and with high probability at least one minimum-weight common basis $S^*$ satisfies $\alpha_{S^*} \neq 0$. By the Schwartz-Zippel lemma, it evaluates to a nonzero value with probability $1 - kr/|\F|$.
Therefore, we get that with probability $1 - \frac{1}{(kr)^{\Omega(1)}}$, isolating the smallest exponent of $Y$ in $\det(\L_\cap)$ and subtracting $krW$ yields $\wt_{\min}$.
\end{proof}

\subsubsection{Sensitivity Oracle for Sparse Linear Matroids}
\label{subsubsec:gen-packing}
We now consider the sensitivity problem for the $k$-packing problem in the setting of sparse linear matroids.  To obtain $k$ disjoint common bases of two matroids $M_1$ and $M_2$, we find a common basis of full rank $kr$ in the $k$-fold union of these matroids. The main technical tool for this proceeds by analyzing the determinant of the mixed Laplacian using \Cref{thm:k-fold-matroid_intersection}. In the static setting, this determinant evaluates to a non-zero polynomial with high probability if and only if a valid $k$-packing exists, and its minimum degree extracts the minimum weight of such a packing.

We formalize this in the sensitivity setting by utilizing the data structure from \Cref{lemma:BrandS-det-General-Q} to handle sparse updates to the base representations of the matroids.

\begin{theorem}
\label{thm:sensitivity_packing}
Let $M_1 = (E, \I_1)$ and $M_2 = (E, \I_2)$ be two sparse linear matroids of rank $r$ defined on a common ground set $E$ of size $m$, with a weight function $\wt: E \to [-W, W]$. Let $A_1$ and~$A_2$ be their respective $r \times m$  linear representation matrices over a field $\F$.

Then, for any $k\geq 1$, there exist  sensitivity oracles that, given a set of $\delta$ updates to $A_1, A_2$, and $\wt$, answer $k$-packing queries for the updated matroids $M'_1$ and $M'_2$ as follows:
\begin{enumerate}
\item \textbf{Decision Oracle:}
Decides if a $k$-packing set $S$ exists. Recall that a subset $S \subseteq E, |S| = kr$ is  said to be $k$-packing if it can be partitioned into $k$ common bases of $M'_1$ and $M'_2$ respectively, where the partitioning corresponding to each of the matroids may differ. This oracle requires $O(k^2m+(kr)^\omega)$ preprocessing time, $O(km+(kr)^2 \log(kr))$ space, and answers queries in $O((k^2\delta)^\omega)$ time.
\item \textbf{Weighted Oracle:} Returns the exact total weight of a minimum-weight $k$-packing set $S$. This oracle requires $\widetilde{O}(k^2m+W (kr)^3)$ preprocessing time, $O(km+W (kr)^3 \log(kr))$ space, and answers queries in $\widetilde{O}(W \cdot kr \cdot (k^2\delta)^\omega)$ time.
\end{enumerate}
\end{theorem}

\begin{proof}
Let $U_1$ and $U_2$ be the $kr \times m$ $k$-fold randomized union representation matrices for $M_1$ and $M_2$ respectively, as obtained by \Cref{lemma:union_rep}.
Let $D$ be the $m \times m$ diagonal weight matrix with $D_{e,e} = x_e Y^{W+\wt(e)}$, and define the $kr \times kr$ mixed Laplacian matrix as $\L_{pack} := U_1 D U_2^T$.
Suppose the base representations $A_1, A_2$ and the weight function undergo at most $\delta$ changes. We first bound the number of non-zero entries in the resulting additive update to the mixed Laplacian. The matrix $\L_{pack}$ can be expressed as:
$$
\L_{pack} =
\sum_{e\in E} D_{e,e} (u_1)_e (u_2)_e^T
$$
where $(u_1)_e$ and $(u_2)_e$ are the $e$-th columns of $U_1$ and $U_2$, respectively.
Computing the matrix $\L_{pack}$ takes $O(k^2 m)$ time, assuming that representation matrices $A_1,A_2$ are provided in sparse column format (see \Cref{definition:sparse-matroid}).
Note that the updates to the $e$-th column of $A_1, A_2$, or to $\wt(e)$ modify only the $e$-th term in the sum, and since $A_1$ and $A_2$ are column-sparse, the product $(u_1)_e (u_2)_e^T$ has at most $O(k^2)$ non-zero entries.
Thus, for any $\delta$ base changes, the updated Laplacian can be represented as $\L'_{pack} = \L_{pack} + C$, where matrix $C$ has at most $O(k^2\delta)$ non-zero entries. By \Cref{thm:k-fold-matroid_intersection}, identifying the existence or computing the minimum weight of a $k$-packing set $S$ in the updated matroids reduces to evaluating the determinant of $\L'_{pack}$.

To efficiently support these queries, we construct two separate oracles based on the data structure from \Cref{lemma:BrandS-det-General-Q}:

\paragraph{Decision Oracle}
We ignore the polynomial weights and work with the Laplacian $\L_{pack}|_{Y=1}$ (where the maximum degree is $d=0$). \Cref{lemma:BrandS-det-General-Q} together with the fact that $M_1,M_2$ are column sparse implies that preprocessing this matrix requires $O(k^2m+(kr)^\omega)$ time and the corresponding data structure takes $O(km+(kr)^2 \log(kr))$ space. 
The additive $O(km)$ term arises from the fact that  storing the union representations $U_1$ and $U_2$ requires $O(km)$ space.
After at most $\delta$ updates to representation matrices $A_1,A_2$, and $\wt$, querying the determinant of the Laplacian $\L'_{pack}|_{Y=1} ~=~\L_{pack}|_{Y=1}+C|_{Y=1}$ requires $O((k^2\delta)^\omega)$ query time.

\paragraph{Weighted Oracle}
The matrix $\L_{pack} \in \F_{2W}[Y]^{kr \times kr}$ is a polynomial matrix where the maximum degree of $Y$ in any entry is bounded by $d = 2W$. \Cref{lemma:BrandS-det-General-Q} together with the fact that $M_1,M_2$ are column sparse implies that preprocessing this matrix requires $\widetilde{O}(k^2m+W (kr)^3)$ time and the oracle takes $O(km+W (kr)^3 \log(kr))$ space. To compute the weight after any $\delta$ updates to $A_1,A_2,\wt$, we evaluate the determinant of the polynomial matrix $\L'_{pack}$ and extract the coefficient of the minimum-degree term of $Y$. Querying the determinant of this updated matrix takes $\widetilde{O}(W \cdot kr \cdot (k^2\delta)^\omega)$ time, and isolating the smallest exponent of $Y$ in $\det(\L'_{pack})$ and subtracting $krW$ yields the exact weight of the minimum-weight $k$-packing set. This completes the proof.
\end{proof}

\begin{remark}[Cost of Parallel Updates]
\label{rem:parallel-edges-update-packing-laplacian}
The query complexity in \Cref{thm:sensitivity_packing} can be refined
when some of the updates correspond to changes in the
multiplicity of parallel elements. Suppose the updates affect distinct
elements $e_1, \dots, e_{\rho}$ (with $\rho \le \delta$), where the
multiplicity of $e_i$ increases or decreases by $\rho_i \le k$, so that
$\sum_{i=1}^{\rho} \rho_i = \delta$. Column-sparsity of $A_1, A_2$
ensures that each $(e_i, \rho_i)$ contributes only $O(k^2)$ non-zero
entries within an $O(k) \times O(k)$ submatrix of the additive update
$C$ to the packing Laplacian
$\L_{\mathrm{pack}} = U_1 D U_2^{T}$. Aggregating over all
$\rho$ affected elements, this yields query time
$O((k^2 \rho)^{\omega})$ rather than $O((k^2 \delta)^{\omega})$.
\end{remark}

\subsubsection{Subset Sensitivity Oracle for General Linear Matroids}
The oracle in \Cref{thm:sensitivity_packing} handles arbitrary updates
to the ground set, incurring a query cost of
$\Omega((k^2\delta)^\omega)$. We now present a complementary oracle for
the \emph{subset sensitivity model}: a \emph{susceptible set}
$E_Q$ of size $\sigma$ is identified at preprocessing time, and all
future queries modify only elements in $E_Q$. The set $E_Q$ need not be
a subset of $E$; we require only that the matroids $M_1$ and $M_2$
admit well-defined extensions to the augmented ground set
$E_0 := E \cup E_Q$.

\begin{theorem}
\label{thm:subset_sensitivity_packing}
Let $M_1 = (E, \I_1)$ and $M_2 = (E, \I_2)$ be two
linear matroids each of rank $r$  on a ground set $E$ of size $m$, with
weight function $\wt : E \to [-W, W]$. Let $E_Q$ be a susceptible set
of size $\sigma$, and suppose $M_1$ and $M_2$ extend to linear matroids of (same) rank $r$
on
$E_0 := E \cup E_Q$ with $r \times |E_0|$ representation matrices
$\bar{A}_1$ and $\bar{A}_2$ over $\F$. Then, for any $k \ge 1$, there
exist the following subset sensitivity oracles that, for any $f$ updates
confined to $E_Q$, answer $k$-packing queries:
\begin{enumerate}
\item \textbf{Decision Oracle:} Decides if a $k$-packing exists. This
oracle requires
$O(k^2(m + \sigma) + \sigma^2 (kr)^{\omega - 2} + (kr)^\omega)$
preprocessing time, $O(\sigma^2)$ space, and answers queries in
$O(f^\omega)$ time.

\item \textbf{Weighted Oracle:} Returns the exact minimum weight of a
$k$-packing. This oracle requires
$\widetilde{O}(k^2(m + \sigma) + W \sigma^2 (kr)^{\omega - 1} + W(kr)^{\omega+1})$
preprocessing time, $\widetilde{O}(W \cdot kr \cdot \sigma^2)$ space,
and answers queries in $\widetilde{O}(W \cdot kr \cdot f^\omega)$ time.
\end{enumerate}
\end{theorem}

\begin{proof}
Let $\bar{U}_1, \bar{U}_2 \in \F^{kr \times |E_0|}$ be the $k$-fold
union representation matrices for the extensions of $M_1$ and $M_2$ to
$E_0$, constructed via \Cref{lemma:union_rep} from $\bar{A}_1$ and
$\bar{A}_2$. Let $D \in \F_{2W}[Y]^{|E_0| \times |E_0|}$ be the
diagonal weight matrix with $D_{e,e} = x_e Y^{W + \wt(e)}$ for
$e \in E_0$, so that the Packing Laplacian is
$\L_{\mathrm{pack}} = \bar{U}_1 D \bar{U}_2^T$. Computing
$\L_{\mathrm{pack}}$ takes $O(k^2(m + \sigma))$ time, since
$|E_0| = m + \sigma$ and each element contributes $O(k^2)$ entries to
the Laplacian.

Since all queries modify only diagonal entries indexed by $E_Q$, we
invoke \Cref{lemma:det} with $A := \bar{U}_1$, $B := \bar{U}_2$, the
diagonal matrix $D$, and susceptible set $Q := E_Q$ of size $\sigma$.

\paragraph{Decision Oracle.}
We apply the scalar formulation of \Cref{lemma:det} with $d := 0$,
working with $\L_{\mathrm{pack}}|_{Y=1}$. Substituting
$r := kr$ and $\sigma = |E_Q|$ into \Cref{lemma:det}, the data
structure construction takes
$O(\sigma^2 (kr)^{\omega - 2} + (kr)^\omega)$ time and $O(\sigma^2)$
space. Combined with the $O(k^2(m + \sigma))$ cost of computing
$\L_{\mathrm{pack}}$, the total preprocessing time is
$O(k^2(m + \sigma) + \sigma^2 (kr)^{\omega - 2} + (kr)^\omega)$. Each
query takes $O(f^\omega)$ time. By
\Cref{thm:k-fold-matroid_intersection}, the determinant is non-zero if
and only if a valid $k$-packing exists.

\paragraph{Weighted Oracle.}
We apply the polynomial formulation of \Cref{lemma:det} with $d := 2W$.
Substituting $r := kr$ into \Cref{lemma:det}, the data structure
construction takes
$\widetilde{O}(W \sigma^2 (kr)^{\omega - 1} + W(kr)^{\omega+1})$ time
and $\widetilde{O}(W \cdot kr \cdot \sigma^2)$ space. Combined with the
Laplacian computation, the total preprocessing time is
$\widetilde{O}(k^2(m + \sigma) + W \sigma^2 (kr)^{\omega - 1} + W(kr)^{\omega+1})$.
Each query takes $\widetilde{O}(W \cdot kr \cdot f^\omega)$ time. The
minimum weight is extracted as the smallest exponent of $Y$ with a
non-zero coefficient, minus $krW$.
\end{proof}

\begin{remark}[Comparison of Oracle Models]
\label{rem:oracle-comparison}
The general and subset sensitivity oracles for matroid packing offer
complementary trade-offs:
\begin{center}
\renewcommand{\arraystretch}{1.25}
\begin{tabular}{lcc}
\toprule
& \textbf{General Sensitivity} & \textbf{Subset Sensitivity} \\
& (\Cref{thm:sensitivity_packing}) & (\Cref{thm:subset_sensitivity_packing}) \\
\midrule
Space (decision) & $O((kr)^2 \log k)$ & $O(\sigma^2)$ \\
Query (decision) & $O((k^2 \delta)^\omega)$ & $O(f^\omega)$ \\
\bottomrule
\end{tabular}
\end{center}
The subset sensitivity oracle achieves query time independent of $k$,
at the cost of space scaling with the susceptible set size $\sigma$
rather than the matroid rank. When $\sigma \ll m$ and $k$ is large
relative to $f$, the subset sensitivity model provides both smaller
space and faster queries. When $\sigma$ is comparable to $m$, the
general sensitivity oracle is more space-efficient.
\end{remark}

\subsection{Sensitivity Oracle for Matroid Covering}
\label{section:covering}

We next present an algebraic framework for addressing {\em matroid covering}.

Let $M = (E, \I)$ be a linear matroid of rank $r$ represented by an $r\times m$ matrix $A$ over a field $\F$. The \emph{$k$-covering problem} asks whether the entire ground set $E$ can be covered by (partitioned into) $k$ independent sets of $M$. This is equivalent to asking whether the entire ground set $E$ is an independent set in the $k$-fold union matroid $M^{\vee k}$.

For $E$ to be an independent set in $M^{\vee k}$, its cardinality $m$ must be at most $kr$. Let $U \in \F^{kr \times m}$ be the union representation matrix obtained via \Cref{lemma:union_rep}. The set $E$ is a $k$-covering if and only if the columns of $U$ are linearly independent. To test this, we augment $U$ into a square matrix by adding random columns.

\begin{lemma}[Augmenting Linear Representations]
\label{lemma:aug-linear-rep}
Let $U \in \F^{kr\times m}$ be the $k$-fold union representation matrix for $M^{\vee k}$. Suppose $m \le kr$. Let $\Delta=kr-m$, and $R \in \F^{kr\times \Delta}$ be a random matrix whose entries are chosen independently and uniformly from a sufficiently large finite field $\F$.

Then, with high probability, the augmented square matrix $U_{aug} := [U \mid R] \in \F^{kr \times kr}$ has full rank $kr$ if and only if the row rank of $U$ is $m$, i.e., $E$ is $k$-coverable.
\end{lemma}

\begin{proof}
We treat the entries of the random matrix $R$ as indeterminates and work over the polynomial ring $\F[\X]$, where
$$
\X := \{R_{i,j} : 1 \le i \le kr,\ 1 \le j \le \Delta\}.
$$
Using the block decomposition, the determinant of the $kr \times kr$ augmented matrix $U_{aug} = [U \mid R]$ can be expressed as:
$$
\det(U_{aug}) = \sum_{\substack{S \subseteq [kr] \\ |S| = m}} (-1)^{\sigma(S)} \det(U_S) \det(R_{S^c}),
$$
where $U_S$ denotes the $m \times m$ submatrix of $U$ induced by the rows in $S$, $R_{S^c}$ denotes the $\Delta \times \Delta$ submatrix of $R$ induced by the rows in $S^c$, and $\sigma(S)$ is the corresponding sign parity.

If the row rank of $U$ is less than $m$, then every $m \times m$ minor $\det(U_S)$ evaluates to~$0$, making $\det(U_{aug})$ identically the zero polynomial.

Let us suppose the row rank of $U$ is exactly $m$. Then there exists at least one subset of rows $S \subseteq [kr]$ such that the minor $\det(U_S)$ is a non-zero constant in $\F$. For each such $S$, the corresponding determinant $\det(R_{S^c})$ contributes distinct monomials in the variables $\X$. Moreover, the monomials generated by one row subset $S$ cannot be cancelled out by the monomials generated by any other subset $S' \neq S$. Thus, $\det(U_{aug})$ is a non-zero polynomial.

As entries of $R$ are chosen independently and uniformly at random from a sufficiently large finite field, by the Schwartz-Zippel lemma $\det(U_{aug})$ evaluates to a non-zero value with high probability. This proves that with high probability, the matrix $U_{aug}$ has rank $kr$ if and only if $U$ has rank $m$.
\end{proof}

\subsubsection{Sensitivity Oracle for Sparse Linear Matroids}
\label{subsubsec:gen-covering}
Building on \Cref{lemma:aug-linear-rep}, we now present the sensitivity oracle for the $k$-covering problem.

\begin{theorem}
\label{thm:sensitivity_covering}
Let $k\geq 1$ be an integer.
Let $M = (E, \I)$ be a linear matroid of rank $r$ defined on a ground set $E$ of size $m \ (\le kr)$, represented by an $r \times m$ matrix $A$ over a field $\F$. Then, there exists a sensitivity oracle that, for any changes to $M$ consisting of at most $\delta$ updates to $A$, can decide whether the entire ground set $E$ can be covered by (partitioned into) $k$ independent sets of the updated matroid $M'$. This oracle requires $O((kr)^\omega)$ preprocessing time, $O((kr)^2 \log(kr))$ space, and answers queries in $O((\delta k)^\omega)$ time.
\end{theorem}

\begin{proof}
Let $U \in \F^{kr \times m}$ be the $k$-fold union representation matrix for $M$, constructed using \Cref{lemma:union_rep}. Let $\Delta = kr - m$, and let $R \in \F^{kr \times \Delta}$ be a matrix of independent uniformly random elements from $\F$. We define the $kr \times kr$ augmented matrix as $U_{aug} := [U \mid R]$.

By \Cref{lemma:aug-linear-rep}, the set $E$ can be covered by $k$ independent sets if and only if the row rank of $U$ is $m$. With high probability, this occurs if and only if the scalar matrix $U_{aug}$ is nonsingular.

To efficiently support queries, we apply the preprocessing algorithm from \Cref{lemma:BrandS-det-General-Q} to $U_{aug}$. This preprocessing requires $O((kr)^\omega)$ time and computes an $O((kr)^2 \log(kr))$ sized oracle.

Now any $\delta$ changes to $A$ translates to exactly $k\delta$ changes in $U$ with the random augmenting matrix $R$ completely unchanged. So, the updated augmented matrix $U'_{aug}$ can be written as $U_{aug} + C$, where the additive update matrix $C \in \F^{kr \times kr}$ contains at most $k\delta$ non-zero entries.

To answer a $k$-covering query for the updated matroid $M'$, we must determine if $\det(U'_{aug})$ is nonzero. By \Cref{lemma:BrandS-det-General-Q}, this requires exactly $O((\delta k)^\omega)$ time. This completes the proof.
\end{proof}

\subsubsection{Subset Sensitivity Oracle for General Linear Matroids}
Analogous to the matroid packing problem, we now present a subset sensitivity oracle for matroid covering. Here, a susceptible set $E_Q$ is identified during preprocessing, and updates are confined to swapping elements in and out of $E_Q$.

\begin{theorem}
\label{thm:subset_sensitivity_covering}
Let $k \ge 1$ be an integer. Let $M = (E, \I)$ be a linear matroid of rank $r$ on a ground set $E$ of initial size $m \le kr$. Let $E_Q$ be a susceptible set of size $\sigma$, and suppose $M$ extends to a matroid on $E_0 := E \cup E_Q$ with an $r \times |E_0|$ representation matrix $\bar{A}$ over $\F$.
Then there exists a subset sensitivity oracle that, for any $f$ updates confined to $E_Q$ (where elements can be added or deleted such that the updated ground set size $m_{curr}$ remains bounded by $kr$), can decide whether the current ground set can be covered by $k$ independent sets.
This oracle requires $O(k^2|E_0| + \sigma^2 (kr)^{\omega - 2} + (kr)^\omega)$ preprocessing time, $O(\sigma^2)$ space, and answers queries in $O(f^\omega)$ time.
\end{theorem}

\begin{proof}
Let $\bar{U} \in \F^{kr \times |E_0|}$ be the $k$-fold union representation matrix for the extension of $M$ to $E_0$.
The initial ground set has size $m$, requiring $\delta_0 = kr - m$ random columns. Since any query consists of at most $f$ insertions/deletions confined to $E_Q$, we generate $\delta_{max} = \delta_0 + f$ independent, uniformly random columns to form a matrix $R \in \F^{kr \times \delta_{\max}}$ and partition it into two sub-matrices: $R_{static}$ containing $\max(0,\delta_0 - f)$ columns, and $R_{switch}$ containing the remaining columns. We define the base augmented matrix as $U := [\bar{U} \mid R_{static} \mid R_{switch}]$.
Let $D$ be a diagonal selection matrix matching the column dimensions of $U$. For the initial ground set $S \subseteq E_0$ of size $m$, we initialize $D$ such that $D_{e,e} = 1$ for $e \in S$, $D_{i,i} = 1$ for all columns in $R_{static}$, and $D_{i,i} = 1$ for exactly $\min(\delta_0,f)$ arbitrary columns in $R_{switch}$. All other diagonal entries are set to $0$. This ensures exactly $kr$ columns are active. We define the Covering Laplacian as:
$$\L_{cov} := U D U^T.$$
Notice that $\L_{cov} = M_{curr} M_{curr}^T$, where $M_{curr}$ is the $kr \times kr$ submatrix of active columns. Because $M_{curr}$ is square, $\det(\L_{cov}) = \det(M_{curr})^2$. By \Cref{lemma:aug-linear-rep}, the active set is $k$-coverable if and only if $\det(\L_{cov}) \neq 0$.
Computing $\L_{cov}$ takes $O(k^2|E_0| + kr \delta_{max})$ time. A query modifies at most $f$ elements in $E_Q$, changing $m_{curr}$ by some $\delta m \in [-f, f]$. To compensate and maintain exactly $kr$ active columns, we toggle exactly $|\delta m|$ diagonal entries corresponding to $R_{switch}$. The columns of $R_{static}$ remain permanently active.
Therefore, the dynamically toggled entries belong to $E_Q$ (size $\sigma$) and $R_{switch}$ (size $2f$). We define the data structure's susceptible set as $Q := E_Q \cup R_{switch}$. Since $f \le \sigma$, the total size is $|Q| = \sigma + 2f = O(\sigma)$.
We invoke the scalar formulation ($d=0$) of \Cref{lemma:det} with $A := U$, $B := U$, the diagonal matrix $D$, and the susceptible set $Q$. The data structure construction requires $O(|Q|^2 (kr)^{\omega - 2} + (kr)^\omega) = O(\sigma^2 (kr)^{\omega - 2} + (kr)^\omega)$ time and $O(|Q|^2) = O(\sigma^2)$ space. An update toggles at most $f$ entries in $E_Q$ and at most $f$ entries in $R_{switch}$, modifying a total of $\le 2f$ diagonal entries. Querying the determinant thus takes $O((2f)^\omega) = O(f^\omega)$ time.
\end{proof}

\begin{remark}[Comparison of Oracle Models for Covering]
\label{rem:oracle-comparison-covering}
Similar to the packing problem, the matroid covering problem has following space and query-time trade-offs:
\begin{center}
\renewcommand{\arraystretch}{1.25}
\begin{tabular}{lccc}
\toprule
& \textbf{General Sensitivity} & \textbf{Subset Sensitivity} \\
& (\Cref{thm:sensitivity_covering}) & (\Cref{thm:subset_sensitivity_covering}) \\
\midrule
Space & $O((kr)^2 \log r)$ & $O(\sigma^2)$ \\
Query Time & $O((\delta k)^\omega)$ & $O(f^\omega)$ \\
\bottomrule
\end{tabular}
\end{center}
\end{remark}

\section{Max-Flows, Min-Cuts, $s$-Min-Cuts, Global Min-Cut}
\label{sec:flows}

\subsection{Sensitivity Oracles for All-Pairs $k$-Bounded Flow}
We first establish the matroid formulation for all-pairs bounded flows. Let $k$ be a positive integer. Let $H = (V_H, E_H)$ be a
directed multigraph on $N$ vertices with a designated source $s^*$ and a
designated sink $t^*$ such that the out-degree of $s^*$ and the in-degree of
$t^*$ are both equal to $\lambda \le k$. We add to each vertex
$v \in V_H \setminus \{s^*, t^*\}$ a set of exactly $k$ zero-cost self-loops,
and to vertices $s^*$ and $t^*$ we add $k - \lambda$ zero-cost self-loops.
Let $M$ denote the number of edges in the resulting graph $H$. We define two
partition matroids on the edge set of $H$:
\begin{itemize}
\item $M_{\mathrm{out}}$: A subset $S \subseteq E_H$ is independent if every
vertex $v \in V_H \setminus \{t^*\}$ has at most one outgoing edge in $S$.
This is represented by a matrix
$A_{\mathrm{out}} \in \{0, 1\}^{(N-1) \times M}$, where rows are indexed by
$V_H \setminus \{t^*\}$. The column for an edge $e = (u, v)$ has a $1$ at row
$u$ (if $u \neq t^*$) and $0$ elsewhere.

\item $M_{\mathrm{in}}$: A subset $S \subseteq E_H$ is independent if every
vertex $v \in V_H \setminus \{s^*\}$ has at most one incoming edge in $S$.
This is represented by a matrix
$A_{\mathrm{in}} \in \{0, 1\}^{(N-1) \times M}$, where rows are indexed by
$V_H \setminus \{s^*\}$. The column for an edge $e = (u, v)$ has a $1$ at row
$v$ (if $v \neq s^*$) and $0$ elsewhere.
\end{itemize}
Both matroids have rank $N - 1$. The following lemma establishes the
equivalence between a $k$-packing of common bases in these matroids and the
existence of edge-disjoint $(s^*, t^*)$-paths.

\begin{lemma}
\label{lemma:path_packing_static}
The directed multigraph $H$ contains $\lambda$ edge-disjoint $(s^*, t^*)$-paths if and only if there exists a set $S \subseteq E_H$ of size $k(N-1)$ that can be partitioned into $k$ bases of $M_{\mathrm{out}}$ as well as $M_{\mathrm{in}}$ (the partitionings may differ).
\end{lemma}

\begin{proof}
For the forward direction, suppose $H$ contains $\lambda$ edge-disjoint $(s^*, t^*)$-paths $P_1, \ldots, P_\lambda$. Using the $k$ self-loops at each $v \in V_H \setminus \{s^*, t^*\}$, we extend these paths to a family of $\lambda$ edge-disjoint sets $(B_1, \ldots, B_\lambda)$ by adding self-loops so that for each $i \in [\lambda]$: (i) $E(P_i) \subseteq B_i$, and (ii) every $v \in V_H \setminus \{s^*, t^*\}$ has exactly one incoming and one outgoing edge from $B_i$, while $s^*$ has exactly one outgoing edge in $B_i$ and $t^*$ has exactly one incoming edge in $B_i$. For the remaining $k - \lambda$ indices, we use only self-loops to obtain edge sets satisfying the same in-/out-degree constraints. Let ($B_1, \dots, B_k)$ denote the resulting family. Each $B_i$ has size $N - 1$ and is independent in both $M_{\mathrm{out}}$ and $M_{\mathrm{in}}$, hence a basis of both. Thus $S := \biguplus_{i=1}^k B_i$ has size $k(N-1)$ and partitions into $k$ bases of each matroid.

For the converse, suppose $S \subseteq E_H$ has size $k(N-1)$ and partitions into $k$ bases of $M_{\mathrm{out}}$ as well as $M_{\mathrm{in}}$. Since both matroids are strongly base-orderable, \Cref{theorem:SBO-ind} implies $S$ can be partitioned into $k$ common independent sets $S = C_1 \uplus \dots \uplus C_k$. Fix $i \in [k]$. Independence in $M_{\mathrm{out}}$ forces every $v \in V_H \setminus \{t^*\}$ to have out-degree at most one in $(V_H, C_i)$, and independence in $M_{\mathrm{in}}$ forces every $v \in V_H \setminus \{s^*\}$ to have in-degree at most one. Hence $(V_H, C_i)$ is a disjoint union of directed paths and cycles, and any nontrivial path must start at $s^*$ and end at $t^*$. Since $|S| = k(N-1)$ and each $|C_i| \le N - 1$, equality holds throughout, so each $C_i$ contains at most one $(s^*, t^*)$-path together with vertex-disjoint cycles and self-loops. Because $s^*$ has out-degree $\lambda$ and $k - \lambda$ self-loops (similarly for $t^*$), exactly $\lambda$ of these $C_i$'s contain an $(s^*, t^*)$-path. The disjointness of the $C_i$'s then yields $\lambda$ edge-disjoint $(s^*, t^*)$-paths in $H$.
\end{proof}

We now return to the sensitivity setting. Let $G = (V, E, \cost)$ be a directed graph with $n$ vertices, $m$ edges, and integer edge costs in $[-W, W]$. Let $k \in [1,n]$ be a fixed integer. For any query $(s, t, U)$ with $s, t \in V$ and $U$ a set of $f$ edge insertions or deletions, our goal is to determine $\lambda^* := \min(k, \lambda_{G + U}(s, t))$ and the minimum total cost of an $(s, t)$-flow of value $\lambda^*$.

\paragraph{Preprocessing}
We construct an auxiliary graph $G^*$ on $N = n + 2$ vertices by introducing a super-source $s^*$ and a super-sink $t^*$, and augmenting the edge set with $k$ zero-cost self-loops at each vertex $v \in V$, yielding $M = m + kn$ edges in total. We then invoke the preprocessing of \Cref{thm:sensitivity_packing} on the partition matroids $M_{\mathrm{out}}, M_{\mathrm{in}}$ associated with $G^*$, instantiating both the decision oracle (for unweighted queries) and the weighted oracle (for min-cost queries).

As linear representations $A_{\mathrm{out}}, A_{\mathrm{in}} \in \{0, 1\}^{(N - 1) \times M}$ are column-sparse, the decision oracle requires $O((kn)^\omega )$ preprocessing time and $O((kn)^2 \log n)$ space, while the weighted oracle requires $\widetilde{O}(W(kn)^3)$ preprocessing time and $O(W(kn)^3 \log n)$ space.

\paragraph{Query Reduction}
For any $s, t \in V$ and $\lambda \le k$, let $E_{s, t, \lambda}$ consist of:
(i) $\lambda$ zero-cost parallel edges from $s^*$ to $s$, (ii) $\lambda$
zero-cost parallel edges from $t$ to $t^*$, and (iii) $k - \lambda$ zero-cost
self-loops at each of $s^*$ and $t^*$. By
\Cref{lemma:path_packing_static}, $G + U$ admits $\lambda$ edge-disjoint
$(s, t)$-paths if and only if the matroids $M_{\mathrm{out}}, M_{\mathrm{in}}$
on the edge set of $G^* + U + E_{s, t, \lambda}$ admit a $k$-packing of
common bases.

The matroid representations of $G^* + U + E_{s, t, \lambda}$ differ from
those of $G^*$ on $O(f)$ edges from $U$ together with the parallel-edge and
self-loops in $E_{s, t, \lambda}$. By
\Cref{rem:parallel-edges-update-packing-laplacian}, the parallel
edges/self-loops in $E_{s, t, \lambda}$ contribute only $O(k^2)$ entries to the additive
update of $\L_{\mathrm{pack}}$. Hence the total update size is $\delta = O(f k^2)$. We thus obtain the following result.

\begin{theorem}
\label{theorem:FT-all-pairs-flow}
Let $G=(V,E,\cost)$ be an $n$-vertex directed graph with integer costs in $[-W,W]$, and $k \in [1,n]$. Then there exist the following sensitivity oracles that, for any query $(s,t,U) \in V^2 \times E^f$, answer queries on $G+U$ with high probability:

\begingroup
\setlength{\aboverulesep}{0pt}
\setlength{\belowrulesep}{0pt}
\begin{center}
\resizebox{0.98\linewidth}{!}{%
\renewcommand{\arraystretch}{1.25}
\begin{tabularx}{\linewidth}{@{}l l l l X@{}}
\rowcolor{yellow!15}
\textsc{Oracle} & \textsc{Prep.} & \textsc{Space} & \textsc{Query} & \textsc{Output} \\
\midrule
Bounded Flow &
$O((kn)^\omega)$ &
$O((kn)^2 \log n)$ &
$O((fk^2)^\omega \log k)$ &
$\lambda^*=\min(k,\lambda_{G+U}(s,t))$ \\
Nearest Min-Cut &
$O((kn)^\omega)$ &
$O((kn)^2 \log n)$ &
$O((fk^2)^\omega n)$ &
$\NMC_{G+U}(s,t)$ (if $\lambda^* < k$) \\
Min-Cost Flow &
$\widetilde{O}(W(kn)^3)$ &
$O(W(kn)^3 \log n)$ &
$\widetilde{O}((fk^2)^\omega kWn)$ &
Min-cost for flow of value $\lambda^*$ \\
\end{tabularx}
}
\end{center}
\endgroup
\end{theorem}

\begin{proof}
We instantiate the matroid intersection packing oracles of
\Cref{thm:sensitivity_packing} on the partition matroids
$M_{\mathrm{out}}, M_{\mathrm{in}}$ associated with the auxiliary graph
$G^*$. Recall from the query reduction above that finding $\lambda$
edge-disjoint $(s, t)$-paths in $G + U$ reduces to deciding whether
$M_{\mathrm{out}}, M_{\mathrm{in}}$ on the edge set
$G^* + U + E_{s, t, \lambda}$ admit a $k$-packing of common bases. By
\Cref{rem:parallel-edges-update-packing-laplacian}, the parallel-edge and
self-loops in $E_{s, t, \lambda}$ together with the $f$ updates in
$U$ yield an additive update of size $\delta = O(f k^2)$ to the packing
Laplacian.

\paragraph{Bounded Flow Oracle.}
We invoke the decision oracle of \Cref{thm:sensitivity_packing}, which has preprocessing time $O((kn)^\omega)$ and space
$O((kn)^2 \log n)$. For a query $(s, t, U)$, we binary search over
$\lambda \in [0, k]$ and, for each $\lambda$, test whether the updated
Laplacian on $G^* + U + E_{s, t, \lambda}$ is non-singular. Each test
takes $O(\delta^\omega) = O((fk^2)^\omega)$ time, and the $O(\log k)$
tests yield $\lambda^*$ in $O((fk^2)^\omega \log k)$ time.

\paragraph{Nearest Min-Cut Oracle.}
The preprocessing is identical to the bounded flow oracle. Given a query
$(s, t, U)$, we first compute $\lambda := \lambda_{G + U}(s, t)$ in
$O((fk^2)^\omega \log k)$ time. For each $w \in V \setminus \{s, t\}$, we
invoke the bounded flow oracle on $G + U + (w, t)$ to test whether
$\lambda_{G + U + (w, t)}(s, t) = \lambda + 1$. By
\Cref{lemma:nearest-mincut}, the set
\[
\bigl\{w \in V : \lambda_{G + U + (w, t)}(s, t) = \lambda + 1\bigr\}
\]
is the source side of the nearest min-cut, and
its complement is the sink side. Each of the $n$ auxiliary
queries adds a single edge to $U$, preserving the update size
$\delta = O(f k^2)$, and takes $O((fk^2)^\omega)$ time. The total
query time is $O((fk^2)^\omega \cdot n)$.

\paragraph{Min-Cost Flow Oracle.}
We invoke the weighted oracle of \Cref{thm:sensitivity_packing}, which has preprocessing time $\widetilde{O}(W(kn)^3)$ and space $O(W(kn)^3 \log n)$. For a query $(s, t, U)$, we first compute $\lambda^* := \min(k, \lambda_{G + U}(s, t))$ using the bounded flow oracle in $O((fk^2)^\omega \log k)$ time. By \Cref{lemma:path_packing_static}, the minimum cost of an $(s, t)$-flow of value $\lambda^*$ in $G + U$ equals the minimum weight of a $k$-packing-set in $M_{\mathrm{out}}, M_{\mathrm{in}}$ on $G^* + U + E_{s, t, \lambda^*}$, which the weighted oracle returns in $\widetilde{O}(W \cdot kn \cdot \delta^\omega) = \widetilde{O}((fk^2)^\omega \cdot kWn)$ time.
\end{proof}

\subsection{Sensitivity Oracles for $(s,t)$-Max-Flow and $(s,t)$-Min-Cut}
\label{subsec:gen-st-flow}

We now turn to the special case of the general max-flow problem for a fixed source-sink pair $(s, t)$. The oracle of \Cref{theorem:FT-all-pairs-flow} with $k = n$ already solves this problem, but its query time is quite high. We show that much better bounds are achievable for the case of a fixed source-sink pair through a flow-rerouting reduction that lets us use \Cref{theorem:FT-all-pairs-flow} as a black box on an auxiliary graph with capacity bound $k = O(f)$.

Let $G = (V, E)$ be a directed graph with unit edge capacities and costs in $[-W, W]$, with fixed source $s$ and sink $t$. We precompute a min-cost $(s, t)$-max-flow $h$ once, during preprocessing. For any query set $U$ of updates, we aim to answer the following: how can $h$ be \emph{repaired} into a valid flow on $G + U$, and by how much does its cost change.

To achieve this, we augment $G$ to a graph $\G$ by adding $2f$ zero-cost parallel $(s, t)$-edges, and saturating the first $f$ of them in the preprocessed flow $h$ (thus, value of flow $h$ (denoted by $\val(h)$) in $\G$ is equal to $f+\lambda_{G}(s,t)$). These parallel edges act as a buffer: when updates arrive, they absorb or release flow as needed, and the number of saturated parallel $(s,t)$ edges in the repaired flow tells us the new max-flow
value.

\begin{lemma}
\label{lemma:reroute-1}
For any update set $U$ of size at most $f$, if every $(s,t)$-flow of value $\val(h)$ in $\G + U$ saturates a minimum of $i$ parallel $(s,t)$ edges, then
$\lambda_{G + U}(s, t) = \lambda_G(s, t) + f - i$.
\end{lemma}

\begin{proof}
Let $\beta := \lambda_{G + U}(s, t) - \lambda_G(s, t)$. If $\beta \ge 0$, $\beta$ units of flow can be moved from the parallel edges back into $G + U$, so the new flow saturates $i = f - \beta$ parallel edges. If $\beta < 0$, the updates cancel $|\beta|$ units of flow in $G$, which must be redirected onto the parallel edges, giving $i = f + |\beta|$. In
both cases, $\beta = f - i$.
\end{proof}

\paragraph{Reducing to an All-Pairs Query}
It remains to test, given a candidate value of $i$, whether $h$ can in fact be repaired in $\G + U$ using exactly $i$
parallel $(s,t)$ edges. We construct an auxiliary multigraph $\G^*_h$ from the residual graph $\G_h$ by adjoining a dummy source $s^*$ and a dummy sink $t^*$, both initially isolated.

Given a query $U = (I, F)$, let $F_h := \{e \in F : h(e) = 1\}$ denote the set of deleted edges that carried flow, and for each such $e = (x, y)$ let $e^{\rev} = (y, x)$ denote its residual counterpart in $\G_h$. We define the auxiliary update set $\U^* = (I^*, F^*)$ on $\G^*_h$ as
\[
\U^*
\;:=\;
\Bigl(\,
\underbrace{I \;\cup\; \{(s^*, x), (y, t^*) : (x, y) \in F_h\}}_{I^*}
\;,\;
\underbrace{(F \setminus F_h) \;\cup\; \{e^{\rev} : e \in F_h\}}_{F^*}
\,\Bigr).
\]
That is, $I^*$ inserts the edges of $I$ together with zero-cost edges incident to the dummy vertices, while $F^*$ deletes the edges of $F$ that did not carry flow and the residual counterparts of those that did. Let $\G^* := \G^*_h + \U^*$; note that $|\U^*| = O(|U|)$.

\begin{lemma}
\label{lemma:reroute-2}
Let $U = (I, F)$ be a set of updates, and write $d := |F_h|$. The flow $h$ in $\G$ can be repaired into a valid flow $h'$ in $\G + U$ if and only if the $(s^*, t^*)$-max-flow in $\G^*$ equals $d$. Moreover, the minimum cost of any such repaired flow satisfies
\[
\cost(h') = \cost_{\G}(h) + \cost_{\G^*}(s^*, t^*)
- \sum_{e \in F_{h}} \cost(e).
\]
\end{lemma}

\begin{proof}
Each edge $(x, y) \in F_h$ leaves $h$ with one unit of excess at $x$ and one of deficit at $y$; all other vertices
remain balanced. Repairing $h$ is therefore equivalent to routing $|F_h|$ edge-disjoint paths in $\G_h + \U^*$ from the excess set $X := \{x : (x, y) \in F_h\}$ to the deficit set $Y := \{y : (x, y) \in F_h\}$. Since $\G^*$ adds exactly one zero-cost edge from $s^*$ into each $x \in X$ and one from each $y \in Y$ to $t^*$, such a routing exists if and only if all $d$ edges leaving $s^*$ in $\G^*$ are saturable, i.e., $\lambda_{\G^*}(s^*, t^*) = d$.

We now derive the cost relation. Any min-cost $(s^*, t^*)$-flow $h^*$ in $\G^*$ is decomposable into $d$ paths $P^*_j = s^* \to x_j \rightsquigarrow y_{\sigma(j)} \to t^*$, for some permutation $\sigma$, and a (possibly empty) collection of negative-cost cycles $C_1, \ldots, C_\ell$ (which can appear due to insertions).  The edges incident to $s^*$ and $t^*$ have zero cost and they cannot lie on any cycle, so $\cost_{\G^*}(h^*)$ equals the total cost of the subpaths $P_j := P^*_j[x_j, y_{\sigma(j)}]$ $(1\leq j\leq d)$ and cycles $C_1, \ldots, C_\ell$ in $\G_h$. That is,
\[
\cost_{\G^*}(s^*, t^*)
\;=\;
\sum_{j=1}^{d} \cost_{\G_h}(P_j) + \sum_{i=1}^{\ell} \cost_{\G_h}(C_i).
\]

The repaired flow $h'$ is obtained from $h$ by canceling flow along edges in $F_h$ and rerouting it through the
$P_j$'s and $C_i$'s. This removes the $|F_h|$ deleted-edge
contributions to $\cost(h)$ and adds the right-hand side above:
\[
\cost(h')
\;=\;
\cost_\G(h)
- \sum_{e \in F_h} \cost(e)
+ \sum_j \cost_{\G_h}(P_j) + \sum_i \cost_{\G_h}(C_i),
\]
which yields the stated cost identity.
\end{proof}

\paragraph{Oracle Construction}
We next provide the oracle construction that applies \Cref{theorem:FT-all-pairs-flow} to $\G^*_h$.

\begin{theorem}
\label{theorem:st-flow}
Let $G=(V,E,\cost)$ be an $n$-vertex directed graph with integer edge costs in $[-W,W]$, and a designated source $s$ and sink $t$. Then there exist the following sensitivity oracles that, for any set $U\in E^f$, answer queries on $G+U$ with high probability:

\begingroup
\setlength{\aboverulesep}{0pt}
\setlength{\belowrulesep}{0pt}
\begin{center}
\resizebox{0.98\linewidth}{!}{%
\renewcommand{\arraystretch}{1.25}
\begin{tabularx}{\linewidth}{@{}l l l l X@{}}
\rowcolor{yellow!15}
\textsc{Oracle} & \textsc{Prep.} & \textsc{Space} & \textsc{Query} & \textsc{Output} \\
\midrule
Max-Flow &
$O((fn)^\omega)$ &
$O((fn)^2 \log n)$ &
$\widetilde{O}(f^{3\omega})$ &
$\lambda_{G+U}(s,t)$ \\
Nearest Min-Cut &
$O((fn)^\omega)$ &
$O((fn)^2 \log n)$ &
$\widetilde{O}(f^{3\omega} n)$ &
$\NMC_{G+U}(s,t)$ \\
Min-Cost Flow &
$\widetilde{O}(W(fn)^3)$ &
$O(W(fn)^3 \log n)$ &
$\widetilde{O}(f^{3\omega+1} Wn)$ &
Min-cost max-flow in $G+U$ \\
\end{tabularx}
}
\end{center}
\endgroup
\end{theorem}

\begin{proof}
During preprocessing, we compute a min-cost $(s, t)$-max-flow $h$ in $G$. We augment $G$ with $2f$ zero-cost parallel $(s, t)$-edges, and saturate the first $f$ of them in the preprocessed flow $h$. Computing $h$ takes $m^{1+o(1)} \log W$ time using the algorithm of~\cite{ChenKLPGS23}. We then construct $\G^*_h$ and instantiate the all-pairs oracles of \Cref{theorem:FT-all-pairs-flow} on $\G^*_h$ with capacity bound $k = O(f)$, from which the preprocessing
time and space follow directly.

For any query $U$, the auxiliary update set $\U^*$ satisfies $|\U^*| = O(f)$, since each edge in $F$ contributes $O(1)$ edges to $\U^*$. Plugging $|\U^*| = O(f)$ and $k = O(f)$ into the query bounds of \Cref{theorem:FT-all-pairs-flow} gives  $(|\U^*| \cdot k^2)^\omega = f^{3\omega}$ cost.

\paragraph{Max-Flow Oracle.}
By \Cref{lemma:reroute-1}, $\lambda_{G + U}(s, t) = \lambda_G(s, t) + f - i$, where $i$ is the smallest number of parallel edges required to repair $h$. We binary search over $i \in [0, 2f]$; for each candidate, we set $U' = (I, F \cup \{e_{i+1}, \ldots, e_{2f}\})$ and apply \Cref{lemma:reroute-2} to translate the repair test into an $(s^*, t^*)$-max-flow query on $\G^*_h + \U^*$. Thus a query costs $\widetilde{O}(f^{3\omega})$ time.

\paragraph{Nearest Min-Cut Oracle.}
We first compute $\lambda := \lambda_{G + U}(s, t)$ via the max-flow oracle. For each $w \in V \setminus \{s, t\}$, we invoke the max-flow oracle on $G + U + (w, t)$ to test whether $\lambda_{G + U + (w, t)}(s, t) = \lambda + 1$. By
\Cref{lemma:nearest-mincut}, the set
\[
\bigl\{w \in V : \lambda_{G + U + (w, t)}(s, t) = \lambda + 1\bigr\}
\]
is the source side of the nearest min-cut, and its complement is the sink side. Each of the $n$ auxiliary queries preserves the update size $O(f)$ and costs $\widetilde{O}(f^{3\omega})$, resulting in total query time
$\widetilde{O}(f^{3\omega} \cdot n)$.

\paragraph{Min-Cost Flow Oracle.}
By \Cref{lemma:reroute-2}, the minimum cost of an $(s, t)$-flow of value $\lambda_{G + U}(s, t)$ in $G + U$ is determined by $\cost_{\G^*}(s^*, t^*)$, the minimum cost of an $(s^*, t^*)$-flow of value $|F_{h}|$ in $\G^*$. We obtain this using the min-cost flow oracle of \Cref{theorem:FT-all-pairs-flow} in $\widetilde{O}((|\U^*| \cdot k^2)^\omega \cdot kWn) = \widetilde{O}(f^{3\omega + 1} \cdot Wn)$ time.
\end{proof}



\subsection{Sensitivity Oracles for $s$-Min-Cut and Global Min-Cut}
\label{sec:s-mincut}

We next present sensitivity oracles for two related cut problems. For a fixed source $s \in V$, the \emph{$s$-min-cut} value in $G$ is $\displaystyle \mu_G(s) := \min_{v \in V \setminus \{s\}} \lambda_G(s, v)$. The \emph{global min-cut} value is $\displaystyle \mu_G := \min_{s \in V} \mu_G(s)$. In this section, we show that the $(s, t)$-max-flow sensitivity oracles of \Cref{theorem:st-flow} can be lifted, with only a small overhead in preprocessing, into sensitivity oracles for both problems.

Our reduction relies on first computing an optimal augmentation of the graph that maximizes the $s$-min-cut.

\begin{theorem}[\cite{Frank90}, Theorem 3.1]
\label{theorem:s-connectivity-augmentation}
There exists a polynomial time algorithm to solve the following problem: Given an unweighted directed graph $G = (V, E)$ with a fixed source $s$ and a number $f$, output a set $I^* \subseteq V \times V$ of $f$ edges such that adding the edges of $I^*$ to $G$ results in the maximum increase in the $s$-min-cut value of $G$, i.e., $\displaystyle I^* = \arg \max_{\substack{I \subseteq V \times V \\ |I| = f}} \mu_{G+I}(s)$.\footnote{Although the theorem gives a characterization for the existence of such augmentation, the polynomial time algorithm can be inferred straightforwardly from the proof. The algorithm finds the minimum number of edges needed to increase the cut by $\gamma$; we can binary search for finding $I^*$.}
\end{theorem}

Let $I^*$ be the optimal set of $f$ edges output by \Cref{theorem:s-connectivity-augmentation}. We compute this augmentation in polynomial time $T_{\text{aug}}$ during preprocessing, and we define $\lambda^* := \mu_{G+I^*}(s)$ as the maximum possible $s$-min-cut achievable by adding $f$ edges to $G$.

For any general query update set $U$ consisting of at most $f$ insertions and deletions, we cannot guarantee that the insertions in $U$ are a subset of $I^*$. However, we can use $I^*$ to bound the candidate target vertices for the new $s$-min-cut. Let $Y$ be the set of head vertices of the edges in $I^* \cup U$. Since $|I^*| = f$ and $|U| \le f$, we have $|Y| \le 2f$.

\begin{lemma}
\label{lemma:insertion-deletion-cut}
For any update set $U$ of size at most $f$, let $Y$ be the set of head vertices of edges in $I^* \cup U$. Then,
$$
\mu_{G + U}(s)
\;=\;
\min\Bigl(\lambda^*,\; \min_{\substack{y \in Y \\ y \neq s}} \lambda_{G + U}(s, y)\Bigr).
$$
\end{lemma}

\begin{proof}
For the $\le$ direction, observe that $\mu_{G+U}(s) \le \lambda_{G+U}(s, y)$ for every $y \neq s$ by the definition of the $s$-min-cut. Furthermore, because $I^*$ maximizes the $s$-min-cut value over all possible insertions of $f$ edges, and $U$ adds at most $f$ edges (and potentially deletes others), we have $\mu_{G + U}(s) \le \mu_{G + I^*}(s) = \lambda^*$.

For the $\ge$ direction, let $(A, B)$ be a minimum $s$-cut in $G + U$. We consider two cases based on whether any edge in $I^* \cup U$ crosses from $A$ to $B$. If no edge of $I^* \cup U$ crosses from $A$ to $B$, then no edge added by $I^*$ crosses the cut, and no edge added or deleted by $U$ crosses the cut. Consequently, the edges crossing from $A$ to $B$ are identical in $G+U$ and $G+I^*$. Thus, the capacity of $(A,B)$ in $G+U$ equals its capacity in $G+I^*$, which is at least $\mu_{G+I^*}(s) = \lambda^*$.

Otherwise, some edge $e \in I^* \cup U$ crosses from $A$ to $B$. Let $e = (x, y)$. Since $s \in A$ and $y \in B$, we have $y \neq s$. Because $e \in I^* \cup U$, we have $y\in Y$. The partition $(A, B)$ is thus an $(s, y)$-cut in $G + U$, of capacity at least $\lambda_{G + U}(s, y) \ge \min_{y' \in Y, y' \neq s} \lambda_{G + U}(s, y')$. 

Hence $\mu_{G+U}(s) \ge \min\bigl(\lambda^*, \min_{y \in Y, y \neq s} \lambda_{G + U}(s, y)\bigr)$. This completes the proof.
\end{proof}

\begin{theorem}
\label{theorem:s-min-cut}
Let $G = (V, E, \cost)$ be an $n$-vertex directed graph, and let $s \in V$. Then there exist the following sensitivity oracles that, for any set of edge updates $U$ where $|U| \le f$, answer queries on $G + U$ with high probability:
\begingroup
\setlength{\aboverulesep}{0pt}
\setlength{\belowrulesep}{0pt}
\begin{center}
\resizebox{0.98\linewidth}{!}{%
\renewcommand{\arraystretch}{1.25}
\begin{tabularx}{\linewidth}{@{}l l l l X@{}}
\rowcolor{yellow!15}
\textsc{Oracle} & \textsc{Prep.} & \textsc{Space} & \textsc{Query} & \textsc{Output} \\
\midrule
$s$-Min-Cut Value &
$O(f^\omega n^{\omega+1} + T_{\text{aug}})$ &
$O(f^2 n^3 \log n)$ &
$\widetilde{O}(f^{3\omega+1})$ &
$\mu_{G+U}(s)$ \\
$s$-Min-Cut Partition~~ &
$O(f^\omega n^{\omega+1} + T_{\text{aug}})$~~ &
$O(f^2 n^3 \log n)$~~ &
$\widetilde{O}(f^{3\omega} n)$~~ &
$s$-min-cut in $G+U$ \\
\end{tabularx}
}
\end{center}
\endgroup
\end{theorem}

\begin{proof}
During preprocessing, we compute the optimal augmenting set $I^*$ of size $f$ in polynomial time $T_{\text{aug}}$ and compute $\lambda^* = \mu_{G+I^*}(s)$
For every vertex $v \in V \setminus \{s\}$, we instantiate the $(s, t)$-max-flow and nearest-min-cut sensitivity oracles of \Cref{theorem:st-flow} on $G$ with source $s$ and sink $v$, with a budget of $f$. 

Given a query update set $U$, we construct the candidate sink set $Y = \{y \mid (x, y) \in I^* \cup U\}$. For each $y \in Y$ such that $y \neq s$, we call the $(s, y)$ max-flow oracle on $G + U$ to compute $\lambda_{G+U}(s, y)$. Each of the at most $2f$ valid queries takes $\widetilde{O}(f^{3\omega})$ time. By \Cref{lemma:insertion-deletion-cut}, $\mu_{G+U}(s) = \min\bigl(\lambda^*,\ \min_{y \in Y, y \neq s} \lambda_{G+U}(s, y)\bigr)$. Evaluating these bounds takes total time $\widetilde{O}(f^{3\omega+1})$.


To report a corresponding vertex partition, we first compute $\mu_{G+U}(s)$. We form two cases, which are handled separately:

\begin{itemize}
    \item If $\mu_{G+U}(s) < \lambda^*$, then let $y^* = \arg\min_{y \in Y, y \neq s} \lambda_{G+U}(s, y)$. We invoke the $(s, y^*)$ nearest-min-cut oracle on $G + U$, which returns the desired partition in $\widetilde{O}(f^{3\omega} n)$ time. 

    \item If $\mu_{G+U}(s) = \lambda^*$, then scan all vertices to find a vertex $v$ such that $\lambda_{G+U}(s,v) = \lambda*$. Then, invoke the $(s, v)$ nearest-min-cut oracle on $G + U$. This procedure returns the desired partition in $\widetilde{O}(f^{3\omega} n)$ time. 
\end{itemize}

So, the total query time for reporting the vertex partition is bounded by $\widetilde{O}(f^{3\omega} n)$.

\end{proof}
\paragraph{Global Min-Cut}
For the global min-cut, we use the standard symmetry property: for any fixed vertex $s$,
\[
\mu_G \;=\; \min\bigl(\mu_G(s),\; \mu_{G^{\mathrm{rev}}}(s)\bigr),
\]
where $G^{\mathrm{rev}}$ is $G$ with all edges reversed. For any global min-cut $(A, B)$, either $s \in A$ (so $\mu_G(s) \le |\delta_G(A, B)|$) or $s \in B$ (so $\mu_{G^{\mathrm{rev}}}(s) \le |\delta_{G^{\mathrm{rev}}}(A, B)|$). Hence, a single arbitrary source $s$, chosen at preprocessing, suffices.

\begin{theorem}
\label{theorem:global-min-cut}
Let $G = (V, E, \cost)$ be an $n$-vertex directed graph with unit edge capacities and integer edge costs in $[-W, W]$. There exist the following sensitivity oracles that, for any  update set $U$ of size $\le f$, answer queries on $G + U$ with high probability:
\begingroup
\setlength{\aboverulesep}{0pt}
\setlength{\belowrulesep}{0pt}
\begin{center}
\renewcommand{\arraystretch}{1.25}
\resizebox{0.98\linewidth}{!}{%
\begin{tabular}{@{}lllll@{}}
\rowcolor{yellow!15}
\textsc{Oracle} & \textsc{Prep.} & \textsc{Space} & \textsc{Query} & \textsc{Output} \\
\midrule
Global Min-Cut Value &
$O(f^\omega n^{\omega+1} + T_{\text{aug}})$ &
$O(f^2 n^3 \log n)$ &
$\widetilde{O}(f^{3\omega + 1})$ &
$\mu_{G+U}$ \\
Global Min-Cut Partition &
$O(f^\omega n^{\omega+1} + T_{\text{aug}})$ &
$O(f^2 n^3 \log n)$ &
$\widetilde{O}(f^{3\omega} n)$ &
Global-min-cut in $G+U$ \\
\end{tabular}%
}
\end{center}
\endgroup
\end{theorem}

\begin{proof}
We choose an arbitrary source vertex $s \in V$ during preprocessing. Note that we compute independent optimal augmentations (and thus distinct $I^*$ sets) for $G$ and $G^{\mathrm{rev}}$ with respect to $s$. We then instantiate the $s$-min-cut value and partition oracles from \Cref{theorem:s-min-cut} for both augmented graphs. The preprocessing time and space match the bounds stated previously. To determine the value of the global min-cut, we query the value oracles on both $G+U$ and $(G+U)^{\mathrm{rev}}$ and return the smaller value, which equals $\mu_{G+U}$. For a partition query, we first determine whether the minimum is attained in $G+U$ or $(G+U)^{\mathrm{rev}}$, and then invoke the corresponding partition oracle. The partition-query time remains identical to \Cref{theorem:s-min-cut}.
\end{proof}

\section{Arborescences, Spanning Trees, Colorful Spanning Trees, Arboricity}
\label{sec:trees}

\subsection{Sensitivity Oracles for $k$-Disjoint Arborescences}
To design sensitivity oracles for $k$-disjoint arborescences, we first formulate the problem using matroid $k$-packing. Consider a directed graph $H=(V_H, E_H)$ with $N$ vertices and $M$ edges, along with a designated root node $z$. A valid $z$-rooted spanning arborescence in $H$ is a set of $N-1$ edges forming a spanning tree in the underlying undirected graph, where every vertex $v \in V_H \setminus \{z\}$ has exactly one incoming edge, and $z$ has no incoming edges. We define two linear matroids on the ground set $E_H$.

\begin{definition}
\label{def:arborescence-matroids-direct}
The matroids $M_{\mathrm{graph}}$ and $M_{\mathrm{part}}$ on the ground set $E_H$, both of rank $r=N-1$, are defined with respect to the root $z$ as follows:
\begin{itemize}
\item $M_{\mathrm{graph}}$: A subset $S \subseteq E_H$ is independent if its edges do not form any cycles in the underlying undirected graph. This is represented over a field $\F$ by its incidence matrix $A_{\mathrm{graph}} \in \{-1, 0, 1\}^{(N-1) \times M}$, where rows are indexed by $V_H \setminus \{z\}$, and for an edge $e=(u,v)$, the $e$-th column has $-1$ at row $u$ (if $u \neq z$) and $1$ at row $v$ (if $v \neq z$).

\item $M_{\mathrm{part}}$: A subset $S \subseteq E_H$ is independent if every vertex $v \in V_H \setminus \{z\}$ has at most one incoming edge in $S$. This is represented by a matrix $A_{\mathrm{part}}\in \{0, 1\}^{(N-1) \times M}$, where rows are indexed by $V_H \setminus \{z\}$. A column for edge $e=(u,v)$ has a $1$ at row $v$ (if $v \neq z$) and $0$ elsewhere.
\end{itemize}
\end{definition}

The following theorem by Edmonds~\cite{Edmonds2003} establishes an equivalence between $k$-disjoint arborescences and packing $k$-common bases of the specified matroids.

\begin{theorem} [\cite{Edmonds2003}]
\label{theorem:edmonds_packing}
A directed graph $H$ contains $k$-edge-disjoint $z$-rooted arborescences if and only if there exists a set $S \subseteq E_H$ of size $k(N-1)$ that can be partitioned into $k$ bases of $M_{\mathrm{graph}}$ as well as $M_{\mathrm{part}}$ (the partitioning may differ).
\end{theorem}



We now return to our sensitivity setting. Let $G = (V, E, \cost)$ be a directed graph with $n$ vertices and integer edge costs in $[-W, W]$. For any query pair $(s,U)$, where $s \in V$ and $U$ is a set of $f$ edge updates (insertions or deletions), our goal is to determine if $G + U$ contains $k$ edge-disjoint $s$-rooted arborescences, and compute the minimum possible total cost of such a family.

\paragraph{Query Reduction}
We construct an auxiliary graph $G^*$ by adding a new isolated dummy root vertex $z$, setting $V(G^*) = V \cup \{z\}$ and $E(G^*) = E$. For any target root $s \in V$, let $E_{s,z}$ denote a set of $k$ zero-cost parallel edges directed from $z$ to $s$. Finding $k$-disjoint $s$-rooted arborescences in $G + U$ is equivalent to finding $k$-disjoint $z$-rooted arborescences in $G^* + U + E_{s,z}$, which by \Cref{theorem:edmonds_packing} maps directly to a $k$-packing on $M_{\mathrm{graph}}$ and $M_{\mathrm{part}}$.

\begin{theorem} \label{thm:sensitivity_arborescences}
Let $G = (V, E, \cost)$ be an $n$-vertex directed graph with integer edge costs in $[-W, W]$, and let $k \ge 1$. There exist the following sensitivity oracles that, for any $s\in V$ and any $U\in E^f$, return the
corresponding outputs on $G + U$ with high probability:
\begingroup
\setlength{\aboverulesep}{0pt}
\setlength{\belowrulesep}{0pt}
\begin{center}
\renewcommand{\arraystretch}{1.25}
\resizebox{0.98\linewidth}{!}{%
\begin{tabular}{@{}lllll@{}}
\rowcolor{yellow!15}
\textsc{Oracle} & \textsc{Prep.} & \textsc{Space} & \textsc{Query} & \textsc{Output} \\
\midrule
Decision &
$O((kn)^\omega)$ &
$O((kn)^2 \log n)$ &
$O((fk^2)^\omega)$ &
Whether $G + U$ has $k$-disjoint $s$-rooted arborescences \\
Weighted &
$\widetilde{O}(W(kn)^3)$ &
$O(W(kn)^3 \log n)$ &
$\widetilde{O}(W kn \cdot (fk^2)^\omega)$ &
Min total cost of such $k$-arborescences  \\
\end{tabular}%
}
\end{center}
\endgroup
\end{theorem}

\begin{proof}
We invoke the matroid packing oracles from \Cref{thm:sensitivity_packing} on $M_{\mathrm{graph}}$ and $M_{\mathrm{part}}$. For the auxiliary graph $G^*$ with $N=n+1$ vertices, both matroids have rank $r = n$.

When a query $(s, U)$ arrives, we incorporate the $f$ structural updates in $U$ and the $k$ parallel query edges $E_{s,z}$. By \Cref{rem:parallel-edges-update-packing-laplacian}, inserting the $k$ parallel edges in $E_{s,z}$ only contributes $O(k^2)$ entries to the additive update of the packing Laplacian $\mathcal{L}_{\mathrm{pack}}$. The $f$ standard edge updates contribute $O(fk^2)$ entries. Thus, the total additive update matrix has $\delta = O(fk^2)$ non-zero entries. Substituting $r=n$ and the effective update bound $\delta = O(fk^2)$ into the query bounds of \Cref{thm:sensitivity_packing} yields the stated bounds.
\end{proof}

\Cref{thm:sensitivity_arborescences}, together with the fact that a directed graph is $k$-strongly connected if and only if there exists a $k$-packing of arborescences rooted at an arbitrarily chosen vertex $s$ in both $G$ and its reverse, yields the following result.

\begin{corollary}[$k$-Strong-Connectivity Oracle]
\label{cor:scc}
For any $n$-vertex directed graph $G$, there exists an $f$-sensitivity oracle that, given any update set $U \subseteq E^f$, reports (with high probability) whether or not $G+U$ is $k$-strongly-connected.
The oracle can be constructed in $O((kn)^\omega)$ preprocessing time, uses $O((kn)^2 \log n)$ space, and answers each query in $O((fk^2)^\omega)$ time.
\end{corollary}

\cite{BaswanaCR19scc} presented an $f$-sensitivity oracle of $O(2^f n^2)$ space that reports the strongly connected components of a graph after $f$ failures in $\widetilde O(2^f n)$ query time. In contrast, for the simpler decision problem of testing whether the updated graph remains strongly connected ($k=1$ case), our oracle uses only $O(n^2\log n)$ space and answers queries in $O(f^\omega)$ time.

\subsection{Sensitivity Oracles for $k$-Disjoint Spanning Trees}

We now apply our algebraic framework to the problem of finding $k$ edge-disjoint spanning trees in an undirected graph. Let $G = (V, E, \cost)$ be an undirected graph with $n$ vertices, $m$ edges, and integer edge costs in $[-W, W]$. A spanning tree is a maximal acyclic subgraph of $G$, containing exactly $n-1$ edges. We can formulate the problem of finding $k$ edge-disjoint spanning trees as a matroid packing problem by utilizing the standard graphic matroid $M_{\mathrm{graph}}$ of rank $r = n-1$.

\begin{lemma} \label{lemma:k-spanning-trees-packing}
An undirected graph $G$ contains $k$ edge-disjoint spanning trees if and only if there exists a set $S \subseteq E$ of size $k(n-1)$ that can be partitioned into $k$ common bases of $M_1$ and $M_2$, where we set both $M_1$ and $M_2$ to $M_{\mathrm{graph}}$.
\end{lemma}

\begin{proof}
For the forward direction, suppose $G$ contains $k$ edge-disjoint spanning trees $T_1, \dots, T_k$. Let $S = \bigcup_{i=1}^k E(T_i)$. Because each $T_i$ is a spanning tree, its edge set $E(T_i)$ forms a basis in $M_{\mathrm{graph}}$. Since the trees are edge-disjoint, the sets $E(T_1), \dots, E(T_k)$ are pairwise disjoint and perfectly partition $S$. Because $M_1$ and $M_2$ are identical to $M_{\mathrm{graph}}$, each $E(T_i)$ is trivially a common basis of both matroids. Furthermore, the total size of $S$ is exactly $\sum_{i=1}^k (n-1) = k(n-1)$.

For the converse, suppose there exists a set $S \subseteq E$ of size $k(n-1)$ that can be partitioned into $k$ disjoint sets $B_1, \dots, B_k$, where each $B_i$ is a common basis of $M_1$ and $M_2$. Since $M_1 = M_{\mathrm{graph}}$, each set $B_i$ must be a basis of the graphic matroid, meaning the edges in $B_i$ form a spanning tree of $G$. Because the sets $B_1, \dots, B_k$ partition $S$, they share no edges and are pairwise disjoint. Therefore, the subgraphs induced by $B_1, \dots, B_k$ constitute exactly $k$ edge-disjoint spanning trees in $G$.
\end{proof}

\paragraph{Query Reduction}
An update set $U$ consists of at most $f$ operations (edge insertions, edge deletions, or cost changes). Because the columns of $A_{\mathrm{graph}}$ correspond directly to the edges of $G$, inserting, deleting, or re-weighting an edge modifies exactly one column in the representation matrix. Therefore, $f$ updates in $G$ translate directly to at most $\delta = f$ column updates. Unlike the arborescence problem, there is no dynamic root vertex, so no auxiliary query edges need to be processed.

\begin{theorem} \label{thm:sensitivity_k_spanning_trees}
Let $G = (V, E, \cost)$ be an $n$-vertex undirected graph with integer edge costs in $[-W, W]$, and let $k \ge 1$. There exist the following sensitivity oracles that, for any $U\in E^f$, answer queries on $G + U$ with high probability:
\begingroup
\setlength{\aboverulesep}{0pt}
\setlength{\belowrulesep}{0pt}
\begin{center}
\renewcommand{\arraystretch}{1.25}
\resizebox{0.98\linewidth}{!}{%
\begin{tabular}{@{}lllll@{}}
\rowcolor{yellow!15}
\textsc{Oracle} & \textsc{Prep.} & \textsc{Space} & \textsc{Query} & \textsc{Output} \\
\midrule
Decision &
$O((kn)^\omega)$ &
$O((kn)^2 \log n)$ &
$O((fk^2)^\omega)$ &
Whether $G + U$ has $k$ edge-disjoint spanning trees \\
Weighted &
$\widetilde{O}(W(kn)^3)$ &
$O(W(kn)^3 \log n)$ &
$\widetilde{O}(W kn (fk^2)^\omega)$ &
Min total cost of such $k$-spanning trees \\
\end{tabular}%
}
\end{center}
\endgroup
\end{theorem}

\begin{proof}
We invoke the matroid packing oracles from \Cref{thm:sensitivity_packing} on $M_1 = M_{\mathrm{graph}}$ and $M_2 = M_{\mathrm{graph}}$. Both matroids have rank $r = n-1$ and their representation matrices are structurally column-sparse.

When a query $U$ arrives, we incorporate the $f$ structural updates. Because each edge update modifies at most one column in the base representations, the total number of base updates is $\delta \le f$. Substituting $r = n-1$ and the update bound $\delta = f$ directly into the complexity bounds of \Cref{thm:sensitivity_packing} yields the stated preprocessing, space, and query bounds.
\end{proof}

\subsection{Sensitivity Oracles for Colorful Spanning Trees}

We now apply our algebraic matroid intersection framework ($1$-packing) to the colorful spanning tree problem. Let $G=(V, E, \cost)$ be an undirected graph with $n=|V|$ vertices and integer costs in $[-W, W]$. The edge set $E$ is partitioned into exactly $n-1$ disjoint color classes, $E = C_1 \uplus C_2 \uplus \dots \uplus C_{n-1}$. A \emph{colorful spanning tree} is a spanning tree of $G$ containing exactly one edge from each color class.

\begin{definition}[Graphic and Color Partition Matroids] \label{def:colorful-matroids}
The matroids $M_{\mathrm{graph}}$ and $M_{\mathrm{color}}$ on the common ground set $E$, both of rank $r=n-1$, are defined as follows:
\begin{itemize}
\item $M_{\mathrm{graph}}$: A subset $S \subseteq E$ is independent if its edges do not form any cycles in $G$. We fix an arbitrary orientation of $G$ and represent $M_{\mathrm{graph}}$ over $\F$ by its truncated incidence matrix $A_{\mathrm{graph}} \in \{-1,0,1\}^{(n-1)\times |E|}$.
\item $M_{\mathrm{color}}$: A subset $S \subseteq E$ is independent if it contains at most one edge from each color class $C_i$. Its linear representation over $\F$ is an $(n-1) \times |E|$ matrix $A_{\mathrm{color}}$, where rows are indexed by the $n-1$ color classes. For any edge $e \in E$, its corresponding column has a $1$ at the row corresponding to its color class, and $0$ elsewhere.
\end{itemize}
\end{definition}

An undirected graph $G$ contains a colorful spanning tree if and only if there exists a set $S \subseteq E$ of size $n-1$ that forms a $1$-packing of common bases in $M_{\mathrm{graph}}$ and $M_{\mathrm{color}}$.

\paragraph{Query Reduction}
An update set $U$ consists of at most $f$ operations (edge insertions, deletions, or cost/color changes). Because the columns of the representation matrices correspond directly to edges, modifying an edge or its color translates to updating exactly one column in the matrices (specifically, swapping the incident vertices in $A_{\mathrm{graph}}$ and moving the $1$ between color rows in $A_{\mathrm{color}}$). Thus, $f$ updates in $G$ translate directly to $\delta \le f$ column updates.

\begin{theorem} \label{thm:sensitivity_colorful}
Let $G=(V, E, \cost)$ be an undirected colored graph. Then there exist the following sensitivity oracles that, given any set of updates $U$ consisting of at most $f$ edge insertions, deletions, or color changes, answer queries on $G + U$ with high probability:
\begingroup
\setlength{\aboverulesep}{0pt}
\setlength{\belowrulesep}{0pt}
\begin{center}
\renewcommand{\arraystretch}{1.25}
\resizebox{0.98\linewidth}{!}{%
\begin{tabular}{@{}lllll@{}}
\rowcolor{yellow!15}
\textsc{Oracle} & \textsc{Prep.} & \textsc{Space} & \textsc{Query} & \textsc{Output} \\
\midrule
Decision &
$O(n^\omega)$ &
$O(n^2 \log n)$ &
$O(f^\omega)$ &
Whether $G + U$ has a colorful spanning tree \\
Weighted &
$\widetilde{O}(Wn^3)$ &
$O(Wn^3 \log n)$ &
$\widetilde{O}(W n f^\omega)$ &
Min cost of a colorful spanning tree \\
\end{tabular}%
}
\end{center}
\endgroup
\end{theorem}

\begin{proof}
We apply \Cref{thm:sensitivity_packing} with $k=1$. Both matroids have rank $r = n-1$ and their representation matrices are column-sparse. Substituting $k=1$, $r = n-1$, and the base update bound $\delta = f$ into \Cref{thm:sensitivity_packing} yields the stated preprocessing, space, and query bounds.
\end{proof}

\subsection{Sensitivity Oracles for Packing Spanning Forests}
We now apply our algebraic matroid covering framework to the problem of packing spanning forests. Let $G = (V, E)$ be an undirected graph with $n$ vertices and $m$ edges. The \emph{forest packing problem} (also known as the arboricity problem) asks whether the entire active edge set $E$ can be partitioned into $k$ disjoint forests.

To utilize our matroid covering framework, we establish the direct equivalence between partitioning the edge set into $k$ forests and covering the ground set with $k$ independent sets of $M_{\mathrm{graph}}$.

\begin{lemma} \label{lemma:forest-covering-equivalence}
The edges of an undirected graph $G=(V, E)$ can be partitioned into $k$ forests if and only if the edge set $E$ is independent in the $k$-fold union matroid $M_{\mathrm{graph}}^{\vee k}$.
\end{lemma}

\begin{proof}
By definition, the $k$-fold union matroid $M_{\mathrm{graph}}^{\vee k}$ defines a set $S \subseteq E$ as independent if and only if $S$ can be partitioned into $k$ disjoint sets $S_1 \uplus \dots \uplus S_k$ such that each $S_i$ is independent in the base matroid $M_{\mathrm{graph}}$. Since independence in $M_{\mathrm{graph}}$ means the edge set is a forest, setting $S = E$ implies that the entire graph can be partitioned into $k$ disjoint forests if and only if $E$ is an independent set in $M_{\mathrm{graph}}^{\vee k}$.
\end{proof}

Note that a necessary condition for a graph to be partitionable into $k$ forests is that the total number of edges satisfies $m \le k(n-1)$. If the number of active edges ever exceeds $k(n-1)$, the graph trivially cannot be partitioned into $k$ forests.

\paragraph{Query Reduction}
An update set $U$ consists of at most $f$ operations (edge insertions or deletions). Because the columns of $A_{\mathrm{graph}}$ correspond directly to the edges of $G$, inserting or deleting an edge modifies exactly one column in the representation matrix. Therefore, $f$ updates in $G$ translate to exactly $\delta \le f$ column updates.

\begin{theorem} \label{thm:sensitivity_k_forests}
Let $G = (V, E)$ be an undirected graph with $n$ vertices, and $k \ge 1$. There exists a sensitivity oracle that, given any set of updates $U$ consisting of at most $f$ edge insertions or deletions (such that the updated active edge count $m_{curr}$ remains bounded by $k(n-1)$), answers queries on $G + U$ with high probability:
\begingroup
\setlength{\aboverulesep}{0pt}
\setlength{\belowrulesep}{0pt}
\begin{center}
\renewcommand{\arraystretch}{1.25}
\resizebox{0.92\linewidth}{!}{%
\begin{tabular}{@{}lllll@{}}
\rowcolor{yellow!15}
\textsc{Oracle} & \textsc{Prep.} & \textsc{Space} & \textsc{Query} & \textsc{Output} \\
\midrule
Decision &
$O((kn)^\omega)$ &
$O((kn)^2 \log n)$ &
$O((kf)^\omega)$ &
Whether $G + U$ partitions into $k$ forests \\
\end{tabular}%
}
\end{center}
\endgroup
\end{theorem}

\begin{proof}
We invoke the matroid covering oracle from \Cref{thm:sensitivity_covering} on the graphic matroid $M_{\mathrm{graph}}$. The rank of the matroid is $r = n-1$.

When a query $U$ arrives, we incorporate the $f$ structural edge updates. Because each edge update inserts or deletes exactly one column in the base representation matrix $A_{\mathrm{graph}}$, the total number of column updates is $\delta \le f$. We substitute the rank $r = n-1$ and the update bound $\delta = f$ directly into the complexity bounds of \Cref{thm:sensitivity_covering}. The preprocessing time is $O((k(n-1))^\omega) = O((kn)^\omega)$, the space is $O((kn)^2 \log n)$, and the query time required to re-evaluate the determinant of the augmented matrix $U'_{aug}$ is $O((\delta k)^\omega) = O((fk)^\omega)$. By \Cref{lemma:forest-covering-equivalence}, this determinant is non-zero if and only if the updated graph can be partitioned into $k$ forests.
\end{proof}

\section{Sensitivity Oracles for Matroid Parity}
\label{sec:matroid_parity}

In this section, we study the \emph{Matroid Parity} (also known as \emph{Matroid Matching}) problem, for linear matroids. The matroid parity problem was introduced by Lawler~\cite{Lawler76} as a common generalization of two fundamental combinatorial optimization problems: Matroid Intersection and Graph Matching. Although the matroid parity problem is intractable for general matroids~\cite{JensenK82, Lovasz80}, it admits a polynomial time algorithm if the underlying matroid is linear~\cite{Lovasz80}. 

\subsection{Algebraic Formulation and Enumeration}

Let $M$ be a linear matroid of (even) rank $r$ over a field $\F$, defined on a ground set of $m$ disjoint pairs of elements $E = \{(a_1, b_1), \dots, (a_m, b_m)\}$, and let $\wt : E \to [-W, W]$ be a weight function (the weights are associated with the pairs in $E$). The Linear Matroid Parity (LMP) problem asks whether there exists a selection of exactly $r/2$ pairs such that the union of all $r$ selected elements forms a basis of $M$. Such a valid selection is called a \emph{parity basis}. The \emph{weighted} version asks for the minimum total weight of a parity basis, if one exists.   

Let $A, B \in \F^{r \times m}$ be the representation matrices where the $i$-th column of $A$ corresponds to $a_i$ and the $i$-th column of $B$ corresponds to $b_i$. Let $D \in \F_{2W}[Y]^{m \times m}$ be a diagonal matrix with $D_{i,i} = x_i Y^{W + \wt(i)}$, where $x_i$ are independent uniformly random variables. 

We define the $r \times r$ skew-symmetric LMP matrix as:
$$
\mathcal{M}_{LMP} := A D B^T - B D A^T
$$

\begin{theorem}[Enumerating Parity Bases~\cite{Lovasz80, CameriniGM92}]
\label{thm:parity_enumeration}
If the linear matroid $M$ contains at least one parity basis, then the determinant of $\mathcal{M}_{LMP}$ enumerates the parity bases of $M$. Specifically,
$$
\det(\mathcal{M}_{LMP}) = Y^{rW} \left( \sum_{\substack{S \subseteq [m] \\ |S| = r/2}} \det(M_S) \Bigl( \prod_{i \in S} x_i \Bigr) Y^{\wt(S)} \right)^2
$$
where $M_S$ is the $r \times r$ submatrix formed by interleaving the columns $a_i$ and $b_i$ for all $i \in S$. Furthermore, with high probability, the minimum degree of $Y$ in $\det(\mathcal{M}_{LMP})$ minus $rW$ equals exactly twice the weight of the minimum-weight parity basis.
\end{theorem}

\begin{proof}
Because $\mathcal{M}_{LMP}$ is a skew-symmetric matrix of even dimension $r$, its determinant is the square of a polynomial known as its Pfaffian, i.e., $\det(\mathcal{M}_{LMP}) = \text{Pf}(\mathcal{M}_{LMP})^2$.

The Pfaffian of the matrix $\sum_{i=1}^m x_i Y^{W + \wt(i)} (a_i b_i^T - b_i a_i^T)$ can be explicitly expanded as the sum of the determinants of all $r \times r$ minors generated by subsets of size $r/2$. Specifically, 
$$
\text{Pf}(\mathcal{M}_{LMP}) = \sum_{\substack{S \subseteq [m] \\ |S| = r/2}} \det(M_S) \prod_{i \in S} \left( x_i Y^{W + \wt(a_i,b_i)} \right)
$$
where $M_S$ is the $r \times r$ matrix whose columns are exactly the pairs $\{a_i, b_i\}$ for $i \in S$. 

Observe that $\det(M_S) \neq 0$ if and only if the columns of $M_S$ are linearly independent over $\F$. Because $M_S$ contains exactly $r$ columns corresponding to $r/2$ pairs, this condition holds if and only if $S$ is a valid parity basis. Factoring out the base weight $W$ from the exponent yields:
$$
\text{Pf}(\mathcal{M}_{LMP}) = Y^{\frac{r}{2}W} \sum_{\substack{S \text{ is a} \\ \text{parity basis}}} \pm\det(M_S) \Bigl( \prod_{i \in S} x_i \Bigr) Y^{\wt(S)}
$$
Squaring the Pfaffian to obtain the determinant directly produces the formula stated in the theorem. 

Let $S^*$ be an optimal basis. 
Then the coefficient of $Y^{\frac{r}{2} W + \wt(S^*)}$ in $\mathrm{Pf}(\mathcal{M}_{LMP})$
is a polynomial of degree $r/2$ in the variables $x_i$; it is not identically
zero since the monomials $\prod_{i\in S} x_i$ of distinct minimum-weight parity
bases $S$ are distinct. By the Schwartz--Zippel lemma, its evaluation at independent uniform $x_i$ from a field of size $\Omega(\mathrm{poly}(m))$
is nonzero with high probability.
Since the Pfaffian is squared, the term corresponding to  optimal basis contributes a $Y$-exponent of $2(\frac{r}{2}W + \wt(S^*)) = rW + 2\wt(S^*)$. Therefore, isolating the smallest exponent of $Y$ in $\det(\mathcal{M}_{LMP})$ and subtracting $rW$ yields exactly $2\wt(S^*)$.
\end{proof}

\subsection{Sensitivity Oracles for Matroid Parity}

\begin{theorem}[Matroid Parity General Sensitivity Oracle]
\label{thm:parity_general_sensitivity}
Let $E$ be a set of $m$ pairs spanning a $c$-sparse linear matroid of rank $r$. There exist randomized sensitivity oracles that, given $f$ pair updates, answer parity basis queries with high probability:
\begin{enumerate}
\item \textbf{Decision Oracle:} Decides if a parity basis exists. This oracle requires $O(c^2m + r^\omega)$ preprocessing time, $O(r^2 \log (r))$ space, and answers queries in $O((c^2f)^\omega)$ time.
\item \textbf{Weighted Oracle:} Returns the exact minimum weight of a parity basis. This oracle requires $\widetilde{O}(c^2m + W r^3)$ preprocessing time, $O(W r^3 \log r)$ space, and answers queries in $\widetilde{O}(W \cdot r \cdot (c^2f)^\omega)$ time.
\end{enumerate}
\end{theorem}

\begin{proof}
The proof is similar to that of~\Cref{thm:sensitivity_packing} and hence, we discuss it in brief. We first apply the preprocessing algorithm from \Cref{lemma:BrandS-det-General-Q} to the $r \times r$ Parity Laplacian. 

For the unweighted decision oracle, we evaluate the matrix at $Y=1$ (maximum degree $d=0$). The preprocessing requires $O(c^2m + r^\omega)$ time and $O(r^2 \log r)$ space. For the weighted oracle, the matrix entries are polynomials of maximum degree $d = O(W)$, requiring $\widetilde{O}(c^2m + W r^3)$ preprocessing time and $O(W r^3 \log r)$ space. 

When a query containing $f$ pair updates arrives, since the matroid is $c$-sparse, the modifications translate to an additive update matrix containing at most $2c^2f$ non-zero entries. By \Cref{lemma:BrandS-det-General-Q}, re-evaluating the determinant takes $O((2c^2f)^\omega) = O((c^2f)^\omega)$ time for the scalar case, and $\widetilde{O}(W r (2c^2f)^\omega) = \widetilde{O}(W r \cdot (c^2f)^\omega)$ time for the polynomial case.
\end{proof}


\section{$\alpha$-Factors and Matchings}
\label{sec:alpha-factors}

Let $G = (V, E, \cost)$ be an undirected graph with $n$ vertices and $m$ edges, where edge costs are integers in the range $[-W, W]$. Let $\alpha: V \to \mathbb{Z}^+$ be a function assigning a degree bound to each vertex in $V$. Following the terminology of~\cite{GabowS21_II}, we define a (perfect) {\bf $\alpha$-factor} as a subgraph $H$ of $G$ such that every vertex $v \in V$ has degree exactly $\alpha(v)$ in $H$.
\footnote{Note that our results carry over to (maximum) $\alpha$-factor as well. The (maximum) $\alpha$-factor considers $\alpha(v)$ as an upper bound on the desired degree of $v$ and the objective is to maximise the number of edges picked. We provide its details in the appendix.}

For convenience, we only consider the (perfect) $\alpha$-factor and refer to it throughout as just $\alpha$-factor. The \emph{minimum-cost} $\alpha$-factor is the $\alpha$-factor with minimum total edge cost, if it exists. 
Throughout this section, we refer to the set of edges incident to a vertex $v$ in graph $G$ as $\delta_G(v)$. We omit the subscript $G$ if the graph $G$ is clear from the context.

\subsection{$\alpha$-Factors as Linear Matroid Parity}

In this subsection, we study the \textit{$k$-bounded $\alpha$-factor problem}, where we assume that the maximum degree constraint is bounded by a parameter $k$, i.e., $\max_{v \in V} \alpha(v) \le k$. 
The \textit{$k$-bounded $\alpha$-factor problem} can be formulated as a linear matroid parity problem over a $k$-sparse linear matroid $M$ having $m$ disjoint pairs of elements. Let $N = \sum_{v \in V} \alpha(v)$ denote the sum of all degree constraints. We assume that $N \leq 2m$ and is even, otherwise an $\alpha$-factor does not exist. 

Furthermore, we assume that the vertex set $V$ has some fixed total ordering, i.e., for any two vertices $a,b \in V$, either $a < b$ or $a > b$.
Since the graph is undirected, we assume that every edge is written as an ordered pair according to the total ordering over $V$, i.e., an edge $e$ in $G$ is written as $(a,b)$, where $a < b$, and $(b, a)$ otherwise.   

Define $\V$ as a set where each vertex $v \in V$ has $\alpha(v)$ distinct copies, and $v_j \ (j \in [\alpha(v)])$ refers to the $j^{th}$ copy of vertex $v$.

Now, consider the following matrix representation matrix $A$ for a matroid $M$: $A$ is a $N \times 2m$ polynomial matrix over the variables $\{x_{u_k, e}\}_{u_k \in \V, e \in \delta(u) }$   whose rows are indexed by $\V$ and the column pairs by $E$. For $e = (u, v)$, the two columns corresponding to $e$ are $(e_{u}, e_{v})$ respectively. For every edge $e = (u, v) \in E$, set the following entries in $A$:

$$A[u_k, e_{u}] = x_{u_k, e}\ ; \  A[v_l, e_{v}] = x_{v_l, e} \ \forall k \in [\alpha(u)], l \in [\alpha(v)].$$
Rest of the entries are set to 0. It can be seen that there are no repeated non-zero entries in $A$.
Now, we show the following:

\begin{lemma}
    \label{lemma:factor-to-parity}
    A graph $G$ contains an $\alpha$-factor if and only if the linear matroid $M$ represented by $A$ admits a parity base. Furthermore, there is a bijection between the $\alpha$-factors of $G$ and the parity bases of $M$.
\end{lemma}
\begin{proof}
We first prove the forward direction. Let $H \subseteq E$ be an $\alpha$-factor of $G$. It can be seen that $|H| = N/2$ and hence, $H$ induces an $N \times N$ submatrix $A[H]$, formed by selecting the column pairs of $A$ corresponding to $H$. Since $|\delta_H(v)| = \alpha(v)$ for all $v \in V$, we can define a perfect matching $\pi$ between the set $\V$ and the columns of $H$ by mapping each copy $v_j \in \V$ to a unique column $e_v$ for some $e \in H \cap \delta_H(v)$.
Since each indeterminate appears exactly once in the entire matrix $A$, the monomial $\prod_{v_j \in \V}A_{v_j, \pi(v_j)}$ in the Leibniz expansion of $\det(A[H])$ does not get cancelled. Thus, $\det(A[H]) \neq 0$, meaning the pairs corresponding to $H$ form a parity base.
    
For the other direction, assume there exists a parity base of $M$ corresponding to the edge set $H \subseteq E$. Note that since $\text{rank}(M) = N$, we have that $|H| = N / 2$. We now show by contradiction that for all $v \in V$, $|\delta_H(v)| = \alpha(v)$. Suppose this is not true. Then, there exists a vertex $w \in V$ such that $|\delta_H(w)| \neq \alpha(w)$. Since $|H| = N/2$ and $\sum_{v \in V} \alpha(v) = N$, we can assume without loss of generality that $|\delta_H(w)| > \alpha(w)$. Define the set $C$ as the collection of columns of $A$ indexed by the set of edges $\{e_w \mid e \in \delta_H(w)\}$. Since $|\delta_H(w)| > \alpha(w)$, we have that $|C| > \alpha(w)$. By the definition of $A$, there exist at most $\alpha(w)$ rows (corresponding to the copies $\{w_1, w_2, \dots, w_{\alpha(w)}\}$) that have non-zero entries for the column vectors in $C$. Hence, the set of columns in $C$ can span a subspace of rank at most $\alpha(w)$. Since $|C| > \alpha(w)$, the columns in $C$ are linearly dependent. This implies $A[H]$ is singular, and therefore the pairs corresponding to $H$ do not form a parity basis, which is a contradiction.
\end{proof}

Since the edges of $G$ correspond to the pairs in $M$, the following corollary of~\Cref{lemma:factor-to-parity} is immediate in the weighted setting:

\begin{corollary}
    \label{corollary:factor-to-parity}
    The cost of the minimum cost $\alpha$-factor of $G$ is equal to the weight of minimum-weight parity basis of $M$, if they exist.
\end{corollary}

\subsection{Sensitivity Oracles for $\alpha$-Factors and Matchings}
\label{sec:sensitivity-alpha-factors}

In this subsection we extend the algebraic formulation of $\alpha$-factors to support edge updates. For any set of $f$ edge insertions, deletions, or cost changes, the oracle reports whether an $\alpha$-factor exists in the updated graph and, in the weighted setting, the weight of a minimum-cost one. We consider two cases. First, in \Cref{sec:bounded-alpha}, we assume the degree bounds are $k$-bounded, i.e.\ $\max_v \alpha(v) \le k$, and work directly with the sensitivity oracle for $c$-sparse linear matroid parity. Then, in \Cref{sec:arbitrary-alpha}, we remove this assumption and handle arbitrary $\alpha$ by reducing an update on $G$ to an update on an \emph{auxiliary graph} whose size is governed by $f$ rather than by $\max_v \alpha(v)$.
Throughout, we let $N = \sum_v \alpha(v)$. Furthermore, we instantiate the indeterminates in linear representation of $\alpha$-factor as uniformly random samples from a sufficiently large finite field, and our results would hold with high probability due to the Schwartz-Zippel lemma.

\subsubsection{Warm Up: The $k$-Bounded Case}
\label{sec:bounded-alpha}

From~\Cref{corollary:factor-to-parity} and \Cref{thm:parity_general_sensitivity}, the following result is immediate:

\begin{theorem}
\label{thm:sensitivity-bounded-alpha}
Let $G=(V,E,\cost)$ be an undirected graph with integer edge costs in $[-W,W]$ and degree bounds $\alpha \le k$. Then there exist the following sensitivity oracles that, for any update set $U$ with $f = |U|$, answer queries on $G+U$ with high probability:

\begingroup
\setlength{\aboverulesep}{0pt}
\setlength{\belowrulesep}{0pt}
\begin{center}
\renewcommand{\arraystretch}{1.25}
\resizebox{0.98\linewidth}{!}{%
\begin{tabular}{@{}l l l l l@{}}
\rowcolor{yellow!15}
\textsc{Oracle} & \textsc{Prep.} & \textsc{Space} & \textsc{Query} & \textsc{Output} \\
\midrule
Existence &
$O((kn)^\omega)$ &
$O((kn)^2 \log n)$ &
$O(f^\omega k^{2\omega})$ &
Whether $G+U$ has an $\alpha$-factor \\
Min-Cost &
$\widetilde{O}(W(kn)^3)$ &
$O(W(kn)^3 \log n)$ &
$\widetilde{O}(W n f^\omega k^{2\omega+1})$ &
Weight of a min-cost $\alpha$-factor of $G+U$ \\
\end{tabular}%
}
\end{center}
\endgroup
\end{theorem}

\begin{proof}
\emph{Existence.} We preprocess $M|_{Y=1}$ with \Cref{lemma:BrandS-det-General-Q} in the scalar case ($d=0$). Since $N \le nk$, this requires $O((kn)^\omega)$ processing time. For a query set $U$, $C$ has $\delta = O(fk^2)$ non-zero entries. By applying \Cref{lemma:BrandS-det-General-Q} to compute the determinant of $M' = M+C$, the query time is $O(\delta^\omega) = O(f^\omega k^{2\omega})$.

\emph{Min-cost.} We apply \Cref{lemma:BrandS-det-General-Q} to the polynomial matrix $M$ with $d = O(W)$. The query $M'=M+C$ involves computing the polynomial determinant via \Cref{lemma:BrandS-det-General-Q}, which takes $\widetilde{O}(d N \delta^\omega)$ operations. Substituting $d = O(W)$, $N \le nk$, and $\delta = O(fk^2)$, we obtain:
$$
\widetilde{O}\bigl(W \cdot (kn) \cdot (fk^2)^\omega\bigr) = \widetilde{O}(W n f^\omega k^{2\omega+1}).
$$
Extracting the minimal exponent as in~\Cref{thm:parity_general_sensitivity} yields the weight.
\end{proof}

\subsubsection{Arbitrary Degree Bounds via Auxiliary Graphs}
\label{sec:arbitrary-alpha}

We now drop the assumption $\alpha \le k$. A direct expansion into $\alpha(v)$ copies would yield a Tutte matrix of dimension $O(n \max_v \alpha(v))$, which may be prohibitively large. Instead, we will fix a precomputed min-cost $\alpha$-factor $M$ of $G$ (we assume such an $\alpha$-factor exists in $G$, this assumption is justified as ideas of~\Cref{subsec:max-alpha-factor} can be used to augment $G$ such that it always admits an $\alpha$-factor) and build an \emph{auxiliary graph} $H$ that captures the symmetric difference $M \triangle M'$ between $M$ and an updated $\alpha$-factor $M'$ of $G+U$.

For this, we first introduce some definitions. A trail $\cal T$ is called a closed trail if $\cal T$ starts and ends at the same vertex. 
A sub-trail $\cal T'$ of $\cal T$ is a trail which uses only the edges of $\cal T$, with no repetitions. Note that $\cal T'$ need not have the same order of walk as $\cal T$, just that all the edges used by $\cal T'$ should lie in $\cal T$.   
We now define the notion of closed alternating trails. A closed alternating trail ${\cal T} = (e_1, e_2, \dots , e_{2k}), k \in \mathbb{N}$ is defined as a closed trail of even length such that if $e_i \in M$, then $e_{i+1} \in M'$ (and vice versa). It can also be seen that $e_1 \in M \iff e_{2k} \in M'$. Since every vertex has both $M$-degree and $M'$-degree equal to $\alpha(v)$, the following observation is immediate.

\begin{observation}
    \label{obs:trail-decomposition}
    The edge set $M \triangle M'$ decomposes into edge-disjoint closed alternating trails.
\end{observation}

The following lemma proves that for any update set $U$, there exists an $\alpha$-factor
 $M'$ in $G+U$ such that the number of closed alternating trails in $M \triangle M'$ is bounded by at most $|U|$.

\begin{lemma}\label{lem:cycle-free-witness}
If $G+U$ admits an $\alpha$-factor, then it admits a min-cost $\alpha$-factor $M'$ such that $D := M \triangle M'$ contains no closed alternating trail $\cal T$ all of whose edges belong to both $G$ and $G+U$.
\end{lemma}

\begin{proof}

Among all min-cost $\alpha$-factors of $G+U$, choose $M'$ minimizing $|M \triangle M'|$. By \Cref{obs:trail-decomposition}, $D$ decomposes into edge-disjoint closed alternating trails.
Suppose $D$ contains an alternating trail $\cal T$ whose edges lie in both $G$ and $G+U$. Toggling it yields $M'' := M' \triangle {\cal T}$, which preserves every vertex degree and is therefore again an $\alpha$-factor of $G+U$. Moreover, ${\cal T} \subseteq G+U$ ensures that $M''$ is feasible in $G+U$. Since $M'$ is min-cost $\alpha$-factor of
$G+U$, $\cost(M'') \ge \cost(M')$, meaning $\cost({\cal T} \cap M') \le \cost({\cal T} \cap M)$. Similarly, ${\cal T} \subseteq G$ implies that $M \triangle {\cal T}$ is an $\alpha$-factor of $G$. Since $M$ is min-cost $\alpha$-factor of $G$, we get that $\cost({\cal T} \cap M) \le \cost({\cal T} \cap M')$.
The two inequalities together imply $\cost({\cal T} \cap M') = \cost({\cal T} \cap M)$ and hence, $\cost(M'') = \cost(M')$. Thus, $M''$ is a min-cost $\alpha$-factor of $G+U$ such that $|M \triangle M''| < |M \triangle M'|$, contradicting the minimality of $M'$. This proves the claim.
\end{proof}

The following claim gives the last structural observation over $M \triangle M'$ needed for our oracle.

\begin{lemma}
\label{lemma:trail-incidence}
Consider any minimal closed alternating trail $\cal T$ in $M \triangle M'$, i.e., no sub-trail of $\cal T$ is a closed alternating trail. Then, $\cal T$ is incident to any vertex of $G$ through at most two edges of $M$ and two edges of $M'$.
\end{lemma}
\begin{proof}
We proceed by contradiction. Suppose there exists a minimal closed alternating trail $\cal T$ that is incident on a vertex $v$ through at least 6 edges (at least 3 from $M$ and 3 from $M'$). Because the subgraph induced by $\cal T$ is Eulerian, any vertex with a degree greater than 2 in this subgraph can be split to decompose the trail into smaller, edge-disjoint closed trails. Consequently, we can extract a smaller sub-trail $\cal T'$ of $\cal T$ that passes through exactly 2 edges of $M$ and 2 edges of $M'$ incident to $v$. This sub-trail $\cal T'$ forms a valid closed alternating trail on its own, which directly contradicts the minimality of $\cal T$. Hence, the claim holds. 
\end{proof}

Now, we bound the number of edges of $M \triangle M'$ incident on any vertex by analysing the closed alternating trails induced by the edges $M \triangle M'$. It follows that if $M'$ is the optimal witness of \Cref{lem:cycle-free-witness}, then $M \triangle M'$ can be decomposed into at most $f$ closed alternating trails. Consider the decomposition $S^* = \{ {\cal T}_1, {\cal T}_2, \dots, {\cal T}_j \}, j \leq f$ such that the number of closed alternating trails is maximised. Then, for each ${\cal T}_i$, no closed alternating sub-trail exists. From \Cref{lemma:trail-incidence}, the number of $M$-edges dropped at any vertex $v$ due to each ${\cal T}_i$ is at most 2. Combining this over all trails in $S^*$, we get that $|\delta_G(v) \cap M \cap (M \triangle M')| = |\delta_{G+U}(v) \cap M' \cap (M \triangle M')| \leq 2f$, i.e., the total number of $M$-edges dropped at any vertex $v$ is bounded by $2f$. We now leverage this local bound to construct the auxiliary graph $H$.

\paragraph{Construction of the auxiliary graph.}
Let $M$ be the fixed initial $\alpha$-factor of the graph $G = (V, E)$ and define $f_v := \min(2f, \alpha(v))$, for $v\in V$. We construct an auxiliary graph $H = (V_1 \cup V_2 \cup R \cup S, E_H)$, and define its degree bounds $\beta$ as follows:
\begin{itemize}
    \item \textbf{Layer 1 (drop layer):}
    Make a copy $V_1$ of the vertex set $V$, such that for any vertex $v \in V$, $v_1$ denotes its copy in the set $V_1$. Now, for each edge $e = (a,b) \in M$, add the edge $(a_1, b_1)$ with cost $-\cost(e)$, and set $\beta(v_1) = f_v$. (Selecting this edge in a $\beta$-factor of $H$ implies $e$ is \emph{dropped} from $M'$).
    
    \item \textbf{Layer 2 (add layer):} 
    Make another copy $V_2$ of the vertex set $V$, such that for any vertex $v \in V$, $v_2$ denotes its copy in the set $V_2$.
    Now, for each edge $e = (a,b)$ of $G$ such that $e \notin M$, add the edge $(a_2, b_2)$ with cost $+\cost(e)$, and set $\beta(v_2) = f_v$. (Selecting this edge implies $e$ is \emph{added} to $M'$).
    
    \item \textbf{Parallel edges:} For each $v \in V$, add $f_v$ parallel edges from $v_1$ to $v_2$, each with cost $0$.

    \item \textbf{Budget sinks $(S, R)$:} For each $v \in V$, add two vertices $s_{v}, r_{v}$ with degree bounds $\beta(s_{v}) = \beta(r_{v}) = f_v$. 
    Furthermore, add $f_v$ parallel edges between $s_{v}$ and $r_{v}$, all of cost $0$. 

\end{itemize}

Because $\beta \le 2f$ everywhere, the corresponding sensitivity oracle matrix for $H$ has a bounded dimension $N_H = O(nf)$, independent of the maximum capacity $\alpha$.

\paragraph{Encoding the update $U_H$.}
Given an update set $U \subseteq E$ with $|U| = f$, we explicitly translate it into $U_H$ as below (see \Cref{fig:factor-reduction}).
\begin{itemize}
    \item \textbf{For $(a,b) \notin M$ (added/removed):} Apply the identical operation to $(a_2, b_2)$ in Layer 2.
    
    \item \textbf{For $(a,b) \in M$ (deleted):} Remove $(a_1, b_1)$ from Layer 1; remove one edge each of the following forms - $(a_1, a_2), (b_1, b_2), (r_{a}, s_{a}), (r_{b}, s_{b})$; and finally add the following edges with zero cost - $(r_{a}, r_{b}), (s_{a}, a_1), (s_{b}, b_1)$. 
\end{itemize}

\begin{figure}
\centering
\includegraphics[width=0.3\linewidth]{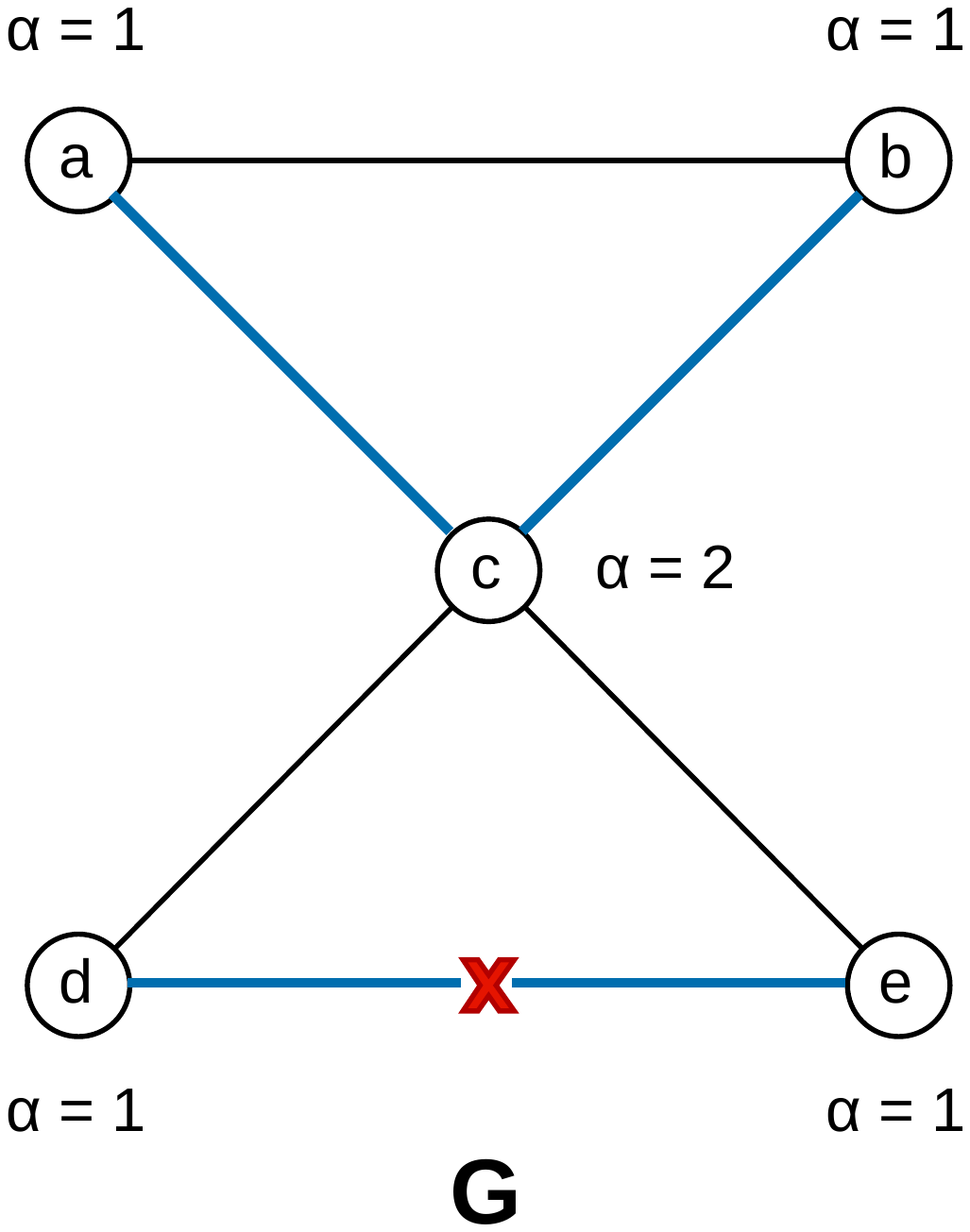}
\includegraphics[width=0.68\linewidth]{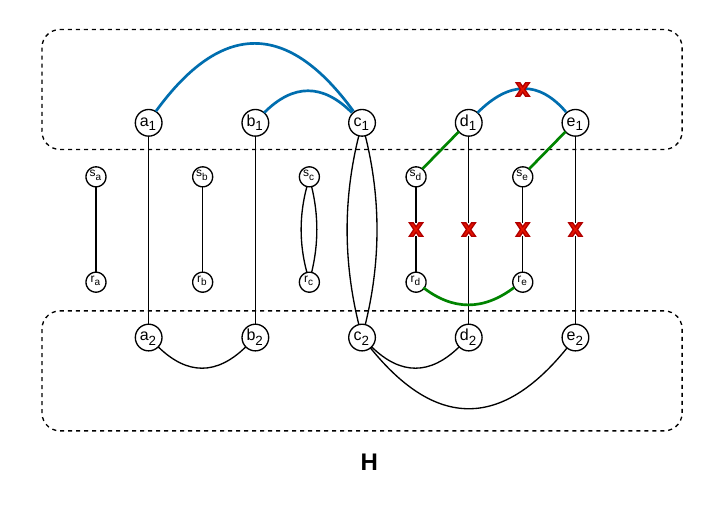}
\caption{The reduction for $f=1$: deleting $(d,e)$ in $G$ (with $\alpha(a)=\alpha(b)=\alpha(e)=\alpha(d)=1$ and $\alpha(c)=2$) induces five deletions and three additions (shown in green) in $H$; the layer-1 edges are toggled out and layer-2 edges are toggled into the $\beta$-factor in $H+U_H$, resulting in the $\alpha$-factor of $G + U$ consisting of the edges $\{(a,b), (c,d), (c,e)\}$.
}
\label{fig:factor-reduction}
\end{figure}

\paragraph{The projection.}
Given a valid $\beta$-factor $F$ of $H + U_H$, we project it back to $G$ by defining $M' := (M \setminus \mathrm{D}) \cup \mathrm{A}$, where $\mathrm{A}$ is the set of non-$M$ edges whose Layer 2 copy is selected, and $\mathrm{D}$ is the set of base $M$-edges that are \emph{either} selected in Layer 1 or deleted by $U$. 

\begin{theorem}\label{thm:bijection}
$G+U$ contains an $\alpha$-factor if and only if $H+U_H$ contains a $\beta$-factor. Moreover, under the costs assigned above,
\begin{equation}
    \cost(M') \;=\; \cost(M) + \cost_H(F) \;\; - \!\!\sum_{e \in U \cap M,\ \text{deleted}}\!\! \cost(e).
\end{equation}
Consequently, a min-cost $\alpha$-factor of $G+U$ corresponds exactly to a min-cost $\beta$-factor of $H+U_H$, and (ignoring costs) existence transfers bi-directionally.
\end{theorem}

\begin{proof}
For any vertex $v$, let $\mathrm{del}(v)$ denote the number of $M$-edges deleted at $v$. The $U_H$ encoding activates exactly $\mathrm{del}(v)$ sink pairs at $v_1$ by altering the edges incident to $r_v$ and $s_v$.

$(\Leftarrow)$ Let $F$ be a valid $\beta$-factor in $H+U_H$. 
Observe that the update $U_H$ deletes $\mathrm{del}(v)$ edges of the form $(s_v, r_v)$ and adds $\mathrm{del}(v)$ zero-cost edges connecting $r_v$ to other $r$-vertices. Thus, $r_v$ maintains an exact available degree of $f_v$. Because $\beta(r_v) = f_v$, the $\beta$-factor $F$ must select all incident edges of $r_v$, including the $f_v - \mathrm{del}(v)$ surviving parallel edges to $s_v$. Since $s_v$ also has a budget of $f_v$, it is forced to satisfy the remainder of its budget by selecting all $\mathrm{del}(v)$ newly added $(s_v, v_1)$ edges. 

Let $d(v)$ and $p(v)$ be the number of Layer 1 and parallel edges $F$ selects at $v_1$, respectively. The budget at $v_1$ mandates that $d(v) + \mathrm{del}(v) + p(v) = f_v$. At $v_2$, selecting $a(v)$ Layer 2 edges alongside the shared $p(v)$ parallel edges yields $a(v) + p(v) = f_v$. Subtracting these two equations gives $a(v) = d(v) + \mathrm{del}(v)$. 

In our projection to $M'$, the total number of $M$-edges dropped at $v$ sums to $d(v) + \mathrm{del}(v)$ (Layer 1 selections plus explicitly deleted edges). This perfectly balances the $a(v)$ edges added from Layer 2. Therefore, $\deg_{M'}(v) = \deg_M(v) = \alpha(v)$. Furthermore, $M'$ is a simple graph because base edges  partition across the layers, preventing simultaneous addition and dropping of the same edge. Finally, $M' \subseteq G+U$ since removed edges are explicitly absent from Layer 2, and deleted $M$-edges are forcibly dropped. Thus, $M'$ is a valid $\alpha$-factor.

$(\Rightarrow)$ Let $M'$ be the optimal witness of \Cref{lem:cycle-free-witness}, and let $d(v)$ denote the number of \emph{non-deleted} $M$-edges dropped at $v$. By the construction of $M'$ and the bounds on the symmetric difference, we know that for all $v \in V$, $d(v) + \mathrm{del}(v) \le f_v$. 

We construct a $\beta$-factor $F$ in $H+U_H$ by selecting: the surviving Layer 1 edges of $M \setminus M'$; the Layer 2 edges of $M' \setminus M$; all incident edges to $r_v$ and $s_v$; and exactly $f_v - d(v) - \mathrm{del}(v)$ parallel edges at each $v$ (which is always valid since $d(v) + \mathrm{del}(v) \le f_v$). Every vertex thus meets its exact budget constraint $\beta$. Since all explicitly deleted $M$-edges correctly reside in $M \setminus M'$ and are accounted for by the sinks, $F$ is a valid $\beta$-factor.

\emph{Cost Analysis.} Only Layer 1 and Layer 2 edges carry non-zero costs. A Layer 2 selection adds $+\cost(e)$, while a Layer 1 selection contributes $-\cost(e)$  for \emph{non-deleted} dropped edges. Hence, the cost of $F$ is:
$$
\cost_H(F) = \sum_{e \in M' \setminus M} \cost(e) \;-\sum_{\substack{e \in M \setminus M' \\ \text{not deleted}}} \cost(e)
\;=\; \cost(M') - \cost(M) + \sum_{e \in U \cap M,\,\text{deleted}} \cost(e).
$$
Rearranging this equation directly yields the cost identity in the theorem. Therefore, minimizing $\cost_H(F)$ minimizes $\cost(M')$. In the unit-cost (existence) setting, all cost labels are discarded, and the exact same structural argument provides the existence equivalence.
\end{proof}

Because $H$ restricts degree bounds to $\beta \leq 2f$, the corresponding sensitivity oracle matrix dimension is $N_H = O(nf)$. Furthermore, the update sparsity is tightly bounded at $|U_H| = O(f)$. Now, we invoke the $k$-bounded oracle of \Cref{sec:bounded-alpha} on $H$, where $k = O(f)$. This gives the following.

\begin{theorem}\label{thm:sensitivity-arbitrary-alpha}
Let $G = (V,E,\cost)$ be an undirected graph with integer edge costs in $[-W,W]$ and arbitrary degree bounds $\alpha$. Then there exist the following sensitivity oracles that, for any update set $U$ with $f = |U|$, answer queries on $G+U$ with high probability:
\begingroup
\setlength{\aboverulesep}{0pt}
\setlength{\belowrulesep}{0pt}
\begin{center}
\renewcommand{\arraystretch}{1.25}
\resizebox{0.98\linewidth}{!}{%
\begin{tabular}{@{}l l l l l@{}}
\rowcolor{yellow!15}
\textsc{Oracle} & \textsc{Prep.} & \textsc{Space} & \textsc{Query} & \textsc{Output} \\
\midrule
Existence &
$O((fn)^\omega + T_{\alpha})$ &
$O((fn)^2 \log n)$ &
$O(f^{3\omega})$ &
Whether $G+U$ has an $\alpha$-factor \\
Min-Cost &
$\widetilde{O}(W(fn)^3)$ &
$O(W(fn)^3 \log n)$ &
$\widetilde{O}(W n f^{3\omega+1})$ &
Weight of a min-cost $\alpha$-factor of $G+U$ \\
\end{tabular}%
}
\end{center}
\endgroup
\end{theorem}

\begin{proof}
First, we compute a fixed initial min-cost $\alpha$-factor $M$ in the graph $G$. Using this, we construct the auxiliary graph $H$, the function $\beta$, and instantiate the bounded $\beta$-factor oracle of~\Cref{thm:sensitivity-bounded-alpha}.

By \Cref{thm:bijection}, a query on $G+U$ is equivalent to the corresponding query on a $\beta$-factor of $H + U_H$, with matching cost. The graph $H$ has $4n$ vertices, degree bounds $\beta \le 2f$, and $|U_H| = O(f)$, so its sensitivity oracle matrix dimension is $N_H = O(nf)$. We instantiate \Cref{thm:sensitivity-bounded-alpha} on $H$ with update size $O(f)$.

\emph{Existence.} During preprocessing, finding the $\alpha$-factor of $G$ takes $O(T_{\alpha})$ time~\cite{Anstee85, DuanHZ20}, then constructing $H$ and bounded $\beta$-factor oracle on $H$ costs $O((2n \cdot f)^\omega) = O((nf)^\omega)$ and takes $O((fn)^2 \log n)$ space; a query with $O(f)$ updates costs $O(f^\omega \cdot f^{2\omega})$ time.

\emph{Min-cost.} During preprocessing, finding min-cost $\alpha$-factor $M$ and then constructing sensitivity oracle on $H$ costs $\widetilde{O}(W(fn)^3)$ time and takes $O(W(fn)^3 \log n)$ space; a query with $O(f)$ updates costs $\widetilde{O}(W n f^{3\omega+1})$ time.
\end{proof}

\section{Subset Sensitivity Oracles}
\label{sec:subset}

\subsection{Subset Sensitivity for Flows, Cuts, $s$-Min-Cut, and Global Min-Cut}
\label{subsec:subset-flow}

We now lift the flow and cut oracles of \Cref{sec:flows} to the subset model. All reductions, namely the path-packing equivalence (\Cref{lemma:path_packing_static}), the flow-rerouting lemmas (\Cref{lemma:reroute-1,lemma:reroute-2}), and the cut characterization (\Cref{lemma:insertion-deletion-cut}), are reused unchanged; we replace only the underlying packing oracle of \Cref{thm:sensitivity_packing} with its subset counterpart, \Cref{thm:subset_sensitivity_packing}.

\subsubsection{All-Pairs $k$-Bounded Flow}
\label{subsubsec:subset-allpairs}

Let $G = (V, E, \cost)$ be an $n$-vertex directed graph with integer edge costs in $[-W, W]$ and flow bound $k \in [1, n]$. We are given a set of \emph{query vertices} $V_Q \subseteq V$ from which source-sink pairs are drawn, and a set of \emph{susceptible edges} $E_Q$ (which may include edges not currently in $E$) subject to insertions or deletions. Any query $(s, t, U)$ satisfies $s, t \in V_Q$ and $U \subseteq E_Q$ with $|U| \le f$.

\paragraph{Ground Set and Susceptible Set}
We construct the auxiliary graph $G^*$ on $N = n + 2$ vertices with super-source $s^*$ and super-sink $t^*$, and the partition matroids $M_{\mathrm{out}}, M_{\mathrm{in}}$ as in \Cref{sec:flows}. To support all-pairs queries from $V_Q$, we include in the ground set $k$ parallel edges from $s^*$ to $s$ for each $s \in V_Q$, and $k$ parallel edges from $t$ to $t^*$ for each $t \in V_Q$.

\paragraph{Query Reduction}
For a query $(s, t, U)$ with $s, t \in V_Q$ and $U \subseteq E_Q$, we modify the diagonal entries of $D$ corresponding to: the $k$ parallel edges from $s^*$ to $s$, the $k$ parallel edges from $t$ to $t^*$, and the $f$ updates in $U$. The total number of modified diagonal entries is at most $f + 2k$. By \Cref{lemma:path_packing_static}, the resulting $k$-packing exists if and only if $G + U$ admits $\lambda$ edge-disjoint $(s, t)$-paths.

\begin{theorem}
\label{theorem:subset-allpairs-flow}
Let $G=(V,E,\cost)$ be an $n$-vertex directed graph with integer edge costs in $[-W,W]$ with designated sets $V_Q\subseteq V$ and $E_Q\subseteq E$, and let $k \in [1,n]$. Let $\sigma := |E_Q| + k|V_Q|$. Then there exist the following subset sensitivity oracles that, for any query $(s,t,U)\in V_Q^2\times E_Q^f$, answer queries on $G+U$ with high probability:
\begingroup
\setlength{\aboverulesep}{0pt}
\setlength{\belowrulesep}{0pt}
\begin{center}
\renewcommand{\arraystretch}{1.25}
\resizebox{0.98\linewidth}{!}{%
\begin{tabular}{@{}l l l l l@{}}
\rowcolor{yellow!15}
\textsc{Oracle} & \textsc{Prep.} & \textsc{Space} & \textsc{Query} & \textsc{Output} \\
\midrule
Bounded Flow &
$O((kn)^{\omega-2}\sigma^2 )$ &
$O(\sigma^2)$ &
$O((f+k)^\omega \log k)$ &
$\lambda^*=\min(k,\lambda_{G+U}(s,t))$ \\
Nearest Min-Cut &
$O((kn)^{\omega-2} (\sigma+n|V_Q|)^2)$ &
$O((\sigma+n|V_Q|)^2)$ &
$O((f+k)^\omega n)$ &
$\NMC_{G+U}(s,t)$ (if $\lambda^* < k$) \\
Min-Cost Flow &
$\widetilde{O}(W (kn)^{\omega-1}\sigma^2 )$ &
$\widetilde{O}(Wkn \sigma^2)$ &
$\widetilde{O}(Wkn(f+k)^\omega)$ &
Min-cost for flow of value $\lambda^*$ \\
\end{tabular}%
}
\end{center}
\endgroup
\end{theorem}

\begin{proof}
We instantiate the subset sensitivity oracles of \Cref{thm:subset_sensitivity_packing} on the partition matroids $M_{\mathrm{out}}, M_{\mathrm{in}}$ associated with $G^*$, with rank $r = N - 1 = n + 1$ and packing parameter $k$.\\

\noindent
\emph{Bounded Flow Oracle.}
The susceptible set $E_0$ is the union of the susceptible graph edges and the potential query edges:
\[
E_0 \;=\; E_Q \;\cup\; \bigcup_{s \in V_Q} E_{s^*, s} \;\cup\; \bigcup_{t \in V_Q} E_{t, t^*},
\]
of size $\sigma_{flow} = |E_Q| + 2k|V_Q| = O(|E_Q| + k|V_Q|) = O(\sigma)$. For a query $(s, t, U)$, we binary search over $\lambda \in [0, k]$. Each candidate modifies at most $f + 2k$ diagonal entries within $E_0$. By \Cref{thm:subset_sensitivity_packing}, each determinant evaluation takes $O((f + k)^\omega)$ time. The binary search yields $\lambda^*$ in $O((f + k)^\omega \log k)$ time. The stated bounds follow by substituting $\sigma_{flow}$ into the theorem.\\

\noindent
\emph{Nearest Min-Cut Oracle.}
We first compute $\lambda := \lambda_{G+U}(s,t)$. For each $w \in V \setminus \{s, t\}$, we query $\lambda_{G+U+(w,t)}(s,t)$ by adding one further diagonal change. To ensure the edge $(w, t)$ lies in the susceptible set for any query pair, we must absorb the edge set $\{(w, t) : w \in V,\, t \in V_Q\}$ into $E_0$ at preprocessing time. This yields a larger susceptible set of size $\sigma_{cut} = O(|E_Q| + k|V_Q| + n|V_Q|) = O(\sigma + n|V_Q|)$. Each of the $n$ queries modifies $O(f + k)$ entries, yielding a total query time of $O((f + k)^\omega n)$.\\

\noindent
\emph{Min-Cost Flow Oracle.}
After computing $\lambda^*$ via the bounded flow oracle, we invoke the weighted subset sensitivity oracle using the smaller susceptible set $\sigma_{flow}$. By \Cref{thm:subset_sensitivity_packing}, substituting $\sigma_{flow}$ yields the required preprocessing and space bounds, with query time scaling by the weight degree $W \cdot kn$.
\end{proof}

\paragraph{Implications for Reachability and Distance Sensitivity.~} 
As a direct application of our bounded-flow subset-sensitivity oracles with $k=1$ and integral edge costs (assuming no negative-weight cycles so that distances are well-defined), we obtain subset reachability and distance sensitivity oracles. Setting $k=1$ in \Cref{theorem:subset-allpairs-flow}, the all-pairs bounded-flow decision oracle reduces to the problem of computing reachability and the min-cost oracle reduces to distances after $f$ edge updates.

\begin{corollary}[Reachability and Distance Subset Sensitivity Oracle]
\label{cor:subset}
Let $G=(V,E,\cost)$ be an $n$-vertex directed graph with integer edge costs in $[-W,W]$ with designated sets $V_Q\subseteq V$ and $E_Q\subseteq E$, having no negative weight cycles. Let $\sigma = |V_Q|+|E_Q|$. Then there exist the following subset sensitivity oracles that, for any query $(s,t,U)\in V_Q^2\times E_Q^f$, answer queries on $G+U$ with high probability:

\begin{center}
\renewcommand{\arraystretch}{1.25}
\begin{tabular}{@{}l l l l l@{}}
\textsc{Oracle} & \textsc{Prep.} & \textsc{Space} & \textsc{Query} \\
\midrule
Reachability &
~$O(\sigma^2 n^{\omega-2})$~ &
~$O(\sigma^2)$~ &
~$O(f^\omega)$~  \\
Distances &
~$\widetilde{O}(W \sigma^2 n^{\omega-1})$~ &
~$\widetilde{O}(Wn \sigma^2)$~ &
~$\widetilde{O}(Wnf^\omega)$~  \\
\end{tabular}
\end{center}
\end{corollary}
 
We note that for the general model, the celebrated distance sensitivity oracle of Brand and Saranurak~\cite{BrandS19} achieves $\widetilde{O}(Wn^{3})$ preprocessing and $\widetilde{O}(Wnf^{\omega})$ query time for $f$ failures using the adjoint oracle on the generic matrix of the graph. 

\subsubsection{$(s,t)$-Max-Flow and $(s,t)$-Min-Cut}
\label{subsubsec:subset-st-flow}

Let $G = (V, E, \cost)$ be an $n$-vertex directed graph with fixed source $s$, sink $t$, and a susceptible edge set $E_Q$ (which may include edges not currently in $E$). We set $k = 2f$.
As in \Cref{subsec:gen-st-flow}, we precompute a min-cost $(s, t)$-max-flow $h$ in $G$, extend it to saturate the first $f$ edges of a set $E_f$ of $2f$ zero-cost parallel $(s, t)$-edges, and construct the auxiliary graph $\G^*_h$ on $N = n + 2$ vertices and $m+2f$ edges by adding a dummy source $s^*$ and a dummy sink $t^*$.

For any query $U = (I, F) \subseteq E_Q$, the auxiliary update set $\U^*$ on $\G^*_h$ involves edges in $E_Q$, their counterparts, and auxiliary edges $(s^*, x)$ and $(y, t^*)$ for endpoints of edges in $F_h := \{e \in F : h(e) = 1\}$. The susceptible set for the subset sensitivity oracle on $\G^*_h$ therefore consists of: (i) the edges in $E_Q$ and their counterparts in $\G_h$, (ii) for each endpoint $x$ or $y$ of an edge in $E_Q$, the auxiliary edges $(s^*, x)$ and $(y, t^*)$, and (iii) the $k$ parallel edges for the fixed query pair $(s^*, t^*)$. The susceptible set size is
\[
\sigma' \;=\; O(|E_Q| + k) \;=\; O(|E_Q| + f).
\]
For a query $U$ with $|U| \le f$, the auxiliary update $\U^*$ modifies $O(f)$ diagonal entries within this susceptible set, and the fixed $(s^*, t^*)$ query edges contribute $O(k) = O(f)$ additional modifications. The total number of modified diagonal entries is $O(f)$.

\begin{theorem}
\label{theorem:subset-st-flow}
Let $G = (V, E, \cost)$ be an $n$-vertex directed graph with integer edge costs in $[-W, W]$, with source $s$,  sink $t$, and susceptible set $E_Q\subseteq E$ of size $\sigma$. Then there exist the following subset sensitivity oracles that, for any $U\in E_Q^f$, answer queries on $G + U$ with high probability:
\begingroup
\setlength{\aboverulesep}{0pt}
\setlength{\belowrulesep}{0pt}
\begin{center}
\resizebox{0.98\linewidth}{!}{
\renewcommand{\arraystretch}{1.25}
\begin{tabularx}{\linewidth}{@{}l l l l X@{}}
\rowcolor{yellow!15}
\textsc{Oracle} & \textsc{Prep.} & \textsc{Space} & \textsc{Query} & \textsc{Output} \\
\midrule
Max-Flow & $O(m^{1+o(1)} + \sigma^2 (fn)^{\omega - 2})$ & $O(\sigma^2)$ & $\widetilde{O}(f^\omega)$ & $\lambda_{G+U}(s,t)$ \\
Nearest Min-Cut & $O( (\sigma^2+n^2) (fn)^{\omega - 2})$ & $O(\sigma^2+n^2)$ & $\widetilde{O}(f^\omega n)$ & $\NMC_{G+U}(s,t)$ \\
Min-Cost Flow & $\widetilde{O}(m^{1+o(1)} + W \sigma^2 (fn)^{\omega - 1})$ & $\widetilde{O}(W fn  \sigma^2)$ & $\widetilde{O}(Wn f^{\omega + 1})$ & Min-cost max-flow in $G+U$\\
\end{tabularx}
}
\end{center}
\endgroup
\end{theorem}

\begin{proof}
We set $k := 2f$ and instantiate the subset sensitivity oracle of \Cref{thm:subset_sensitivity_packing} on the partition matroids $M_{\mathrm{out}}, M_{\mathrm{in}}$ associated with the auxiliary graph $\G^*_h$, with rank $r = n + 1$, packing parameter $k = O(f)$, and susceptible set of size $\sigma' = O(|E_Q| + f)$.

The preprocessing bounds follow by first calculating min-cost $(s,t)$ max-flow in $G$ using the algorithm of~\cite{ChenKLPGS23} and then, substituting $kr = O(fn)$ and $\sigma = \sigma'$ into \Cref{thm:subset_sensitivity_packing}. The space is $O({\sigma'}^2)$ for the decision oracle.

\paragraph{Max-Flow Oracle.}
By \Cref{lemma:reroute-1},
$\lambda_{G+U}(s,t) = \lambda_G(s,t) + f - i$, where $i$ is determined by binary search. For each candidate $i$, the auxiliary update set $\U^*$ together with the fixed $(s^*, t^*)$ query edges modifies $O(f)$ diagonal entries within the susceptible set. By \Cref{thm:subset_sensitivity_packing}, each determinant evaluation takes $O(f^\omega)$ time; the binary search yields the answer in $\widetilde{O}(f^\omega)$ time.

\paragraph{Nearest Min-Cut Oracle.}
We first compute $\lambda := \lambda_{G+U}(s,t)$ via the max-flow oracle. For each $w \in V \setminus \{s, t\}$, we invoke the max-flow oracle on $G + U + (w, t)$, adding one additional element change. To accommodate the edges $(w,t)$ for all $w \in V$, the susceptible set in this case grows to size $O(|E_Q|+f+n)$, whence the space term $O(\sigma^2 + n^2)$. Each of the $n$ queries takes $\widetilde{O}(f^\omega)$ time by \Cref{lemma:nearest-mincut}, giving $\widetilde{O}(f^\omega \cdot n)$ total.

\paragraph{Min-Cost Flow Oracle.}
By \Cref{lemma:reroute-2}, the minimum cost is determined by the min-cost $(s^*, t^*)$-flow in $\G^*$. We invoke the weighted oracle of \Cref{thm:subset_sensitivity_packing} with $kr = O(fn)$ and $O(f)$ modified diagonal entries, giving query time $\widetilde{O}(W \cdot fn \cdot f^\omega) = \widetilde{O}(Wnf^{\omega + 1})$.
\end{proof}

\subsubsection{$s$-Min-Cut and Global Min-Cut}
\label{subsubsec:subset-mincut}
 
We now leverage the subset framework to obtain space-efficient oracles for $s$-min-cut and global min-cut, provided the updates are confined to $E_Q \subseteq E$.
 
Recall the two ingredients from \Cref{sec:s-mincut}. First, using \Cref{theorem:s-connectivity-augmentation}, we precompute in time $T_{\text{aug}}$ an optimal $f$-edge augmenting set $I^*$ that maximizes the $s$-min-cut, and set $\lambda^* := \mu_{G+I^*}(s)$. Second, by \Cref{lemma:insertion-deletion-cut}, for any update set $U$ with $|U| \le f$,
\[
\mu_{G+U}(s) \;=\; \min\Bigl(\lambda^*,\; \min_{\substack{y \in Y \\ y \ne s}} \lambda_{G+U}(s,y)\Bigr),
\qquad
Y := \{\, y : (x,y) \in I^* \cup U \,\},\quad |Y| \le 2f.
\]
The candidate targets therefore lie in $V(E_Q) \cup V(I^*)$, where $V(E_Q)$ is the set of endpoints of the edges in $E_Q$ and $V(I^*)$ is the set of (at most $f$) endpoints of the augmenting edges. Both $V(E_Q)$ and $V(I^*)$ are known before any query, so we pre-initialize one subset $(s,y)$-max-flow oracle for every $y \in V(E_Q) \cup V(I^*)$. Since $|V(E_Q)| \le 2\sigma$ and $|V(I^*)| \le f \le \sigma$, this amounts to $O(\sigma)$ oracles, each of which is exactly the subset $(s,y)$ oracle of \Cref{theorem:subset-st-flow}.

\begin{theorem}
\label{theorem:subset-s-min-cut}
Let $G = (V, E, \cost)$ be an $n$-vertex directed graph with a fixed source $s$ and a susceptible set $E_Q\subseteq E$ of size $\sigma$. Then there exist the following subset sensitivity oracles that, for any update set $U\in E_Q^f$ consisting of insertions and deletions, answer queries on $G + U$ with high probability:
\begingroup
\setlength{\aboverulesep}{0pt}
\setlength{\belowrulesep}{0pt}
\begin{center}
\renewcommand{\arraystretch}{1.25}
\resizebox{0.98\linewidth}{!}{
\begin{tabular}{@{}lllll@{}}
\rowcolor{yellow!15}
\textsc{Oracle} & \textsc{Prep.} & \textsc{Space} & \textsc{Query} & \textsc{Output} \\
\midrule
$s$-Min-Cut Value &
$O(\sigma m^{1+o(1)} + \sigma^3(fn)^{\omega - 2} + T_{\text{aug}})$ &
$O(\sigma^3)$ &
$\widetilde{O}(f^{\omega + 1})$ &
$\mu_{G+U}(s)$ \\
$s$-Min-Cut Partition &
$O\bigl((\sigma^3+\sigma n^2)(fn)^{\omega - 2} + T_{\text{aug}}\bigr)$ &
$O(\sigma^3+\sigma n^2)$ &
$\widetilde{O}(f^\omega n)$ &
$s$-min-cut in $G+U$  \\
\end{tabular}%
}
\end{center}
\endgroup
\end{theorem}
 
\begin{proof}
During preprocessing we compute in $T_{\text{aug}}$ time the optimal augmenting set $I^*$ of size $f$ using \Cref{theorem:s-connectivity-augmentation}, together with $\lambda^* = \mu_{G+I^*}(s)$. 
Also, find a minimum-cut partition $(A_0, V\setminus A_0)$ of $G$ subject to the constraint that $(\{s\} \cup V(E_Q)) \subseteq A_0$.

For every target $y \in (V(E_Q) \cup V(I^*)) \setminus \{s\}$ we instantiate the subset sensitivity max-flow and nearest-min-cut oracles of \Cref{theorem:subset-st-flow} on $G$, with source $s$, sink $y$, and susceptible set $E_Q$. There are $O(\sigma)$ such targets. For each target, the max-flow oracle uses a susceptible set of size $O(|E_Q| + f) = O(\sigma)$ and hence $O(\sigma^2)$ space, while the nearest-min-cut oracle uses a susceptible set of size $O(|E_Q| + f + n) = O(\sigma + n)$ (to accommodate the auxiliary edges $(w,y)$) and hence $O(\sigma^2 + n^2)$ space. Summing over the $O(\sigma)$ targets gives space $O(\sigma^3)$ for the value oracle and $O(\sigma^3 + \sigma n^2)$ for the partition oracle; the preprocessing times are the corresponding $(fn)^{\omega-2}$-scaled bounds of \Cref{theorem:subset-st-flow}, plus an additive  $T_{\text{aug}}$ term.
 
Given a query $U \in E_Q^f$, we form the candidate sink set $Y = \{\, y : (x,y) \in I^* \cup U \,\}$, with $|Y| \le 2f$ and $Y \subseteq V(E_Q) \cup V(I^*)$. For each $y \in Y$ with $y \ne s$ we call the $(s,y)$ subset max-flow oracle on $G+U$, computing $\lambda_{G+U}(s,y)$ in $\widetilde{O}(f^\omega)$ time; over the $\le 2f$ targets this takes $\widetilde{O}(f^{\omega+1})$ time. By \Cref{lemma:insertion-deletion-cut},
\[
\mu_{G+U}(s) \;=\; \min\Bigl(\lambda^*,\; \min_{y \in Y,\, y\ne s}\lambda_{G+U}(s,y)\Bigr).
\]


For the partition query, we consider two cases:
\begin{itemize}
    \item If $\mu_{G+U}(s) < \lambda^*$, we set $y^* = \arg\min_{y\in Y,\, y\ne s}\lambda_{G+U}(s,y)$ and invoke the $(s,y^*)$ subset nearest-min-cut oracle on $G+U$, which returns the partition in $\widetilde{O}(f^\omega n)$ time. 

    \item If $\mu_{G+U}(s) = \lambda^*$, we query for all vertices $v \in V(E_Q)$ if $\lambda_{G+U}(s,v) = \lambda^*$. 
    
    If there exists such a vertex $v$, then invoking the $(s, v)$ subset nearest-min-cut oracle on $G + U$ returns the desired partition in $\widetilde{O}(f^\omega n)$ time.
    
    If no such $v \in V(E_Q)$ exists, then all vertices of $V(E_Q)$ lie on the source side of the $s$-min-cut in $G+U$. Since updates to edges in $E_Q$ do not affect the value of such cuts, we obtain that the cut $(A_0, V \setminus A_0)$ is the desired min-cut vertex partition. 
\end{itemize}

So, the total query time to output the vertex partition is $\widetilde{O}(f^\omega n)$.

\end{proof}
 
\begin{theorem}
\label{theorem:subset-global-min-cut}
Let $G = (V, E, \cost)$ be an $n$-vertex directed graph with a susceptible set $E_Q \subseteq E$ of size $\sigma$. Then there exist the following subset sensitivity oracles that, for any update set $U\in E_Q^f$ consisting of insertions and deletions, answer queries on $G + U$ with high probability:
\begingroup
\setlength{\aboverulesep}{0pt}
\setlength{\belowrulesep}{0pt}
\begin{center}
\renewcommand{\arraystretch}{1.25}
\resizebox{0.98\linewidth}{!}{
\begin{tabular}{@{}lllll@{}}
\toprule
\rowcolor{yellow!15}
\textsc{Oracle} & \textsc{Prep.} & \textsc{Space} & \textsc{Query} & \textsc{Output} \\
\midrule
Global Min-Cut Value &
$O(\sigma m^{1+o(1)} + \sigma^3(fn)^{\omega - 2} + T_{\text{aug}})$ &
$O(\sigma^3)$ &
$\widetilde{O}(f^{\omega + 1})$ &
$\mu_{G+U}$ \\
Global Min-Cut Partition &
$O\bigl((\sigma^3+\sigma n^2)(fn)^{\omega - 2} + T_{\text{aug}}\bigr)$ &
$O(\sigma^3+\sigma n^2)$ &
$\widetilde{O}(f^\omega n)$ &
Global min-cut in $G+U$ \\
\bottomrule
\end{tabular}
}
\end{center}
\endgroup
\end{theorem}
 
\begin{proof}
We fix an arbitrary vertex $s\in V$ at preprocessing. As in the general construction (\Cref{theorem:global-min-cut}), we compute \emph{independent} optimal augmentations, and hence distinct augmenting sets $I^*$, for $G$ and for $G^{\mathrm{rev}}$ with respect to $s$. We then instantiate the subset $s$-min-cut value and partition oracles of \Cref{theorem:subset-s-min-cut} for both $G$ (with susceptible set $E_Q$) and $G^{\mathrm{rev}}$ (with the reversed susceptible set $E_Q^{\mathrm{rev}}$, of the same size $\sigma$). The preprocessing and space bounds are same as in \Cref{theorem:subset-s-min-cut}.
 
Using the fact $\mu_{G+U} = \min\bigl(\mu_{G+U}(s),\ \mu_{(G+U)^{\mathrm{rev}}}(s)\bigr)$, a value query evaluates the $s$-min-cut value on both $G+U$ and $(G+U)^{\mathrm{rev}}$ in $\widetilde{O}(f^{\omega+1})$ time and returns the smaller value. For a partition query, we first determine which of the two attains the minimum and then invoke the corresponding subset partition oracle in $\widetilde{O}(f^\omega n)$ time.
\end{proof}
\subsection{Subset Sensitivity for Spanning Structures}
\label{subsec:subset-trees}

Finally, we lift the spanning-structure oracles of \Cref{sec:trees} to the subset model. The reductions of \Cref{theorem:edmonds_packing,lemma:k-spanning-trees-packing,lemma:forest-covering-equivalence} are reused unchanged; only the underlying packing/covering oracle is replaced by its subset counterpart.

\subsubsection{$k$-Disjoint Arborescences}
\label{subsubsec:subset-arb}

Let $G = (V, E, \cost)$ be an $n$-vertex directed graph. We are given a set of \emph{query roots} $V_Q \subseteq V$ from which the arborescence root can be chosen, and a set of \emph{susceptible edges} $E_Q$ (which may include edges not currently in $E$) subject to insertions or deletions.

\paragraph{Ground Set and Susceptible Set}
We construct $G^*$ with the dummy root $z$ and matroids $M_{\mathrm{graph}}$ and $M_{\mathrm{part}}$ as in \Cref{sec:trees}. To support dynamic root queries from $V_Q$, we include in the full ground set $k$ parallel edges from $z$ to $s$ for each $s \in V_Q$. The susceptible set is the union of the susceptible graph edges and the potential root query edges:
\[
E_0 \;=\; E_Q \;\cup\; \bigcup_{s \in V_Q} E_{z, s}, \qquad \sigma \;:=\; |E_0| \;=\; |E_Q| + k|V_Q|.
\]
A query $(s, U)$ with $s \in V_Q$ and $U \subseteq E_Q$ modifies exactly $f$ diagonal entries for the edge updates and exactly $k$ diagonal entries to activate the $(z,s)$ query edges. Thus, all queries toggle at most $f + k$ entries, strictly within $E_0$.

\begin{theorem}
\label{theorem:subset-arborescences}
Let $G = (V, E, \cost)$ be an $n$-vertex directed graph with integer edge costs in $[-W, W]$, query roots $V_Q\subseteq V$, and a susceptible set $E_Q$, and let $k \ge 1$. Let $\sigma = |E_Q| + k|V_Q|$. Then there exist the following subset sensitivity oracles that, for any $s\in V_Q$ and any $U\in E_Q^f$, answer queries on $G + U$ with high probability:
\begingroup
\setlength{\aboverulesep}{0pt}
\setlength{\belowrulesep}{0pt}
\begin{center}
\renewcommand{\arraystretch}{1.25}
\resizebox{0.98\linewidth}{!}{%
\begin{tabular}{@{}lllll@{}}
\rowcolor{yellow!15}
\textsc{Oracle} & \textsc{Prep.} & \textsc{Space} & \textsc{Query} & \textsc{Output} \\
\midrule
Decision &
$O(\sigma^2 (kn)^{\omega - 2} + (kn)^\omega)$ &
$O(\sigma^2)$ &
$O((f + k)^\omega)$ &
Whether $G + U$ has $k$-disjoint $s$-rooted arborescences \\
Weighted &
$\widetilde{O}(W \sigma^2 (kn)^{\omega - 1} + W(kn)^{\omega+1})$ &
$\widetilde{O}(W  kn  \sigma^2)$ &
$\widetilde{O}(W  kn  (f + k)^\omega)$ &
Min total cost of such  $k$-arborescences \\
\end{tabular}%
}
\end{center}
\endgroup
\end{theorem}

\begin{proof}
We invoke \Cref{thm:subset_sensitivity_packing} on $M_{\mathrm{graph}}$ and $M_{\mathrm{part}}$ with rank $r = n$, packing parameter $k$, and susceptible set size $\sigma$. Since any query modifies exactly $f + k$ diagonal entries in the weight matrix, we substitute these parameters directly into the bounds of \Cref{thm:subset_sensitivity_packing} to complete the proof.
\end{proof}

\paragraph{Improved $s$-Min-Cut and Global Min-Cut via Fixed-Root Arborescences.}
Recall that the $(f+k)^\omega$ query time in \Cref{theorem:subset-arborescences} is due to the fact that a query toggles $f+k$ diagonal entries: $f$ for the edge updates and $k$ for activating the parallel $(z,s)$ query edges. If the root
$s$ is \emph{fixed} at preprocessing, the $k$ parallel edges $E_{z,s}$ can be activated once during preprocessing, so a query toggles only the $f$ entries of $U$. This yields a fixed-root oracle with query time $O(f^\omega)$ and
space $O(\sigma^2)$, where now $\sigma = |E_Q|$, independent of $k$.

This observation gives an alternate route to $s$-min-cut and global min-cut oracles. By Edmonds' theorem, $\mu_G(s)$ equals the maximum $k$ for which $G$ admits $k$ edge-disjoint $s$-rooted arborescences, and $f$ updates change this value by at most $f$. Hence, letting $\rho := \mu_G(s)$, it suffices to instantiate the fixed-root decision oracle for each $k \in [\max(0,\rho-f),\, \rho+f]$, that is, $O(f)$ instances. A query performs a binary search over these instances to locate the largest feasible $k$. Global min-cut follows by maintaining these structures for both $G$ and $G^{\mathrm{rev}}$ with respect to an arbitrary fixed $s$, since
$\mu_{G+U} = \min\bigl(\mu_{G+U}(s),\ \mu_{(G+U)^{\mathrm{rev}}}(s)\bigr)$.
We thus obtain the following.

\begin{corollary}[$s$-Min-Cut and Global Min-Cut Value via Arborescence Packing]
\label{cor:s-global-mincut-via-arb}
Let $G = (V, E)$ be an $n$-vertex directed graph with a susceptible set $E_Q \subseteq E$ of size $\sigma$. Then there exist subset sensitivity oracles that, for any update set $U \in E_Q^f$, report the values $\mu_{G+U}(s)$ (for a source $s$ fixed at preprocessing) and $\mu_{G+U}$, each with $O(f\sigma^2)$ space and $\widetilde{O}(f^\omega)$ query time, with high probability.

Moreover, in the general sensitivity model, taking $E_Q = E$ yields $s$-min-cut and global min-cut value oracles with $O(fm^2)$ space and $\widetilde{O}(f^\omega)$ query time.
\end{corollary}

Compared to \Cref{theorem:subset-s-min-cut,theorem:subset-global-min-cut}, \Cref{cor:s-global-mincut-via-arb} improves the query time from $\widetilde{O}(f^{\omega+1})$ to $\widetilde{O}(f^\omega)$ and the space from $O(\sigma^3)$ to $O(f\sigma^2)$. Similarly, in the general model it provides an alternative to \Cref{theorem:s-min-cut,theorem:global-min-cut}, improving the query time from $\widetilde{O}(f^{3\omega+1})$ to $\widetilde{O}(f^\omega)$ at the cost of increased space. We note that \Cref{cor:s-global-mincut-via-arb} reports only the cut \emph{value}; reporting the cut partition still goes through the nearest-min-cut approach of \Cref{theorem:subset-s-min-cut}.

\subsubsection{$k$-Disjoint Spanning Trees}
\label{subsubsec:subset-spanning}
Recall from \Cref{thm:sensitivity_k_spanning_trees} that finding $k$ edge-disjoint spanning trees was cast as a $k$-packing of common bases of two copies of the
graphic matroid. We now restrict all updates to
a susceptible set fixed at preprocessing time, replacing the general packing oracle of \Cref{thm:sensitivity_packing} with its subset counterpart,
\Cref{thm:subset_sensitivity_packing}.

Let $E_Q$ be a susceptible set of size $\sigma$ (which may include edges not currently in $E$), representing edges that can be added, deleted, or re-weighted. Because this problem requires no auxiliary root vertices, the susceptible set for the data structure consists entirely of the edges in $E_Q$, and modifying $f$ elements toggles or re-weights exactly $f$ diagonal entries in the packing Laplacian.

\begin{theorem} \label{thm:subset_sensitivity_k_spanning_trees}
Let $k \ge 1$ be an integer, and $G = (V, E, \cost)$ be an $n$-vertex undirected graph with integer edge costs in $[-W, W]$ and a susceptible edge set $E_Q$ of size $\sigma$. Then there exist the following subset sensitivity oracles that, for any set of updates $U\subseteq E_Q$ of size at most $f$, answer queries on $G + U$ with high probability:
\begingroup
\setlength{\aboverulesep}{0pt}
\setlength{\belowrulesep}{0pt}
\begin{center}
\renewcommand{\arraystretch}{1.25}
\resizebox{0.98\linewidth}{!}{%
\begin{tabular}{@{}lllll@{}}
\rowcolor{yellow!15}
\textsc{Oracle} & \textsc{Prep.} & \textsc{Space} & \textsc{Query} & \textsc{Output} \\
\midrule
Decision &
$O(\sigma^2 (kn)^{\omega - 2} + (kn)^\omega)$ &
$O(\sigma^2)$ &
$O(f^\omega)$ &
Whether $G + U$ has $k$ edge-disjoint spanning trees \\
Weighted &
$\widetilde{O}(W\sigma^2 (kn)^{\omega - 1} + W(kn)^{\omega+1})$ &
$\widetilde{O}(W kn \cdot \sigma^2)$ &
$\widetilde{O}(W kn \cdot f^\omega)$ &
Min total cost of such $k$-spanning trees \\
\end{tabular}%
}
\end{center}
\endgroup
\end{theorem}

\begin{proof}
We invoke the subset sensitivity oracle from \Cref{thm:subset_sensitivity_packing} on $M_1 = M_{\mathrm{graph}}$ and $M_2 = M_{\mathrm{graph}}$. The rank of both matroids is $r = n-1$, the packing parameter is $k$, and the susceptible set size is $\sigma = |E_Q|$. A query consists of $f$ edge updates, modifying at most $f$ diagonal entries in the representation. Substituting these parameters directly into the bounds of \Cref{thm:subset_sensitivity_packing} yields the result.
\end{proof}

\subsubsection{Colorful Spanning Trees}
\label{subsubsec:subset-colorful}

A colorful spanning tree is a
$1$-packing of common bases of the graphic matroid $M_{\mathrm{graph}}$ and the
color-partition matroid $M_{\mathrm{color}}$. We keep the reduction presented in \Cref{thm:sensitivity_colorful} unchanged and
confine updates to a fixed susceptible set $E_Q$ of size $\sigma$ (which may include edges not currently in $E$, each assigned to one of the $n-1$ color classes), representing edges whose presence, cost, or color can change. Since queries introduce no auxiliary structures, the susceptible set for the data structure is exactly $E_Q$, and modifying $f$ elements toggles or re-weights exactly $f$ diagonal entries.

\begin{theorem} \label{thm:subset_sensitivity_colorful}
Let $G=(V, E, \cost)$ be an undirected colored graph with a susceptible edge set $E_Q$ of size $\sigma$. Then there exist the following subset sensitivity oracles for any set of updates $U \subseteq E_Q$ of size at most $f$:
\begingroup
\setlength{\aboverulesep}{0pt}
\setlength{\belowrulesep}{0pt}
\begin{center}
\renewcommand{\arraystretch}{1.25}
\resizebox{0.98\linewidth}{!}{%
\begin{tabular}{@{}lllll@{}}
\rowcolor{yellow!15}
\textsc{Oracle} & \textsc{Prep.} & \textsc{Space} & \textsc{Query} & \textsc{Output} \\
\midrule
Decision &
$O(\sigma^2 n^{\omega - 2} + n^\omega)$ &
$O(\sigma^2)$ &
$O(f^\omega)$ &
Whether $G + U$ has a colorful spanning tree \\
Weighted &
$\widetilde{O}(W\sigma^2 n^{\omega - 1} + Wn^{\omega+1})$ &
$\widetilde{O}(W n \sigma^2)$ &
$\widetilde{O}(W n f^\omega)$ &
Min cost of a colorful spanning tree \\
\end{tabular}%
}
\end{center}
\endgroup
\end{theorem}

\begin{proof}
We invoke \Cref{thm:subset_sensitivity_packing} on $M_{\mathrm{graph}}$ and $M_{\mathrm{color}}$ with $k=1$, rank $r = n-1$, and susceptible set size $\sigma = |E_Q|$. A query modifies at most $f$ diagonal entries. Substituting these parameters directly into the bounds of \Cref{thm:subset_sensitivity_packing} yields the result.
\end{proof}

\subsubsection{Packing Spanning Forests}
\label{subsubsec:subset-forests}

Let $E_Q$ be a susceptible set of size $\sigma$ (which may include edges not currently in $E$), representing edges that can be dynamically inserted or deleted. Modifying $f$ elements toggles $f$ diagonal entries in the covering Laplacian.

\begin{theorem} \label{thm:subset_sensitivity_k_forests}
Let $k \ge 1$ be an integer, and $G=(V, E)$ be an undirected graph with $n$ vertices and a susceptible edge set $E_Q$ of size $\sigma$. Let $E_0 = E \cup E_Q$. There exists a subset sensitivity oracle that, for any set of edge insertions or deletions $U \subseteq E_Q$ of size at most $f$ (where the updated active edge count $m_{curr} \le k(n-1)$), answers queries on $G + U$ with high probability:
\begingroup
\setlength{\aboverulesep}{0pt}
\setlength{\belowrulesep}{0pt}
\begin{center}
\renewcommand{\arraystretch}{1.25}
\resizebox{0.98\linewidth}{!}{%
\begin{tabular}{@{}lllll@{}}
\rowcolor{yellow!15}
\textsc{Oracle} & \textsc{Prep.} & \textsc{Space} & \textsc{Query} & \textsc{Output} \\
\midrule
Decision &
$O(k^2|E_0| + \sigma^2 (kn)^{\omega - 2} + (kn)^\omega)$ &
$O(\sigma^2)$ &
$O(f^\omega)$ &
Whether $G + U$ partitions into $k$ forests \\
\end{tabular}%
}
\end{center}
\endgroup
\end{theorem}

\begin{proof}
We invoke the subset sensitivity covering oracle from \Cref{thm:subset_sensitivity_covering} on the graphic matroid $M_{\mathrm{graph}}$. The rank of the matroid is $r = n-1$, the covering parameter is $k$, and the susceptible set size is bounded by $\sigma = |E_Q|$.

A query consists of up to $f$ edge insertions or deletions confined to the set $E_Q$. In the algebraic framework of \Cref{thm:subset_sensitivity_covering}, inserting or deleting an edge corresponds to toggling its respective diagonal entry in the selection matrix $D$. Furthermore, the data structure dynamically toggles at most $f$ auxiliary random columns in $R_{switch}$ to maintain a perfectly square active submatrix.

Thus, the total number of toggled diagonal entries is at most $2f$. Substituting the matroid rank $r = n-1$ and the susceptible set size $\sigma$ into the bounds of \Cref{thm:subset_sensitivity_covering} yields the result. The space complexity relies strictly on the susceptible set, $O(\sigma^2)$, and the query time is fully decoupled from the graph size and the parameter $k$, taking exactly $O(f^\omega)$ operations.
\end{proof}

\subsection{Subset Sensitivity for Matroid Parity}
Recall from \Cref{sec:matroid_parity} that enumerating or finding the minimum-weight parity basis of a linear matroid of rank $r$ with $m$ pairs reduces to maintaining the determinant (or Pfaffian) of the $r \times r$ skew-symmetric Linear Matroid Parity (LMP) matrix:
$$
\mathcal{M}_{LMP} := A D B^T - B D A^T,
$$
where $A, B \in \F^{r \times m}$ represent the static coordinate vectors of the pairs, and $D \in \F(x)[Y]^{m \times m}$ is the diagonal matrix encoding pair weights and formal variables. 

To maintain $\det(\mathcal{M}_{LMP})$ in the sensitivity setting using our framework (\Cref{lemma:det}), we factorize it into the standard form $\hat{U} D_{ext} \hat{V}^T$. 
We define two augmented representation matrices: 
$$
\hat{U} := [A \mid B] \in \F^{r \times 2m} \quad
\text{and} \quad \hat{V} := [B \mid -A] \in \F^{r \times 2m}.
$$
Let $D_{ext}$ be the $2m \times 2m$ block-diagonal matrix constructed by duplicating $D$:
$$
D_{ext} := \begin{pmatrix} D & 0 \\ 0 & D \end{pmatrix}.
$$
The \emph{Parity Laplacian} is defined as follows and is equivalent to the LMP matrix:
$$
\mathcal{L}_{parity} := \hat{U} D_{ext} \hat{V}^T = A D B^T - B D A^T = \mathcal{M}_{LMP}.
$$
This block-matrix factorization allows us to directly apply our subset sensitivity determinant data structures.

Suppose the representations of the pairs or their weights undergo an update set $U$ of size $f$. An update to pair $i$ modifies exactly two diagonal entries in $D_{ext}$ (at indices $i$ and $i+m$). Therefore, applying $f$ pair updates translates to modifying at most $2f$ diagonal entries in $D_{ext}$.

\begin{theorem}[Matroid Parity Subset Sensitivity Oracle]
\label{thm:parity_subset_sensitivity}
Let $E_Q \subseteq E$ be a susceptible set of $\sigma$ pairs. There exist randomized subset sensitivity oracles that, given $f$ updates confined to pairs in $E_Q$, answer parity basis queries with high probability:
\begin{enumerate}
\item \textbf{Decision Oracle:} Requires $O(\sigma^2 r^{\omega - 2} + r^\omega)$ preprocessing time, $O(\sigma^2)$ space, and answers queries in $O(f^\omega)$ time.
\item \textbf{Weighted Oracle:} Requires $\widetilde{O}(W \sigma^2 r^{\omega - 1} + W r^{\omega+1})$ preprocessing time, $\widetilde{O}(W r \sigma^2)$ space, and answers queries in $\widetilde{O}(W r f^\omega)$ time.
\end{enumerate}
\end{theorem}

\begin{proof}
We invoke \Cref{lemma:det} on the factorization $\hat{U} D_{ext} \hat{V}^T$. The matrix $D_{ext}$ has dimension $2m \times 2m$. The designated susceptible set corresponds to the $2\sigma$ diagonal entries representing the pairs in $E_Q$. Modifying $f$ pairs toggles at most $2f$ diagonal entries within this susceptible set. Substituting the susceptible size $2\sigma$ and update size $2f$ directly into the bounds of \Cref{lemma:det} completes the proof, reducing the space complexity to $O(\sigma^2)$.
\end{proof}

\subsection{Subset Sensitivity for $\alpha$-Factors and Matching}

We now lift the $\alpha$-factor oracles of \Cref{sec:sensitivity-alpha-factors} to the subset model. Recall that in \Cref{lemma:factor-to-parity}, we reduced the $\alpha$-factor problem to matroid parity, which allowed us to handle two regimes in the general setting: the $k$-bounded case (\Cref{sec:bounded-alpha}), and the arbitrary-degree case (\Cref{sec:arbitrary-alpha}), where an update on $G$ is reduced to an update on an auxiliary graph $H$ of dimension $O(nf)$ using \Cref{thm:bijection}. To obtain subset sensitivity, we replace the general sensitivity oracle for linear matroid parity (\Cref{thm:parity_general_sensitivity}) with its subset sensitivity counterpart (\Cref{thm:parity_subset_sensitivity}), which fixes a susceptible set $E_Q$ of size $\sigma$ during preprocessing. The following results are obtained by a straightforward combination of \Cref{lemma:factor-to-parity} and \Cref{thm:parity_subset_sensitivity}.

\paragraph{The $k$-Bounded Case} We obtain the following oracle for the case when $\alpha(v) \leq k$ for all $v \in V$.

\begin{theorem}
\label{thm:subset-bounded-alpha}
Let $G=(V,E,\cost)$ be an undirected graph with integer edge costs in $[-W,W]$, degree bounds $\alpha \le k$, and a susceptible edge set $E_Q$ of size $\sigma$. Then there exist the following subset sensitivity oracles that, for any update set $U \subseteq E_Q$ with $f = |U|$, answer queries on $G+U$ with high probability:

\begingroup
\setlength{\aboverulesep}{0pt}
\setlength{\belowrulesep}{0pt}
\begin{center}
\renewcommand{\arraystretch}{1.25}
\resizebox{0.98\linewidth}{!}{%
\begin{tabular}{@{}lllll@{}}
\rowcolor{yellow!15}
\textsc{Oracle} & \textsc{Prep.} & \textsc{Space} & \textsc{Query} & \textsc{Output} \\
\midrule
Existence &
$O(\sigma^2 (kn)^{\omega-2} + (kn)^\omega)$ &
$O(\sigma^2)$ &
$O(f^\omega)$ &
Whether $G+U$ has an $\alpha$-factor \\
Min-Cost &
$\widetilde{O}(W \sigma^2 (kn)^{\omega-1} + W(kn)^{\omega+1})$ &
$\widetilde{O}(W kn \sigma^2)$ &
$\widetilde{O}(W kn f^\omega )$ &
Weight of a min-cost $\alpha$-factor of $G+U$ \\
\end{tabular}%
}
\end{center}
\endgroup
\end{theorem}

\begin{proof}
Recall from \Cref{sec:bounded-alpha} that when $\max_v \alpha(v) \le k$, the underlying matrix $A$ for the matroid parity problem is of size $kn \times 2m$ (yielding a rank $r = kn$), and each edge corresponds to one pair in the matroid ground set. Therefore, plugging $r = kn$ into \Cref{thm:parity_subset_sensitivity} directly yields the desired oracle.
\end{proof}

\paragraph{Arbitrary Degree Bounds via Auxiliary Graphs}
We now remove the assumption $\alpha \le k$. As established in \Cref{sec:arbitrary-alpha}, a direct expansion is prohibitively large, so we retain the auxiliary graph $H$ and the bijection of \Cref{thm:bijection}: a query on $G+U$ is answered by an equivalent query on a $\beta$-factor of $H + U_H$, where $H$ has degree bounds $\beta \le 2f$, dimension $N_H = O(nf)$, and update size $|U_H| = O(f)$. In the subset model, the susceptible edges $E_Q$ of $G$ translate, under the encoding of \Cref{sec:arbitrary-alpha}, to a susceptible set $E_Q^H$ in $H$: each susceptible edge $(a,b) \in E_Q$ affects only its Layer 1 copy $(a_1,b_1)$ or Layer 2 copy $(a_2,b_2)$, together with the associated parallel and sink edges at $a$ and $b$. Hence, $|E_Q^H| = O(\sigma)$, and we can instantiate the $k$-bounded subset oracle of \Cref{thm:subset-bounded-alpha} on $H$ using the parameter $k := 2f$.

\begin{theorem}
\label{thm:subset-arbitrary-alpha}
Let $G = (V,E,\cost)$ be an undirected graph with integer edge costs in $[-W,W]$, arbitrary degree bounds $\alpha$, and a susceptible edge set $E_Q$ of size $\sigma$. Then there exist the following subset sensitivity oracles that, for any update set $U \subseteq E_Q$ with $f = |U|$, answer queries on $G+U$ with high probability:

\begingroup
\setlength{\aboverulesep}{0pt}
\setlength{\belowrulesep}{0pt}
\begin{center}
\renewcommand{\arraystretch}{1.25}
\resizebox{0.98\linewidth}{!}{%
\begin{tabular}{@{}lllll@{}}
\rowcolor{yellow!15}
\textsc{Oracle} & \textsc{Prep.} & \textsc{Space} & \textsc{Query} & \textsc{Output} \\
\midrule
Existence &
$O(\sigma^2 (fn)^{\omega-2} + (fn)^\omega + T_{\alpha}) $ &
$O(\sigma^2)$ &
$O(f^\omega)$ &
Whether $G+U$ has an $\alpha$-factor \\
Min-Cost &
$\widetilde{O}(W \sigma^2 (fn)^{\omega-1} + W(fn)^{\omega+1})$ &
$\widetilde{O}(W fn \sigma^2)$ &
$\widetilde{O}(W n f^{\omega+1} )$ &
Weight of a min-cost $\alpha$-factor of $G+U$ \\
\end{tabular}%
}
\end{center}
\endgroup
\end{theorem}

\begin{proof}
First, we compute a fixed initial $\alpha$-factor of minimum-cost, and refer to it as $M$. Finding the $\alpha$-factor of $G$ takes $O(T_{\alpha})$ time~\cite{Anstee85, DuanHZ20}. Then we construct the graph $H$ and the function $\beta$.

By \Cref{thm:bijection}, a query on $G+U$ is equivalent, with matching cost, to a $\beta$-factor query on $H + U_H$, where $H$ has $4n$ vertices, degree bounds $\beta \le 2f$, dimension $N_H = O(nf)$, and $|U_H| = O(f)$. The susceptible edges $E_Q$ of $G$ correspond to a susceptible set $E_Q^H$ of $H$ with $|E_Q^H| = O(\sigma)$. We instantiate \Cref{thm:subset-bounded-alpha} on $H$ with capacity parameter $k := 2f$, susceptible-set size $O(\sigma)$, and update size $O(f)$. Substituting $k = O(f)$ and matrix dimension $N_H = O(nf)$ into the bounds of \Cref{thm:subset-bounded-alpha} yields the stated preprocessing, space, and query complexities.
\end{proof}

\section{Lower Bounds}
\label{sec:lb}
In this section, we present our lower bound results. 

We first show that any $(s,t)$-\BMF \ sensitivity oracle requires $\Omega(\min\{\sigma^2,n^2\})$ space, irrespective of the query time, assuming at least two edge updates. This holds even for the case when flow value is bounded by 1, and imply that our oracles of~\Cref{theorem:FT-all-pairs-flow},~\ref{theorem:st-flow} have almost optimal space for a constant number of edge updates.

\begin{theorem}[Flows and Cuts]
\label{thm:flow-cut}
For any positive integers $n$ and $\sigma$ satisfying $4 \le \sigma \le n(n-1)$, there exists an $n$-vertex directed unit-capacity graph $G=(V,E)$ and a susceptible edge set $E_Q$ disjoint from $E$, with $|E_Q|=O(\sigma)$, such that any subset sensitivity oracle supporting at most $f=2$ updates requires $\Omega(\min\{\sigma^2,n^2\})$ bits of space for each of the following problems:
\begin{enumerate}
\item[(a)] $k$-bounded $(s,t)$ maximum flow and $(s,t)$ minimum cut, for $k=1$;
\item[(b)] Global directed minimum cut;
\item[(c)] $s$-min-cut and $s$-rooted spanning arborescences (fixed source $s$).
\end{enumerate}
\end{theorem}

\begin{proof}
Let $p = \min(\lfloor \sigma/4 \rfloor, \lfloor n/4 \rfloor)$ and let $M\in\{0,1\}^{p\times p}$ be an arbitrary binary matrix. Define a graph $G$ with vertex-set $V(G)=X\cup Y\cup\{s,t\}\cup Z$, where $X=\{x_1,\ldots,x_p\}$, $Y=\{y_1,\ldots,y_p\}$, and $Z$ consists of $n-2p-2$ isolated dummy vertices, and the edge set 
$$
E(G)=\{(x_a,y_b):M[a,b]=1\}\cup\{(v,s),(t,v):v\in V(G)\setminus\{s,t\}\}.
$$
Also, define the susceptible edge set $E_Q=\{(s,x_a),(y_b,t):a,b\in[p]\}$, so $E_Q\cap E(G)=\emptyset$ and $|E_Q|=2p\le\sigma$.

Consider an update $U=\{(s,x_i),(y_j,t)\}\subseteq E_Q$. In $G+U$, every simple $(s,t)$-path must use the inserted edges $(s,x_i)$ and $(y_j,t)$, since the edges $(v,s)$ and $(t,v)$ cannot lie on a simple $(s,t)$-path. Hence $(s, x_i, y_j, t)$ is the unique candidate $(s,t)$-path, which exists if and only if $(x_i,y_j)\in E(G)$, that is, $M[i,j]=1$.
We now show that each of the listed problems distinguishes whether $M[i,j]=1$.

\begin{enumerate}
\item[(a)] Since there is at most one $(s,t)$-path, the value of the $1$-bounded $(s,t)$-maximum flow (equivalently, the $(s,t)$-minimum cut) value in $G+U$ is $1$ if $M[i,j]=1$ and $0$ otherwise.

\item[(b)] 
Observe that any cut having $s$ on the sink side or $t$ on the source side will have a capacity of 1, because of the fixed edges of form $(v,s)$ and $(t,v)$. As discussed previously, $(s,t)$-min-cut in $G+U$ is 1 if and only if $M[i, j] = 1$. Hence the value of global directed minimum cut in $G+U$ equals $M[i,j]$.


\item[(c)] From the above discussions, $s$-min-cut in $G+U$ has value exactly $M[i,j]$. Moreover, an $s$-rooted spanning arborescence exists if and only if every vertex is reachable from $s$, which holds exactly when the edge $(x_i,y_j)$ is present, that is, when $M[i,j]=1$.
\end{enumerate}

Thus, each update reveals the corresponding matrix entry $M[i,j]$. Since correctly answering all $p^2$ queries uniquely determines $M$ and $p=\Theta(\min\{\sigma,n\})$, any subset sensitivity oracle requires
$\Omega(p^2)=\Omega\ \left(\min\{\sigma^2,n^2\}\right)$ bits of space.
\end{proof}

We next present our lower bound result for colorful spanning trees.

\begin{theorem}[Colorful Spanning Trees]
\label{thm:colorful-tree}
Let $n, \sigma$ be positive integers with $4 \le \sigma \le n(n-1)$. There exists an $n$-vertex edge-colored graph carrying a base spanning tree $T^*$ and a susceptible set $E_Q$ of size $\Theta(\sigma)$ such that any subset sensitivity oracle for colorful spanning tree requires $\Omega(\min(\sigma^2, n^2))$ bits of space under $f=2$ updates.
\end{theorem}

\begin{proof}
Set $p = \min(\lfloor \sigma/4 \rfloor, \lfloor (n-2)/2 \rfloor) \ge 1$, which ensures that $2p + 2 \le n$ and $p+1 \le \sigma$. Let $M \in \{0,1\}^{p \times p}$ be an arbitrary boolean matrix.

We construct an undirected graph $G = (V, E)$ on exactly $n$ vertices. First, we build the core sub-graph on $2p + 2$ vertices defined by $V_{\text{core}} = \{s, t\} \cup X \cup Y$, where $X = \{x_1, \dots, x_p\}$ and $Y = \{y_1, \dots, y_p\}$. The remaining $n - (2p + 2)$ vertices are added as dummy vertices to satisfy the vertex count. The edge set and color assignments for the core graph are defined as follows:

\begin{enumerate}
\item For each $k \in [p]$, add the edge $(s, x_k)$ with color $k$. These $p$ edges have distinct colors $1, \dots, p$.
\item For each $k \in [p]$, add the edge $(t, y_k)$ with color $p + k$. These $p$ edges have distinct colors $p+1, \dots, 2p$.
\item Add the edge $(s, t)$ with color $2p + 1$.
\item For each $(i,j) \in [p] \times [p]$ with $M[i,j] = 1$, add the edge $(x_i, y_j)$ with color $i$.
\end{enumerate}

To ensure the entire graph is connected, add edges between the $n - (2p + 2)$ dummy vertices and $s$ using $n - (2p + 2)$ new, distinct colors. The total number of colors is exactly $n-1$. Define the susceptible set as $E_{Q} = \{t\}\times (X\cup Y)$ (note $|E_Q| = O(\sigma)$).

Consider a target pair $(i, j) \in [p] \times [p]$. The $f = 2$ update query $U$ consists of: 
\begin{itemize}
\item Deletion of the edge $(y_j, t)$ of color $p+j$.
\item Insertion of a new edge $(x_i,t)$ of color $p + j$.
\end{itemize}

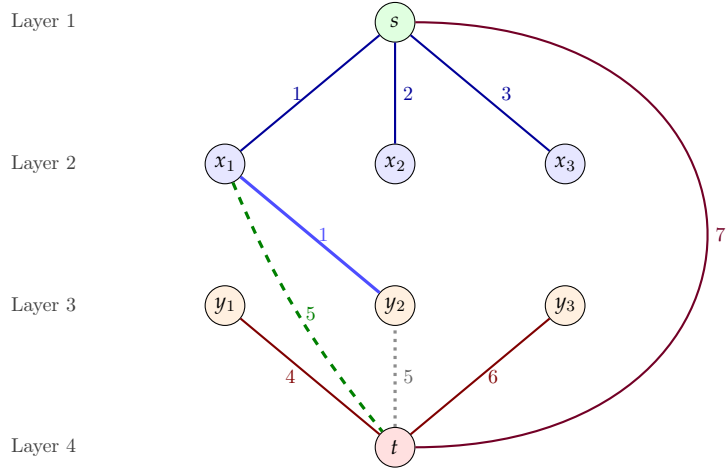
\begin{figure}[t]
\centering
~\quad~
\begin{tikzpicture}[
    vtx/.style={circle, draw, minimum size=7mm, inner sep=0pt},
    sstar/.style={thick, blue!60!black},
    tstar/.style={thick, red!50!black},
    stconn/.style={thick, purple!60!black},
    clique/.style={thick, black!20, bend left=20},
    failed/.style={very thick, blue!60!black, dotted},
    inserted/.style={very thick, green!50!black, dashed},
    matrix/.style={very thick, blue!70},
    scale=0.75, transform shape
]

\node[vtx, fill=green!12] (s) at (0, 6) {$s$};

\node[vtx, fill=blue!10] (x1) at (-3, 3.5) {$x_1$};
\node[vtx, fill=blue!10] (x2) at (0, 3.5)  {$x_2$};
\node[vtx, fill=blue!10] (x3) at (3, 3.5)  {$x_3$};

\node[vtx, fill=orange!12] (y1) at (-3, 1) {$y_1$};
\node[vtx, fill=orange!12] (y2) at (0, 1)  {$y_2$};
\node[vtx, fill=orange!12] (y3) at (3, 1)  {$y_3$};

\node[vtx, fill=red!12] (t) at (0, -1.5) {$t$};

\draw[sstar]  (s) -- (x1) node[midway, left, font=\small] {$1$};
\draw[sstar] (s) -- (x2) node[midway, right, font=\small] {$2$};
\draw[sstar] (s) -- (x3) node[midway, right, font=\small] {$~~3$};

\draw[inserted] (x1) to[bend right=8] node[midway, right, font=\small, green!50!black]{$5~~$}  (t);

\draw[stconn] (s) to[bend left=90,looseness=2.35] node[midway, right, font=\small] {$7$} (t);

\draw[tstar] (t) -- (y1) node[midway, left, font=\small] {$4~$};
\draw[failed, black!45] (t) -- (y2) node[midway, right, font=\small, black!60] {$5$};
\draw[tstar] (t) -- (y3) node[midway, right, font=\small] {$6$};

\draw[matrix] (x1) -- (y2) node[midway, right, font=\small, blue!70]{$1$};

\node[font=\small, darkgray] at (-6.2, 6)   {Layer 1};
\node[font=\small, darkgray] at (-6.2, 3.5) {Layer 2};
\node[font=\small, darkgray] at (-6.2, 1)   {Layer 3};
\node[font=\small, darkgray] at (-6.2,-1.5) {Layer 4};
\end{tikzpicture}
\caption{Depiction of lower bound for colorful spanning tree for a query that deletes $(y_2, t)$ (dotted) and adds $(x_1,t)$ with color $5$ (dashed).}
\label{fig:colorful-spanning-tree}
\end{figure}

Let $G + U$ be the updated graph. To form a valid colorful spanning tree in $G+U$, exactly one edge from each of the $n-1$ color classes must be selected. We argue below that $G+U$ has a colorful spanning tree if and only if $M[i,j]=1$.

If $M[i,j]=1$, then the following is a colorful spanning tree of $G+U$: all edges incident to dummy vertices, edge $(s,t)$, edges $(t,y_k)$ for $k\ne j$, edges $(s,x_r)$ for $r \ne i$, the inserted edge $(t,x_i)$ of color $p+j$, and edge $(x_i,y_j)$ of color $i$. It is spanning, acyclic, and uses each color exactly once.

Now let $T$ be a colorful spanning tree of $G+U$. Since $T$ contains one edge per color, it contains all edges incident to dummy vertices, edge $(s,t)$, edge $(t,y_k)$ for $k \ne j$, and the inserted edge $(t,x_i)$ (the unique edge of color $p+j$ after the update). 
The vertex $y_j$ is incident only to matrix edges, so $T$ must contain one or more edge of the form $(x_a,y_j)$, where $a\in [p]$. 
Let $\A$ be collection of indices $a\in [p]$ such that $(x_a,y_j)$ is an edge in $T$. If $i\in \A$, then we have proved $M[i,j]=1$. If $i\notin \A$, then the component of $y_j$ in $T$ only consists of $y_j$ together with degree-one vertices from $X\setminus\{x_i\}$ and cannot contain $t$. Indeed, each $x_a\in X$ must have exactly one edge of color $a$ incident to it in $T$.
This contradicts connectivity of $T$, and proves that a colorful spanning tree cannot exist if $M[i,j]=0$.

For each $(i, j) \in [p]^2$, the oracle's answer uniquely determines the entry $M[i, j]$. The oracle must therefore distinguish between $2^{p^2}$ configurations, establishing a space lower bound of $p^2 = \Omega(\min(\sigma^2, n^2))$ bits.
\end{proof}

We finally show that any oracle which reports the maximum matching size of a bipartite graph, after $f=2$ edge updates, requires at least $\Omega(\min(\sigma^2, n^2))$ space, irrespective of the query time.

\begin{theorem}[Bipartite Matching]
\label{thm:matching}
Let $n, \sigma$ be positive integers with $4 \le \sigma \le n(n-1)$. There exists an $n$-vertex bipartite graph $G=(L \cup R, E)$ and a susceptible set $E_Q \subseteq E$ of size $\Theta(\sigma)$ such that any subset sensitivity oracle for maximum bipartite matching requires $\Omega(\min(\sigma^2, n^2))$ bits of space under $f=2$ updates.
\end{theorem}

\begin{proof}
Let $p=\min(\lfloor\sigma/4\rfloor,\lfloor n/4\rfloor)$ and let $K\in\{0,1\}^{p\times p}$ be an arbitrary binary matrix. Define a bipartite graph $G$ with bipartitions
$$
L=\{l_1,\ldots,l_{2p}\}, \qquad
R=\{r_1,\ldots,r_{2p}\},
$$
and edge set
$$
E=\{(l_i,r_j):K[i,j]=1\}\cup M^*,
$$
where
$$
M^*=\{(l_i,r_{p+i}),(l_{p+i},r_i):i\in[p]\}.
$$
The remaining $n-4p$ vertices, if any, are isolated. Let $E_Q=M^*$, so $|E_Q|=2p\le\sigma$.
Observe that $M^*$ is the unique perfect matching of $G$. Now consider the deletion update
$F=\{(l_i,r_{p+i}),(l_{p+j},r_j)\}\subseteq E_Q$.
The deleted edges isolate the vertices $r_{p+i}$ and $l_{p+j}$, so every matching in $G-F$ has size at most $2p-1$. See \Cref{fig:dual-FT-match-LB}.

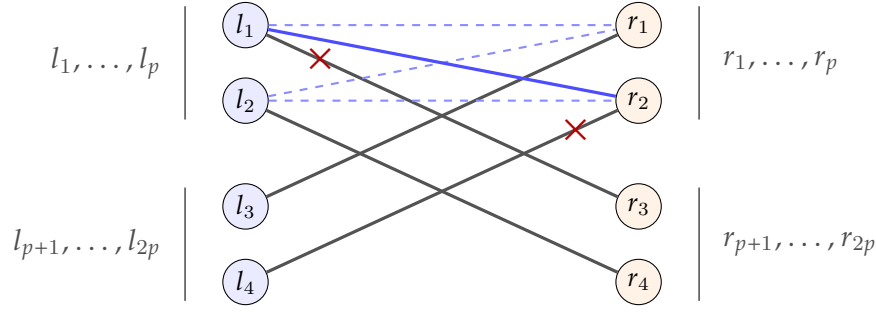
\begin{figure}[!ht]
\centering
\begin{tikzpicture}[
    Lvtx/.style={circle, draw, minimum size=6mm, inner sep=0pt, font=\small, fill=blue!8},
    Rvtx/.style={circle, draw, minimum size=6mm, inner sep=0pt, font=\small, fill=orange!10},
    mstar/.style={very thick, gray!65!black}, 
    bedge/.style={dashed, thick, blue!45},      
    comp/.style={very thick, blue!70},      
]

\foreach \k/\y in {1/3,2/2,3/0.6,4/-0.4}{
  \node[Lvtx] (l\k) at (-2.6,\y) {$l_{\k}$};
  \node[Rvtx] (r\k) at ( 2.6,\y) {$r_{\k}$};
}

\draw[mstar] (l1) -- (r3);
\draw[mstar] (l2) -- (r4);
\draw[mstar] (l3) -- (r1);
\draw[mstar] (l4) -- (r2);

\draw[bedge] (l1) -- (r1);   
\draw[bedge] (l1) -- (r2);   
\draw[bedge] (l2) -- (r1);
\draw[bedge] (l2) -- (r2);

\node at ($(l1)!0.19!(r3)$) {\textcolor{red!70!black}{\Large $\times$}};
\node at ($(r2)!0.16!(l4)$) {\textcolor{red!70!black}{\Large $\times$}};

\draw[comp] (l1) -- (r2);

\draw (-3.4,3.25) -- (-3.4,1.75)
      node[midway,left=5pt, gray!60!black]{$l_1,\dots,l_p$};
\draw (-3.4,0.85) -- (-3.4,-0.65)
      node[midway,left=5pt, gray!60!black]{$l_{p+1},\dots,l_{2p}$};
\draw (3.4,3.25) -- (3.4,1.75)
      node[midway,right=5pt, gray!60!black]{$r_1,\dots,r_p$};
\draw (3.4,0.85) -- (3.4,-0.65)
      node[midway,right=5pt, gray!60!black]{$r_{p+1},\dots,r_{2p}$};
\end{tikzpicture}
\caption{Depiction of space lower bound for matching  under $f=2$ updates.
Query $(1,2)$ deletes $(l_1,r_3)$ and $(l_4,r_2)$, isolating $r_3$ and $l_4$.
The matching in updated graph has size $2p{-1}$ iff $K[1,2]=1$.
}
\label{fig:dual-FT-match-LB}
\end{figure}

We claim that the maximum matching in $G-F$ has size $2p-1$ if and only if $K[i,j]=1$. Indeed, if $K[i,j]=1$, then $(l_i,r_j)\in E$, and replacing the two deleted matching edges by $(l_i,r_j)$ yields a matching of size $2p-1$. Conversely, suppose $G-F$ admits a matching of size $2p-1$. Since $r_{p+i}$ and $l_{p+j}$ are isolated, every other vertex must be matched. In particular, all remaining degree-one vertices are forced to use their unique incident edges, leaving $l_i$ and $r_j$ to be matched together. Hence $(l_i,r_j)\in E$, implying $K[i,j]=1$.

Thus, each query reveals the corresponding matrix entry $K[i,j]$. Since correctly answering all $p^2$ queries uniquely determines $M$ and $p=\Theta(\min\{\sigma,n\})$, any subset sensitivity oracle requires
$\Omega(p^2)=\Omega\left(\min\{\sigma^2,n^2\}\right)$
bits of space.
\end{proof}

\newcommand{\etalchar}[1]{$^{#1}$}

\appendix

\section*{Appendix}

\section{Omitted Proofs}

\subsection{Subset Sensitivity Oracle for Matrix Determinant}
\label{subsec:det}

\begin{lemma}
\label{lemma:det}
Let $A, B \in \F_d[Y]^{r \times m}~(r\leq m)$,
$D \in \F_d[Y]^{m \times m}$ be a diagonal matrix and let
$L = ADB^T$. Let $Q \subseteq [m]$ be a designated
\emph{susceptible set} of size $\sigma$, identified at preprocessing
time, such that all future queries modify only diagonal entries of $D$
indexed by $Q$.

Then there is an algorithm that preprocesses $A$, $B$, and $D$ to
construct a data structure such that, given a diagonal matrix $D'$ differing from $D$ in at most $f$ positions (all within $Q$), 
computes $\det(AD'B^T)$ with the following complexity bounds:

\begin{enumerate}
\item \textbf{Scalar ($d=0$):} If the matrices $A,B$ have entries in $\F$, then the preprocessing time is $O(\sigma^2 r^{\omega-2} + r^\omega)$, the space requirement 
is $O(\sigma^2)$, and the query time is $O(f^\omega)$.

\item \textbf{Polynomial ($d \ge 1$):} If the matrices have entries in
$\F_d[Y]$, then the preprocessing time is $\widetilde{O}(d\,\sigma^2 r^{\omega-1} + d\,r^{\omega+1})$, the space requirement is $\widetilde{O}(d\,r\,\sigma^2)$, and the query time is
$\widetilde{O}(d\,r\,f^{\omega})$.
\end{enumerate}
\end{lemma}

\begin{proof}
Let $A, B \in \F_d[Y]^{r \times m}$ with $r\leq m$, and let $D \in \F_d[Y]^{m \times m}$. Let $Q \subseteq [m]$ be the given \emph{susceptible set} of size $\sigma$.
Let $L = A D B^T \in \F_d[Y]^{r \times r}$. For a query $D'$ that differs from $D$ on a set $S \subseteq Q$ with $|S| \le f$, let $\Delta_S$ be the $f \times f$ diagonal matrix with diagonal entries $D'_{i_j i_j} - D_{i_j i_j}$ for $i_j \in S$, and let $A_S,B_S$ be the submatrices of $A,B$ with columns indexed by $S$. The updated matrix is:
\[
L' = A D' B^T = L + A_S \Delta_S B_S^T.
\]
Let $\delta = r - \operatorname{rank}(L) \ge 0$. Since $\operatorname{rank}(A_S \Delta_S B_S^T) \le f$, we have:
\[
\operatorname{rank}(L') \le (r - \delta) + f.
\]
If $\delta > f$, then $\operatorname{rank}(L') < r$ for all valid queries, so $\det(L') = 0$ identically. In this case, we can answer every query in $O(1)$ time and we are done. Thus, we assume $\delta \le f$.

Take random matrices $U,V \in \F_d^{r \times \delta}$ and let $L_{\mathrm{pad}} = L + U V^T$. Since the rank deficiency of $L$ is exactly $\delta$, there exist $U,V$ for which $L_{\mathrm{pad}}$ is nonsingular. By the Schwartz-Zippel lemma, with high probability a random choice of $U,V$ yields $\det(L_{\mathrm{pad}}) \neq 0$, allowing us to assume $L_{\mathrm{pad}}$ is invertible. (Note that we assume the field $\F$ is sufficiently large, in particular of size $\mathrm{poly}(r)$, which ensures that a random choice succeeds with high probability directly over $\F$).

For any query, we can rewrite $L'$ as:
\[
L' = L_{\mathrm{pad}} + \begin{pmatrix} A_S & U \end{pmatrix} \begin{pmatrix} \Delta_S & 0 \\ 0 & -I_\delta \end{pmatrix} \begin{pmatrix} B_S^T \\ V^T \end{pmatrix}.
\]
Let $X = \begin{pmatrix} A_S & U \end{pmatrix}$, $Y = \begin{pmatrix} B_S & V \end{pmatrix}$, and $M_S = \begin{pmatrix} \Delta_S & 0 \\ 0 & -I_\delta \end{pmatrix}$. Applying the Matrix Determinant lemma to the invertible matrix $L_{\mathrm{pad}}$ yields:
\begin{equation}
\label{eq:det_update_appendix}
\det(L') = \det(L_{\mathrm{pad}}) \det\!\bigl(I_{f+\delta} + M_S Y^T L_{\mathrm{pad}}^{-1} X\bigr).
\end{equation}

\paragraph{Preprocessing.}
We compute $L$, its rank, and $\delta$. If $\delta > f$, we store a flag to return $0$ for all queries. Otherwise, we sample $U,V$ until $\det(L_{\mathrm{pad}}) \ne 0$, then compute and store $\det(L_{\mathrm{pad}})$.

Because all queries are restricted to the susceptible set $Q$ of size $\sigma$, the matrix $Y^T L_{\mathrm{pad}}^{-1} X$ will always be a principal submatrix of a larger precomputed matrix restricted to $Q$ and the pad. We construct and store this $(\sigma+\delta) \times (\sigma+\delta)$ matrix:
\[
W_Q := \begin{pmatrix} B_Q & V \end{pmatrix}^T L_{\mathrm{pad}}^{-1} \begin{pmatrix} A_Q & U \end{pmatrix}.
\]
Note that since $\delta \le f \le \sigma$, the dimensions of $W_Q$ are $O(\sigma \times \sigma)$.
\begin{itemize}
\item \textbf{For $d=0$:} Inverting $L_{\mathrm{pad}}$ costs $O(r^\omega)$ operations. Multiplying the $r \times r$ matrix $L_{\mathrm{pad}}^{-1}$ by the $r \times (\sigma+\delta)$ matrix $(A_Q \mid U)$ requires $O(\sigma r^{\omega-1} + r^\omega)$ time using rectangular matrix multiplication. Multiplying $(B_Q \mid V)^T$ by the result takes $O(\sigma^2 r^{\omega-2} + r^\omega)$ time. Thus, total preprocessing is $O(\sigma^2 r^{\omega-2} + r^\omega)$, and storing $W_Q$ requires $O(\sigma^2)$ space.
\item \textbf{For $d \ge 1$:} $L_{\mathrm{pad}}$ is an $r \times r$ matrix with entries of degree $d$. Thus, $L_{\mathrm{pad}}^{-1}$ has entries that are rational functions of degree $O(dr)$. From~\cite{Storjohann15}, computing $L_{\mathrm{pad}}^{-1}$ costs $\widetilde{O}(d\,r^{3})$ operations. By the same rectangular multiplication bounds scaled for polynomials, computing $W_Q$ requires $\widetilde{O}(d\,\sigma^2 r^{\omega-1} + d\,r^{\omega+1})$ time. Storing the $O(\sigma \times \sigma)$ matrix $W_Q$ requires $\widetilde{O}(d\,r\,\sigma^2)$ space.
\end{itemize}

\paragraph{Query.}
If we stored $\delta > f$, we immediately answer $0$. Otherwise, let $I_S$ be the index set corresponding to the queried columns $S$ alongside the $\delta$ pad indices. The required $(f+\delta) \times (f+\delta)$ matrix $Y^T L_{\mathrm{pad}}^{-1} X$ is exactly the submatrix $(W_Q)_{I_S,I_S}$. 

We extract $(W_Q)_{I_S,I_S}$, multiply it by $M_S$, add the identity matrix $I_{f+\delta}$, and compute the determinant of the resulting $(f+\delta) \times (f+\delta)$ matrix as per Equation \eqref{eq:det_update_appendix}. 
\begin{itemize}
\item \textbf{For $d = 0$:} Since $f+\delta \le 2f$, evaluating this determinant costs $O(f^\omega)$ field operations.
\item \textbf{For $d \ge 1$:} The entries of $(W_Q)_{I_S,I_S}$ have degree $O(dr)$, so computing the determinant costs $\widetilde{O}(d\,r\,f^\omega)$ field operations. Multiplying the result by the precomputed scalar $\det(L_{\mathrm{pad}})$ completes the evaluation.
\end{itemize}
This proves the claim.
\end{proof}

\subsection{Extension of \Cref{lemma:BrandS-det} to General matrices}

\begin{lemma}
\label{lemma:BrandS-det-General-Q}
Let $M \in \F_d[Y]^{N \times N}$ be a matrix (possibly singular). Let $Q \subseteq [N]$ be a designated \emph{susceptible set} of size $\sigma$, identified at preprocessing time, such that all future queries add a modification matrix $C \in \F_d[Y]^{N \times N}$ whose non-zero entries are entirely confined to $Q \times Q$.
Then, there is a randomized algorithm that preprocesses $M$ using $\widetilde{O}(dN^{3})$ operations to create a data structure occupying $\widetilde{O}(d N \sigma^2)$ space with the following guarantee: given $C$ with at most $\delta_C$ nonzero entries, one can query $\det(M+C)$ in $\widetilde{O}(dN \delta_C^{\omega})$ operations.

Moreover, if $d=0$, then the preprocessing time, space, and query time are $O(N^\omega)$, $O(\sigma^2)$, and $O(\delta_C^\omega)$ respectively.
\end{lemma}

\begin{proof}
Let $\delta = N - \operatorname{rank}(M) \ge 0$ be the rank deficiency of the base matrix $M$. The modification matrix $C$ has at most $\delta_C$ nonzero entries, meaning $\operatorname{rank}(C) \le \delta_C$. Therefore, the rank of the updated matrix $M' = M + C$ is bounded by:

$$\operatorname{rank}(M') \le (N - \delta) + \delta_C.$$

If $\delta > \delta_C$, then $\operatorname{rank}(M') < N$ for all possible queries, implying $\det(M+C) = 0$ identically. In this scenario, the algorithm flags this condition during preprocessing and answers all queries with $0$ in $O(1)$ time. Thus, we assume henceforth that $\delta \le \delta_C$.

We sample random matrices $U, V \in \F_d^{N \times \delta}$. Since the rank deficiency of $M$ is exactly $\delta$, a generic rank-$\delta$ correction makes the matrix nonsingular. By the Schwartz-Zippel lemma, assuming the field $\F_d$ is of size $\mathrm{poly}(N)$, a random choice of $U,V$ ensures that $M_{\mathrm{pad}} = M + UV^T$ is invertible with high probability.

Since the matrix $C$ contains at most $\delta_C$ nonzero entries, all located within $Q \times Q$, we can explicitly factor it as a low-rank product using only indices from $Q$. Let $k \le \delta_C$ be the exact number of nonzero entries in $C$, located at coordinates $(i_1, j_1), \dots, (i_k, j_k)$ where $i_m, j_m \in Q$. We decompose $C$ as $C = U_C \Delta_C V_C^T$, where $\Delta_C \in \F_d[Y]^{k \times k}$ is a diagonal matrix with its $m$-th diagonal entry being $C_{i_m, j_m}$. The matrices $U_C, V_C \in \F_d^{N \times k}$ are constructed from the standard basis vectors of $\F_d^N$: the $m$-th column of $U_C$ is $e_{i_m}$, and the $m$-th column of $V_C$ is $e_{j_m}$.

We can now express the query matrix $M'$ by introducing and immediately subtracting the numeric pad:

$$M' = M + C = M_{\mathrm{pad}} + C - UV^T = M_{\mathrm{pad}} + \begin{pmatrix} U_C & U \end{pmatrix} \begin{pmatrix} \Delta_C & 0 \\ 0 & -I_\delta \end{pmatrix} \begin{pmatrix} V_C^T \\ V^T \end{pmatrix}.$$

Let $X = \begin{pmatrix} U_C & U \end{pmatrix}$, $Y = \begin{pmatrix} V_C & V \end{pmatrix}$, and $\Sigma = \begin{pmatrix} \Delta_C & 0 \\ 0 & -I_\delta \end{pmatrix}$. Applying the Matrix Determinant lemma to the invertible matrix $M_{\mathrm{pad}}$ yields:
\begin{equation}
\label{eq:adjoint_det_update_Q}
\det(M+C) = \det(M_{\mathrm{pad}}) \det\bigl(I_{k+\delta} + \Sigma Y^T M_{\mathrm{pad}}^{-1} X\bigr).
\end{equation}

\paragraph{Preprocessing.}
We compute the rank of $M$ and determine $\delta$. If $\delta > \delta_C$, we store a short-circuit flag. Otherwise, we sample $U, V$ until $\det(M_{\mathrm{pad}}) \neq 0$. We then compute $\det(M_{\mathrm{pad}})$ and the complete inverse matrix $M_{\mathrm{pad}}^{-1}$.

Because all queries restrict the non-zero entries of $C$ to the susceptible set $Q$ of size $\sigma$, the selection matrices $U_C$ and $V_C$ only ever sample rows and columns from $Q$. Thus, the inner matrix $Y^T M_{\mathrm{pad}}^{-1} X$ will always be a principal submatrix of a larger precomputed matrix restricted to $Q$ and the pad. We construct and store this $(\sigma+\delta) \times (\sigma+\delta)$ matrix:

$$W_Q := \begin{pmatrix} I_{N, Q} & V \end{pmatrix}^T M_{\mathrm{pad}}^{-1} \begin{pmatrix} I_{N, Q} & U \end{pmatrix},$$

where $I_{N, Q}$ represents the $N \times \sigma$ submatrix of the identity matrix $I_N$ consisting of the columns indexed by $Q$. Note that since $\delta \le \delta_C \le \sigma$, the dimensions of $W_Q$ are $O(\sigma \times \sigma)$.
\begin{itemize}
\item \textbf{For $d=0$:} Computing $M_{\mathrm{pad}}^{-1}$ costs $O(N^\omega)$ arithmetic operations. Extracting $W_Q$ and storing it requires $O(\sigma^2)$ space.
\item \textbf{For $d \ge 1$:} The entries of $M_{\mathrm{pad}}$ have degree $d$, meaning the entries of $M_{\mathrm{pad}}^{-1}$ are rational functions of degree $O(dN)$. From \cite{Storjohann15}~(Corollary 6.12), computing this inverse costs $\widetilde{O}(dN^{3})$ operations. Storing the $O(\sigma) \times O(\sigma)$ matrix $W_Q$ requires $\widetilde{O}(dN \sigma^2)$ space.
\end{itemize}

\paragraph{Query.}
Given a modification matrix $C$, if the short-circuit flag $\delta > \delta_C$ is active, we return $0$. Otherwise, the query algorithm identifies the $k \le \delta_C$ coordinates of the nonzero entries to form the standard basis selection matrices $U_C$ and $V_C$.

To form the inner matrix $Y^T M_{\mathrm{pad}}^{-1} X$, we simply extract the relevant rows and columns from our precomputed matrix $W_Q$. Because $U_C$ and $V_C$ map directly to indices within $Q$, this amounts to an $O(\delta_C) \times O(\delta_C)$ submatrix lookup, alongside the static $U$ and $V$ blocks. Since $\delta \le \delta_C$, the total dimension of this inner matrix is at most $2\delta_C \times 2\delta_C$.

We multiply this submatrix by $\Sigma$, add $I_{k+\delta}$, and evaluate its determinant.
\begin{itemize}
\item \textbf{For $d = 0$:} Computing the determinant of this $O(\delta_C) \times O(\delta_C)$ scalar matrix costs $O(\delta_C^\omega)$ field operations.
\item \textbf{For $d \ge 1$:} The entries extracted from $W_Q$ have rational degree $O(dN)$. Computing the determinant of an $O(\delta_C) \times O(\delta_C)$ matrix with entries of degree $O(dN)$ takes $\widetilde{O}(dN \delta_C^\omega)$ operations.
\end{itemize}

Multiplying this final determinant by the precomputed scalar $\det(M_{\mathrm{pad}})$ completes the query procedure.
\end{proof}

\subsection{Proof of~\Cref{lemma:union_rep}}
\begin{proof}
We treat the diagonal entries $\{c_{i,e} : i \in [k],\, e \in E\}$ of the matrices $C_1, \ldots, C_k$ as formal variables and work over the polynomial ring $\F[\{c_{i,e}\}]$.

We first show that $S \subseteq E$ is independent in $M^{\vee k}$ if and only if the submatrix $U_S$ has full column rank over the polynomial ring (i.e., some maximal minor of $U_S$ is a non-zero polynomial).

Fix $S \subseteq E$ with $|S| = s$. Suppose $S$ is independent in $M^{\vee k}$. By definition, there is a partition $S = S_1 \uplus \dots \uplus S_k$ such that each $S_i$ is independent in $M$. Since $M$ is represented by $A$, for each $i \in [k]$ there exists a subset $R_i \subseteq [r]$ with $|R_i| = |S_i|$ such that $\det(A_{R_i, S_i}) \neq 0$.

Let $R$ be the set of $s$ rows in $U$ formed by taking the rows corresponding to $R_i$ from the $i$-th block  $A C_i$ (more formally, the set of the rows in $R$ is given by $\{(i-1)r + j: i \in [k], j \in R_i\}$).

Consider the $s \times s$ square submatrix $U_{R,S}$. By the generalized Laplace expansion
(Lemma~\ref{lemma:generalised-laplace-expansion}), on expanding along
the $k$ row blocks $R_1, \dots, R_k$, $\det(U_{R,S})$ can be written as:
\begin{equation}
\label{eq:detURS}
\det(U_{R,S}) \;=\; \sum_{\substack{T_1 \uplus \dots \uplus T_k = S \\ |T_i| = |R_i|}} \pm\, \prod_{i=1}^k \det\bigl((A C_i)_{R_i, T_i}\bigr).
\end{equation}

For each $i$ and block $T_i$, we can factor out the diagonal variables:
$$
\det\bigl((A C_i)_{R_i, T_i}\bigr) = \det\bigl(A_{R_i, T_i}\bigr) \prod_{e \in T_i} c_{i,e},
$$
where $c_{i,e}$ denotes the $(e,e)$ diagonal entry of the matrix $C_i$.

For two distinct partitions $(T_1, \ldots, T_k)$ and $(T_1', \ldots, T_k')$, the monomials $\prod_{i \in [k]} \prod_{e \in T_i} c_{i,e}$ and $\prod_{i \in [k]} \prod_{e \in T_i'} c_{i,e}$ are distinct (since the variable indices $(i, e)$ differ). Therefore, distinct partitions contribute linearly independent monomials, and $\det(U_{R,S})$ is a non-zero polynomial if at least one partition yields a non-zero coefficient. The partition $(S_1, \ldots, S_k)$ contributes the coefficient $\prod_{i=1}^k \det(A_{R_i, S_i}) \neq 0$. Therefore, $\det(U_{R,S})$ is a non-zero polynomial.

Conversely, suppose the submatrix $U_S$ has full column rank. Then $\det(U_{R,S})$ is a non-zero polynomial for some choice of rows $R = R_1 \uplus \dots \uplus R_k$. By~\eqref{eq:detURS}, there exists a partition $(T_1, \dots, T_k)$ with $\prod_{i=1}^k \det(A_{R_i, T_i}) \neq 0$, implying each $T_i$ is independent in $M$. Hence $S$ is independent in $M^{\vee k}$.

If $S$ is independent in $M^{\vee k}$, we have shown that $P_S := \det(U_{R,S})$ is a non-zero polynomial for some $R$. The degree of $P_S$ is at most $s \le kr$ (each variable appears with degree at most $1$). By the Schwartz-Zippel lemma, under a uniformly random assignment from $\F$ to the variables:
$$
\Pr[P_S = 0] \;\le\; \frac{\deg(P_S)}{|\F|} \;\le\; \frac{kr}{|\F|}.
$$
This proves the first claim as $\F$ is of size $\poly(kr)$.
To establish the second claim, note that any set $S \subseteq E$ which is dependent in $M^{\vee k}$ yields the identically zero polynomial for all choices of $R$. So, such sets remain deterministically dependent under any field assignment.
\end{proof}

\subsection{Extension of Perfect $\alpha$-factor Sensitivity Oracle to Maximum $\alpha$-factor}
\label{subsec:max-alpha-factor}

We extend the sensitivity oracle results of the (perfect) $\alpha$-factor (\Cref{thm:sensitivity-arbitrary-alpha,thm:subset-arbitrary-alpha}) to the (maximum) $\alpha$-factor, achieving near-identical bounds. 
Using the terminology from~\cite{GabowS21_I}, the (maximum) $\alpha$-factor treats $\alpha(v)$ as an upper bound on the desired degree of $v$, and the objective is to maximize the total number of edges picked. In contrast, the (perfect) $\alpha$-factor treats $\alpha(v)$ as the exact required degree for $v$, and the oracle strictly answers its existence. 

Given a graph $G = (V, E)$, the function $\alpha$, edge costs $\cost: E \rightarrow [-W, W]$, a precomputed min-cost (maximum) $\alpha$-factor $M$ of $G$, and a parameter $f > 0$, we first define the \emph{deficiency} $def(v)$ for any vertex $v \in V$ as the difference between $\alpha(v)$ and the number of edges incident to $v$ in $M$. 
Let $\delta_M(v)$ denote the set of edges incident to $v$ in $M$. Hence, $def(v) = \alpha(v) - |\delta_M(v)|$. By definition, $def(v) \geq 0$. 

We now construct an auxiliary undirected graph $\G$ from $G$ by adding three dummy vertices $z_1, z_2, z_3$, along with the following sets of zero-cost edges: 
\begin{itemize}
    \item $E_{z_1}$: For all $v \in V$, add $def(v) + f$ parallel edges between $z_1$ and $v$.
    \item $E_{triangle}$: Add $f$ parallel edges between each of the pairs $(z_1, z_2), (z_2, z_3),$ and $(z_3, z_1)$.
\end{itemize} 

This yields the graph $\G = (\V, \E)$, where $\V = V \cup \{z_1, z_2, z_3\}$ and $\E = E \cup E_{z_1} \cup E_{triangle}$. 

We now define a function $\beta$ over the vertices of $\G$ as follows: $\beta(z_1) = 2f + \sum_{v \in V} def(v)$; $\beta(z_2) = \beta(z_3) = 2f$; and $\beta(v) = \alpha(v)$ for all $v \in V$. By construction, $\G$ admits a (perfect) $\beta$-factor using a subset of the edges of $M$, $E_{z_1}$, and $E_{triangle}$. Furthermore, the cost of the min-cost maximum $\alpha$-factor of $G$ is exactly equal to the cost of the min-cost perfect $\beta$-factor of $\G$.

For any set of edge updates $U$ in graph $G$, we describe the corresponding updates $\U$ to the graph $\G$. For an update to any edge $e \in G$, we apply the identical update to $e$ in $\G$. We then augment the set $\U$ into a parameterized set $\U^i$ for $i \in [-f, f]$ with the following additional dummy-edge updates: 
\begin{itemize}
    \item If $i > 0$: delete $i$ edges between $z_2$ and $z_3$; add $i$ zero-cost edges between both $(z_1,z_3)$ and $(z_1, z_2)$.  
    \item If $i \le 0$: add $|i|$ zero-cost edges between $z_2$ and $z_3$; delete $|i|$ edges between both $(z_1,z_3)$ and $(z_1, z_2)$.  
\end{itemize}

The following lemma establishes the equivalence between the two factors under these updates:

\begin{lemma}
    \label{lemma:maximum-to-perfect-alpha}
    The cardinality of the maximum $\alpha$-factor of $G+U$ is equal to $|M| + i^*$, where $i^*$ is the maximum possible integer $i \in [-f,f]$ such that the graph $\G + \U^i$ admits a perfect $\beta$-factor. Furthermore, the cost of the minimum-cost maximum $\alpha$-factor of $G + U$ is equal to the cost of the minimum-cost perfect $\beta$-factor of $\G + \U^{i^*}$.
\end{lemma}
\begin{proof}    
    Let $M'$ be an arbitrary maximum $\alpha$-factor in $G+U$. Since $f$ updates can change the size of the maximum $\alpha$-factor by at most $f$, we know that $|M| - f \leq |M'| \leq |M| + f$. Without loss of generality, we can assume that for all $v \in V$, the local change is bounded: $| |\delta_M(v)| - |\delta_{M'}(v)| | \leq f$. 
    
    Observe that $M'$ can be extended to a perfect $\beta$-factor $M_{i}$ in $\G + \U^{i}$, where $i = |M'| - |M|$. This is achieved by picking exactly $\sum_{v \in V} \alpha(v) - 2|M'|$ edges from $E_{z_1}$ (which is always possible since $| |\delta_M(v)| - |\delta_{M'}(v)| | \leq f$). This leaves a remaining demand of $2f + 2(|M'| - |M|)$ at vertex $z_1$, which must be strictly satisfied using edges from $E_{triangle}$.
    
    Since $z_2$ and $z_3$ have fixed capacities of $2f$, this remaining demand uniquely forces $M_{i}$ to pick exactly $f + (|M'| - |M|)$ edges on both $(z_1, z_2)$ and $(z_1, z_3)$, and exactly $f - (|M'| - |M|)$ edges on $(z_2, z_3)$. This assignment forms a valid perfect $\beta$-factor in $\G + \U^i$. 
    
    Thus, the maximum possible $i$ (denoted $i^*$) that yields a perfect $\beta$-factor in $\G + \U^{i}$ perfectly corresponds to the maximum $\alpha$-factor of $G+U$. This proves the cardinality equivalence. Finally, because all dummy edges carry zero cost, the min-cost perfect $\beta$-factor strictly optimizes and matches the cost of the min-cost maximum $\alpha$-factor.
\end{proof}

We can now construct any perfect $\beta$-factor sensitivity oracle over the graph $\G$. By \Cref{lemma:maximum-to-perfect-alpha}, this oracle resolves sensitivity queries for the maximum $\alpha$-factor in $G$. When an update $\U$ is queried, the optimal integer $i^*$ can be found using a binary search over the interval $[-f, f]$. This binary search introduces a logarithmic overhead of $O(\log f)$, mapping the previous $O(f^{3\omega})$ query time to $\widetilde{O}(f^{3\omega})$.

Hence, the following corollaries of \Cref{thm:sensitivity-arbitrary-alpha} and \Cref{thm:subset-arbitrary-alpha} are immediate:

\begin{corollary}
    \label{corollary:sensitivity-arbitrary-max-alpha}
    Let $G = (V,E,\cost)$ be an undirected graph with integer edge costs in $[-W,W]$ and arbitrary degree bounds $\alpha$, and let $M$ be a fixed initial (maximum) $\alpha$-factor. Then there exist the following sensitivity oracles that, for any update set $U$ with $f = |U|$, answer queries on $G+U$ with high probability:
    \begingroup
    \setlength{\aboverulesep}{0pt}
    \setlength{\belowrulesep}{0pt}
    \begin{center}
    \renewcommand{\arraystretch}{1.25}
    \resizebox{0.98\linewidth}{!}{
    \begin{tabular}{@{}l l l l l@{}}
    \rowcolor{yellow!15}
    \textsc{Oracle} & \textsc{Prep.} & \textsc{Space} & \textsc{Query} & \textsc{Output} \\
    \midrule
    Maximum $\alpha$-factor &
    $O((fn)^\omega + T_{\alpha})$ &
    $O((fn)^2 \log n)$ &
    $\widetilde{O}(f^{3\omega})$ &
    Size of maximum $\alpha$-factor in $G+U$ \\
    Min-Cost Max $\alpha$-factor &
    $\widetilde{O}(W(fn)^3)$ &
    $O(W(fn)^3 \log n)$ &
    $\widetilde{O}(W n f^{3\omega+1})$ &
    Cost of min-cost max $\alpha$-factor in $G+U$ \\
    \end{tabular}%
    }
    \end{center}
    \endgroup
\end{corollary}

\begin{corollary}
    \label{corollary:subset-arbitrary-alpha}
    Let $G = (V,E,\cost)$ be an undirected graph with integer edge costs in $[-W,W]$, arbitrary degree bounds $\alpha$, a fixed initial min-cost maximum $\alpha$-factor $M$, and a susceptible edge set $E_Q$ of size $\sigma$. Then there exist the following subset sensitivity oracles that, for any update set $U \subseteq E_Q$ with $f = |U|$, answer queries on $G+U$ with high probability:
    
    \begingroup
    \setlength{\aboverulesep}{0pt}
    \setlength{\belowrulesep}{0pt}
    \begin{center}
    \renewcommand{\arraystretch}{1.25}
    \resizebox{0.98\linewidth}{!}{
    \begin{tabular}{@{}lllll@{}}
    \rowcolor{yellow!15}
    \textsc{Oracle} & \textsc{Prep.} & \textsc{Space} & \textsc{Query} & \textsc{Output} \\
    \midrule
    Maximum $\alpha$-factor &
    $O(\sigma^2 (fn)^{\omega-2} + (fn)^\omega + T_{\alpha})$ &
    $O(\sigma^2)$ &
    $\widetilde{O}(f^\omega)$ &
    Size of maximum $\alpha$-factor in $G+U$ \\
    Min-Cost Max $\alpha$-factor &
    $\widetilde{O}(W \sigma^2 (fn)^{\omega-1} + W(fn)^{\omega+1})$ &
    $\widetilde{O}(W fn \sigma^2)$ &
    $\widetilde{O}(W n f^{\omega+1} )$ &
    Cost of min-cost max $\alpha$-factor in $G+U$ \\
    \end{tabular}%
    }
    \end{center}
    \endgroup
\end{corollary}


\begin{thebibliography}{VDBCK{\etalchar{+}}24}

\bibitem[AB19]{AssadiB19}
Sepehr Assadi and Aaron Bernstein.
\newblock {Towards a Unified Theory of Sparsification for Matching Problems}.
\newblock In Jeremy~T. Fineman and Michael Mitzenmacher, editors, {\em 2nd
  Symposium on Simplicity in Algorithms (SOSA 2019)}, volume~69 of {\em Open
  Access Series in Informatics (OASIcs)}, pages 11:1--11:20, Dagstuhl, Germany,
  2019. Schloss Dagstuhl -- Leibniz-Zentrum f{\"u}r Informatik.

\bibitem[ABR24]{AzarmehrBR24}
Amir Azarmehr, Soheil Behnezhad, and Mohammad Roghani.
\newblock {\em Fully Dynamic Matching: $(2-\sqrt{2})$-Approximation in Polylog
  Update Time}, pages 3040--3061.
\newblock 2024.

\bibitem[ACP{\etalchar{+}}26]{AhiCPPS26}
Mridul Ahi, Keerti Choudhary, Shlok Pande, Pushpraj, and Lakshay Saggi.
\newblock {Maximum-Flow and Minimum-Cut Sensitivity Oracles for Directed
  Graphs}.
\newblock In Shubhangi Saraf, editor, {\em 17th Innovations in Theoretical
  Computer Science Conference (ITCS 2026)}, volume 362 of {\em Leibniz
  International Proceedings in Informatics (LIPIcs)}, pages 5:1--5:24,
  Dagstuhl, Germany, 2026. Schloss Dagstuhl -- Leibniz-Zentrum f{\"u}r
  Informatik.

\bibitem[AGI{\etalchar{+}}19]{AbboudGIKPTUW19}
Amir Abboud, Loukas Georgiadis, Giuseppe~F. Italiano, Robert Krauthgamer, Nikos
  Parotsidis, Ohad Trabelsi, Przemyslaw Uznanski, and Daniel Wolleb{-}Graf.
\newblock Faster algorithms for all-pairs bounded min-cuts.
\newblock In Christel Baier, Ioannis Chatzigiannakis, Paola Flocchini, and
  Stefano Leonardi, editors, {\em 46th International Colloquium on Automata,
  Languages, and Programming, {ICALP} 2019, July 9-12, 2019}, volume 132 of
  {\em LIPIcs}, pages 7:1--7:15, 2019.

\bibitem[AJ23]{AkmalJ23}
Shyan Akmal and Ce~Jin.
\newblock An efficient algorithm for all-pairs bounded edge connectivity.
\newblock In Kousha Etessami, Uriel Feige, and Gabriele Puppis, editors, {\em
  50th International Colloquium on Automata, Languages, and Programming,
  {ICALP} 2023, July 10-14, 2023, Paderborn, Germany}, volume 261 of {\em
  LIPIcs}, pages 11:1--11:20. Schloss Dagstuhl - Leibniz-Zentrum f{\"{u}}r
  Informatik, 2023.

\bibitem[Akm24]{Akmal24}
Shyan Akmal.
\newblock {\em An Enumerative Perspective on Connectivity}, pages 179--198.
\newblock 2024.

\bibitem[Ans85]{Anstee85}
R.P Anstee.
\newblock An algorithmic proof of tutte's f-factor theorem.
\newblock {\em Journal of Algorithms}, 6(1):112--131, 1985.

\bibitem[BBP22]{BaswanaBP22}
Surender Baswana, Koustav Bhanja, and Abhyuday Pandey.
\newblock Minimum+1 (s, t)-cuts and dual edge sensitivity oracle.
\newblock In Mikolaj Bojanczyk, Emanuela Merelli, and David~P. Woodruff,
  editors, {\em 49th International Colloquium on Automata, Languages, and
  Programming, {ICALP} 2022, July 4-8, 2022, Paris, France}, volume 229 of {\em
  LIPIcs}, pages 15:1--15:20. Schloss Dagstuhl - Leibniz-Zentrum f{\"{u}}r
  Informatik, 2022.

\bibitem[BCC{\etalchar{+}}23]{BiloCCCFKS23}
Davide Bil\`{o}, Shiri Chechik, Keerti Choudhary, Sarel Cohen, Tobias
  Friedrich, Simon Krogmann, and Martin Schirneck.
\newblock Approximate distance sensitivity oracles in subquadratic space.
\newblock In {\em Proceedings of the 55th Annual ACM Symposium on Theory of
  Computing}, STOC 2023, page 1396–1409, New York, NY, USA, 2023. Association
  for Computing Machinery.

\bibitem[BCDW24]{BansalCDW24}
Shivam Bansal, Keerti Choudhary, Harkirat Dhanoa, and Harsh Wardhan.
\newblock {Fault-Tolerant Bounded Flow Preservers}.
\newblock In Juli\'{a}n Mestre and Anthony Wirth, editors, {\em 35th
  International Symposium on Algorithms and Computation (ISAAC 2024)}, volume
  322 of {\em Leibniz International Proceedings in Informatics (LIPIcs)}, pages
  9:1--9:14, Dagstuhl, Germany, 2024. Schloss Dagstuhl -- Leibniz-Zentrum
  f{\"u}r Informatik.

\bibitem[BCFS21]{BiloCFS21}
Davide Bil\`{o}, Sarel Cohen, Tobias Friedrich, and Martin Schirneck.
\newblock {Near-Optimal Deterministic Single-Source Distance Sensitivity
  Oracles}.
\newblock In Petra Mutzel, Rasmus Pagh, and Grzegorz Herman, editors, {\em 29th
  Annual European Symposium on Algorithms (ESA 2021)}, volume 204 of {\em
  Leibniz International Proceedings in Informatics (LIPIcs)}, pages
  18:1--18:17, Dagstuhl, Germany, 2021. Schloss Dagstuhl -- Leibniz-Zentrum
  f{\"u}r Informatik.

\bibitem[BCR16]{BaswanaCR16}
Surender Baswana, Keerti Choudhary, and Liam Roditty.
\newblock Fault tolerant subgraph for single source reachability: generic and
  optimal.
\newblock In {\em Proceedings of the Forty-Eighth Annual ACM Symposium on
  Theory of Computing}, STOC '16, page 509–518, New York, NY, USA, 2016.
  Association for Computing Machinery.

\bibitem[BCR17]{BaswanaCR:17}
Surender Baswana, Keerti Choudhary, and Liam Roditty.
\newblock An efficient strongly connected components algorithm in the fault
  tolerant model.
\newblock In {\em 44th International Colloquium on Automata, Languages, and
  Programming, {ICALP} 2017, July 10-14, 2017, Warsaw, Poland}, pages
  72:1--72:15, 2017.

\bibitem[BCR19]{BaswanaCR19scc}
Surender Baswana, Keerti Choudhary, and Liam Roditty.
\newblock An efficient strongly connected components algorithm in the fault
  tolerant model.
\newblock {\em Algorithmica}, 81(3):967--985, 2019.

\bibitem[Beh23]{Behnezhad23}
Soheil Behnezhad.
\newblock {\em Dynamic Algorithms for Maximum Matching Size}, pages 129--162.
\newblock 2023.

\bibitem[BFGM25]{BentertFGM25}
Matthias Bentert, Fedor~V. Fomin, Petr~A. Golovach, and Laure Morelle.
\newblock {Fault-Tolerant Matroid Bases}.
\newblock In Anne Benoit, Haim Kaplan, Sebastian Wild, and Grzegorz Herman,
  editors, {\em 33rd Annual European Symposium on Algorithms (ESA 2025)},
  volume 351 of {\em Leibniz International Proceedings in Informatics
  (LIPIcs)}, pages 83:1--83:14, Dagstuhl, Germany, 2025. Schloss Dagstuhl --
  Leibniz-Zentrum f{\"u}r Informatik.

\bibitem[BGH{\etalchar{+}}24]{BuchbinderGHKS24}
Niv Buchbinder, Anupam Gupta, Daniel Hathcock, Anna~R. Karlin, and Sherry
  Sarkar.
\newblock {\em Maintaining Matroid Intersections Online}, pages 4283--4304.
\newblock 2024.

\bibitem[BGLP16]{Bilo16-stacs}
Davide Bil{\`{o}}, Luciano Gual{\`{a}}, Stefano Leucci, and Guido Proietti.
\newblock Multiple-edge-fault-tolerant approximate shortest-path trees.
\newblock In {\em 33rd Symposium on Theoretical Aspects of Computer Science,
  {STACS} 2016, February 17-20, 2016, Orl{\'{e}}ans, France}, pages
  18:1--18:14, 2016.

\bibitem[BGRS22]{BanerjeeGRS22}
Niranka Banerjee, Manoj Gupta, Venkatesh Raman, and Saket Saurabh.
\newblock Output sensitive fault tolerant maximum matching.
\newblock In {\em Computer Science – Theory and Applications: 17th
  International Computer Science Symposium in Russia, CSR 2022, Virtual Event,
  June 29 – July 1, 2022, Proceedings}, page 115–132, Berlin, Heidelberg,
  2022. Springer-Verlag.

\bibitem[Bha24]{Bhanja24}
Koustav Bhanja.
\newblock Optimal sensitivity oracle for steiner mincut.
\newblock In Juli{\'{a}}n Mestre and Anthony Wirth, editors, {\em 35th
  International Symposium on Algorithms and Computation, {ISAAC} 2024, December
  8-11, 2024, Sydney, Australia}, volume 322 of {\em LIPIcs}, pages
  10:1--10:18. Schloss Dagstuhl - Leibniz-Zentrum f{\"{u}}r Informatik, 2024.

\bibitem[Bha25]{Bhanja25}
Koustav Bhanja.
\newblock {Minimum+1 Steiner Cut and Dual Edge Sensitivity Oracle: Bridging Gap
  between Global and (s,t)-cut}.
\newblock In Keren Censor-Hillel, Fabrizio Grandoni, Jo\"{e}l Ouaknine, and
  Gabriele Puppis, editors, {\em 52nd International Colloquium on Automata,
  Languages, and Programming (ICALP 2025)}, volume 334 of {\em Leibniz
  International Proceedings in Informatics (LIPIcs)}, pages 27:1--27:20,
  Dagstuhl, Germany, 2025. Schloss Dagstuhl -- Leibniz-Zentrum f{\"u}r
  Informatik.

\bibitem[BK09]{BernsteinK09}
Aaron Bernstein and David~R. Karger.
\newblock {A Nearly Optimal Oracle for Avoiding Failed Vertices and Edges}.
\newblock In {\em Proceedings of the 41st Symposium on Theory of Computing
  (STOC)}, pages 101--110, 2009.

\bibitem[BKSW23]{BhattacharyaKSW23}
Sayan Bhattacharya, Peter Kiss, Thatchaphol Saranurak, and David Wajc.
\newblock {\em Dynamic Matching with Better-than-2 Approximation in
  Polylogarithmic Update Time}, pages 100--128.
\newblock 2023.

\bibitem[BMNT23]{BlikstadMNT23}
Joakim Blikstad, Sagnik Mukhopadhyay, Danupon Nanongkai, and Ta{-}Wei Tu.
\newblock Fast algorithms via dynamic-oracle matroids.
\newblock In Barna Saha and Rocco~A. Servedio, editors, {\em Proceedings of the
  55th Annual {ACM} Symposium on Theory of Computing, {STOC} 2023, Orlando, FL,
  USA, June 20-23, 2023}, pages 1229--1242. {ACM}, 2023.

\bibitem[BP22]{BaswanaP22}
Surender Baswana and Abhyuday Pandey.
\newblock Sensitivity oracles for all-pairs mincuts.
\newblock In Joseph~(Seffi) Naor and Niv Buchbinder, editors, {\em Proceedings
  of the 2022 {ACM-SIAM} Symposium on Discrete Algorithms, {SODA} 2022, Virtual
  Conference / Alexandria, VA, USA, January 9 - 12, 2022}, pages 581--609.
  {SIAM}, 2022.

\bibitem[BS15]{BernsteinS15}
Aaron Bernstein and Cliff Stein.
\newblock Fully dynamic matching in bipartite graphs.
\newblock In Magn{\'u}s~M. Halld{\'o}rsson, Kazuo Iwama, Naoki Kobayashi, and
  Bettina Speckmann, editors, {\em Automata, Languages, and Programming}, pages
  167--179, Berlin, Heidelberg, 2015. Springer Berlin Heidelberg.

\bibitem[BS16]{BernsteinS16}
Aaron Bernstein and Cliff Stein.
\newblock Faster fully dynamic matchings with small approximation ratios.
\newblock In {\em Proceedings of the Twenty-Seventh Annual ACM-SIAM Symposium
  on Discrete Algorithms}, SODA '16, page 692–711, USA, 2016. Society for
  Industrial and Applied Mathematics.

\bibitem[BS21]{BercziS21}
Krist{\'o}f B{\'e}rczi and Tam{\'a}s Schwarcz.
\newblock Complexity of packing common bases in matroids.
\newblock {\em Mathematical Programming}, 188(1):1--18, Jul 2021.

\bibitem[CC20]{ChechikC20}
Shiri Chechik and Sarel Cohen.
\newblock Distance sensitivity oracles with subcubic preprocessing time and
  fast query time.
\newblock STOC 2020, page 1375–1388, New York, NY, USA, 2020. Association for
  Computing Machinery.

\bibitem[CCFK17]{ChechikC:17}
Shiri Chechik, Sarel Cohen, Amos Fiat, and Haim Kaplan.
\newblock (1 + epsilon)-approximate \emph{f}-sensitive distance oracles.
\newblock In {\em Proceedings of the Twenty-Eighth Annual {ACM-SIAM} Symposium
  on Discrete Algorithms, {SODA} 2017, Barcelona, Spain, Hotel Porta Fira,
  January 16-19}, pages 1479--1496, 2017.

\bibitem[CCZ25]{ChandrasekaranCZ25}
Karthekeyan Chandrasekaran, Chandra Chekuri, and Weihao Zhu.
\newblock {Online Disjoint Spanning Trees and Polymatroid Bases}.
\newblock In Keren Censor-Hillel, Fabrizio Grandoni, Jo\"{e}l Ouaknine, and
  Gabriele Puppis, editors, {\em 52nd International Colloquium on Automata,
  Languages, and Programming (ICALP 2025)}, volume 334 of {\em Leibniz
  International Proceedings in Informatics (LIPIcs)}, pages 44:1--44:20,
  Dagstuhl, Germany, 2025. Schloss Dagstuhl -- Leibniz-Zentrum f{\"u}r
  Informatik.

\bibitem[CGM92]{CameriniGM92}
P.M Camerini, G~Galbiati, and F~Maffioli.
\newblock Random pseudo-polynomial algorithms for exact matroid problems.
\newblock {\em Journal of Algorithms}, 13(2):258--273, 1992.

\bibitem[Cho16]{Choudhary16}
Keerti Choudhary.
\newblock An optimal dual fault tolerant reachability oracle.
\newblock In {\em 43rd International Colloquium on Automata, Languages, and
  Programming, {ICALP} 2016}, pages 130:1--130:13, 2016.

\bibitem[CKL{\etalchar{+}}23]{ChenKLPGS23}
Li~Chen, Rasmus Kyng, Yang~P. Liu, Richard Peng, Maximilian~Probst Gutenberg,
  and Sushant Sachdeva.
\newblock Almost-linear-time algorithms for maximum flow and minimum-cost flow.
\newblock {\em Commun. ACM}, 66(12):85–92, November 2023.

\bibitem[CKL{\etalchar{+}}24]{ChenKLMP24}
Li~Chen, Rasmus Kyng, Yang~P. Liu, Simon Meierhans, and Maximilian
  Probst~Gutenberg.
\newblock Almost-linear time algorithms for incremental graphs: Cycle
  detection, sccs, s-t shortest path, and minimum-cost flow.
\newblock In {\em Proceedings of the 56th Annual ACM Symposium on Theory of
  Computing}, STOC 2024, page 1165–1173, New York, NY, USA, 2024. Association
  for Computing Machinery.

\bibitem[CLL11]{CheungLL11}
Ho~Yee Cheung, Lap~Chi Lau, and Kai~Man Leung.
\newblock Graph connectivities, network coding, and expander graphs.
\newblock In {\em 2011 IEEE 52nd Annual Symposium on Foundations of Computer
  Science}, pages 190--199, 2011.

\bibitem[CLL14]{CheungLL14}
Ho~Yee Cheung, Lap~Chi Lau, and Kai~Man Leung.
\newblock Algebraic algorithms for linear matroid parity problems.
\newblock {\em {ACM} Trans. Algorithms}, 10(3):10:1--10:26, 2014.

\bibitem[CLPR12]{ChechikLPR:12}
Shiri Chechik, Michael Langberg, David Peleg, and Liam Roditty.
\newblock f-sensitivity distance oracles and routing schemes.
\newblock {\em Algorithmica}, 63(4):861--882, 2012.

\bibitem[Cun86]{Cunningham86}
William~H. Cunningham.
\newblock Improved bounds for matroid partition and intersection algorithms.
\newblock {\em SIAM Journal on Computing}, 15(4):948--957, 1986.

\bibitem[DG24]{Dey024}
Dipan Dey and Manoj Gupta.
\newblock Near optimal dual fault tolerant distance oracle.
\newblock In Timothy~M. Chan, Johannes Fischer, John Iacono, and Grzegorz
  Herman, editors, {\em 32nd Annual European Symposium on Algorithms, {ESA}
  2024, Royal Holloway, London, United Kingdom, September 2-4, 2024}, volume
  308 of {\em LIPIcs}, pages 45:1--45:23. Schloss Dagstuhl - Leibniz-Zentrum
  f{\"{u}}r Informatik, 2024.

\bibitem[DHZ20]{DuanHZ20}
Ran Duan, Haoqing He, and Tianyi Zhang.
\newblock {A Scaling Algorithm for Weighted f-Factors in General Graphs}.
\newblock In Artur Czumaj, Anuj Dawar, and Emanuela Merelli, editors, {\em 47th
  International Colloquium on Automata, Languages, and Programming (ICALP
  2020)}, volume 168 of {\em Leibniz International Proceedings in Informatics
  (LIPIcs)}, pages 41:1--41:17, Dagstuhl, Germany, 2020. Schloss Dagstuhl --
  Leibniz-Zentrum f{\"u}r Informatik.

\bibitem[DK26]{DeyK26}
Dipan Dey and Telikepalli Kavitha.
\newblock Low-cost arborescence under edge faults, 2026.

\bibitem[DM76]{DaviesMcDiarmid1976}
Joan Davies and Colin McDiarmid.
\newblock Disjoint common transversals and exchange structures.
\newblock {\em Journal of the London Mathematical Society}, s2-14(1):55--62,
  1976.

\bibitem[DN95]{DinitzN95}
Yefim Dinitz and Zeev Nutov.
\newblock A 2-level cactus model for the system of minimum and minimum+1
  edge-cuts in a graph and its incremental maintenance.
\newblock In {\em Proceedings of the Twenty-Seventh Annual ACM Symposium on
  Theory of Computing}, STOC '95, page 509–518, New York, NY, USA, 1995.
  Association for Computing Machinery.

\bibitem[DP09]{DuanP09}
Ran Duan and Seth Pettie.
\newblock {\em Dual-Failure Distance and Connectivity Oracles}, pages 506--515.
\newblock 2009.

\bibitem[DP10]{DuanP:10}
Ran Duan and Seth Pettie.
\newblock Connectivity oracles for failure prone graphs.
\newblock In {\em Proceedings of the 42nd {ACM} Symposium on Theory of
  Computing, {STOC} 2010, Cambridge, Massachusetts, USA, 5-8 June 2010}, pages
  465--474, 2010.

\bibitem[DP17]{DuanP:17}
Ran Duan and Seth Pettie.
\newblock Connectivity oracles for graphs subject to vertex failures.
\newblock In {\em Proceedings of the Twenty-Eighth Annual {ACM-SIAM} Symposium
  on Discrete Algorithms, {SODA} 2017, Barcelona, Spain, Hotel Porta Fira,
  January 16-19}, pages 490--509, 2017.

\bibitem[DR22]{DuanR22}
Ran Duan and Hanlin Ren.
\newblock Maintaining exact distances under multiple edge failures.
\newblock In {\em Proceedings of the 54th Annual ACM SIGACT Symposium on Theory
  of Computing}, STOC 2022, page 1093–1101, New York, NY, USA, 2022.
  Association for Computing Machinery.

\bibitem[DTCR08]{DTCR08}
Camil Demetrescu, Mikkel Thorup, Rezaul~Alam Chowdhury, and Vijaya
  Ramachandran.
\newblock Oracles for distances avoiding a failed node or link.
\newblock {\em {SIAM} J. Comput.}, 37(5):1299--1318, 2008.

\bibitem[DV00]{DinitzV00}
Yefim Dinitz and Alek Vainshtein.
\newblock The general structure of edge-connectivity of a vertex subset in a
  graph and its incremental maintenance. odd case.
\newblock {\em {SIAM} J. Comput.}, 30(3):753--808, 2000.

\bibitem[dVC26]{VosC26}
Tijn de~Vos and Aleksander Christiansen.
\newblock Tree-packing revisited: Faster fully dynamic min-cut and arboricity.
\newblock {\em Algorithmica}, 88(3):52, Jun 2026.

\bibitem[dVG26]{VosG26}
Tijn de~Vos and Mara Grilnberger.
\newblock Dynamic matroids: Base packing and covering, 2026.

\bibitem[Edm03]{Edmonds2003}
Jack Edmonds.
\newblock {\em Submodular Functions, Matroids, and Certain Polyhedra}, pages
  11--26.
\newblock Springer Berlin Heidelberg, Berlin, Heidelberg, 2003.

\bibitem[FF62]{FulkersonF:62}
J.R. Ford and D.R. Fullkerson.
\newblock {\em Flows in networks}.
\newblock Princeton University Press, Princeton, 1962.

\bibitem[Fra90]{Frank90}
A.~Frank.
\newblock Augmenting graphs to meet edge-connectivity requirements.
\newblock In {\em Proceedings [1990] 31st Annual Symposium on Foundations of
  Computer Science}, pages 708--718 vol.2, 1990.

\bibitem[FS05]{FleischnerS05}
Herbert Fleischner and Stefan Szeider.
\newblock On edge-colored graphs covered by properly colored cycles.
\newblock {\em Graphs and Combinatorics}, 21:301--306, 2005.

\bibitem[GDHI{\etalchar{+}}17]{GeorgiadisDIKP17}
Loukas Georgiadis, Thomas Dueholm~Hansen, Giuseppe~F. Italiano, Sebastian
  Krinninger, and Nikos Parotsidis.
\newblock {Decremental Data Structures for Connectivity and Dominators in
  Directed Graphs}.
\newblock In Ioannis Chatzigiannakis, Piotr Indyk, Fabian Kuhn, and Anca
  Muscholl, editors, {\em 44th International Colloquium on Automata, Languages,
  and Programming (ICALP 2017)}, volume~80 of {\em Leibniz International
  Proceedings in Informatics (LIPIcs)}, pages 42:1--42:15, Dagstuhl, Germany,
  2017. Schloss Dagstuhl -- Leibniz-Zentrum f{\"u}r Informatik.

\bibitem[GHN{\etalchar{+}}23]{GoranciHNSTW23}
Gramoz Goranci, Monika Henzinger, Danupon Nanongkai, Thatchaphol Saranurak,
  Mikkel Thorup, and Christian Wulff-Nilsen.
\newblock {\em Fully Dynamic Exact Edge Connectivity in Sublinear Time}, pages
  70--86.
\newblock 2023.

\bibitem[GIK{\etalchar{+}}17]{GeorgiadisIKPP17}
Loukas Georgiadis, Giuseppe~F. Italiano, Aikaterini Karanasiou, Charis
  Papadopoulos, and Nikos Parotsidis.
\newblock Sparse certificates for 2-connectivity in directed graphs.
\newblock {\em Theor. Comput. Sci.}, 698:40--66, 2017.

\bibitem[GIP17]{GIP17}
Loukas Georgiadis, Giuseppe~F. Italiano, and Nikos Parotsidis.
\newblock Strong connectivity in directed graphs under failures, with
  applications.
\newblock In {\em Proceedings of the Twenty-Eighth Annual {ACM-SIAM} Symposium
  on Discrete Algorithms, {SODA} 2017, Barcelona, Spain, Hotel Porta Fira,
  January 16-19}, pages 1880--1899, 2017.

\bibitem[GP13]{GuptaP13}
Manoj Gupta and Richard Peng.
\newblock Fully dynamic (1+ e)-approximate matchings.
\newblock In {\em Proceedings of the 2013 IEEE 54th Annual Symposium on
  Foundations of Computer Science}, FOCS '13, page 548–557, USA, 2013. IEEE
  Computer Society.

\bibitem[GR21]{GuR21}
Yong Gu and Hanlin Ren.
\newblock {Constructing a Distance Sensitivity Oracle in $O(n^2.5794 M)$ Time}.
\newblock In Nikhil Bansal, Emanuela Merelli, and James Worrell, editors, {\em
  48th International Colloquium on Automata, Languages, and Programming (ICALP
  2021)}, volume 198 of {\em Leibniz International Proceedings in Informatics
  (LIPIcs)}, pages 76:1--76:20, Dagstuhl, Germany, 2021. Schloss Dagstuhl --
  Leibniz-Zentrum f{\"u}r Informatik.

\bibitem[GS18]{GuptaS18}
Manoj Gupta and Aditi Singh.
\newblock Generic single edge fault tolerant exact distance oracle.
\newblock In Ioannis Chatzigiannakis, Christos Kaklamanis, D{\'{a}}niel Marx,
  and Donald Sannella, editors, {\em 45th International Colloquium on Automata,
  Languages, and Programming, {ICALP} 2018, Prague, Czech Republic, July 9-13,
  2018}, volume 107 of {\em LIPIcs}, pages 72:1--72:15. Schloss Dagstuhl -
  Leibniz-Zentrum f{\"{u}}r Informatik, 2018.

\bibitem[GS21a]{GabowS21_I}
Harold~N. Gabow and Piotr Sankowski.
\newblock Algorithms for weighted matching generalizations i: Bipartite graphs,
  b-matching, and unweighted f-factors.
\newblock {\em SIAM Journal on Computing}, 50(2):440--486, 2021.

\bibitem[GS21b]{GabowS21_II}
Harold~N. Gabow and Piotr Sankowski.
\newblock Algorithms for weighted matching generalizations ii: f-factors and
  the special case of shortest paths.
\newblock {\em SIAM Journal on Computing}, 50(2):555--601, 2021.

\bibitem[GW19]{GrandoniVW19}
Fabrizio Grandoni and Virginia~Vassilevska Williams.
\newblock Faster replacement paths and distance sensitivity oracles.
\newblock {\em ACM Trans. Algorithms}, 16(1), December 2019.

\bibitem[Har09]{Harvey09}
Nicholas J.~A. Harvey.
\newblock Algebraic algorithms for matching and matroid problems.
\newblock {\em SIAM Journal on Computing}, 39(2):679--702, 2009.

\bibitem[HKP24]{HuKP24}
Bingbing Hu, Evangelos Kosinas, and Adam Polak.
\newblock {Connectivity Oracles for Predictable Vertex Failures}.
\newblock In Timothy Chan, Johannes Fischer, John Iacono, and Grzegorz Herman,
  editors, {\em 32nd Annual European Symposium on Algorithms (ESA 2024)},
  volume 308 of {\em Leibniz International Proceedings in Informatics
  (LIPIcs)}, pages 72:1--72:16, Dagstuhl, Germany, 2024. Schloss Dagstuhl --
  Leibniz-Zentrum f{\"u}r Informatik.

\bibitem[HS24]{HuangS24}
Chien-Chung Huang and François Sellier.
\newblock {\em Robust Sparsification for Matroid Intersection with
  Applications}, pages 2916--2940.
\newblock 2024.

\bibitem[HSSY24]{HenzingerSSY24}
Monika Henzinger, Barna Saha, Martin~P. Seybold, and Christopher Ye.
\newblock {On the Complexity of Algorithms with Predictions for Dynamic Graph
  Problems}.
\newblock In Venkatesan Guruswami, editor, {\em 15th Innovations in Theoretical
  Computer Science Conference (ITCS 2024)}, volume 287 of {\em Leibniz
  International Proceedings in Informatics (LIPIcs)}, pages 62:1--62:25,
  Dagstuhl, Germany, 2024. Schloss Dagstuhl -- Leibniz-Zentrum f{\"u}r
  Informatik.

\bibitem[IKP21]{ItalianoKP21}
Giuseppe~F. Italiano, Adam Karczmarz, and Nikos Parotsidis.
\newblock Planar reachability under single vertex or edge failures.
\newblock In D{\'{a}}niel Marx, editor, {\em Proceedings of the 2021 {ACM-SIAM}
  Symposium on Discrete Algorithms, {SODA} 2021, Virtual Conference, January 10
  - 13, 2021}, pages 2739--2758. {SIAM}, 2021.

\bibitem[JK82]{JensenK82}
Per~M. Jensen and Bernhard Korte.
\newblock Complexity of matroid property algorithms.
\newblock {\em SIAM Journal on Computing}, 11(1):184--190, 1982.

\bibitem[JS22]{JinS22}
Wenyu Jin and Xiaorui Sun.
\newblock Fully dynamic s-t edge connectivity in subpolynomial time (extended
  abstract).
\newblock In {\em 2021 IEEE 62nd Annual Symposium on Foundations of Computer
  Science (FOCS)}, pages 861--872, 2022.

\bibitem[JST24]{JinST24}
Wenyu Jin, Xiaorui Sun, and Mikkel Thorup.
\newblock {\em Fully Dynamic Min-Cut of Superconstant Size in Subpolynomial
  Time}, pages 2999--3026.
\newblock 2024.

\bibitem[Kar24]{Karczmarz24}
Adam Karczmarz.
\newblock Max \emph{s}, \emph{t}-flow oracles and negative cycle detection in
  planar digraphs.
\newblock In David~P. Woodruff, editor, {\em Proceedings of the 2024 {ACM-SIAM}
  Symposium on Discrete Algorithms, {SODA} 2024, Alexandria, VA, USA, January
  7-10, 2024}, pages 1606--1620. {SIAM}, 2024.

\bibitem[Kis23]{Kiss23}
Peter Kiss.
\newblock Deterministic dynamic matching in worst-case update time.
\newblock {\em Algorithmica}, 85(12):3741--3765, Dec 2023.

\bibitem[KKPW22]{KimKPW22}
Eun~Jung Kim, Stefan Kratsch, Marcin Pilipczuk, and Magnus Wahlstr\"{o}m.
\newblock Directed flow-augmentation.
\newblock STOC 2022, page 938–947, New York, NY, USA, 2022. Association for
  Computing Machinery.

\bibitem[KS23a]{KarczmarzS23}
Adam Karczmarz and Piotr Sankowski.
\newblock Sensitivity and dynamic distance oracles via generic matrices and
  frobenius form.
\newblock In {\em 64th {IEEE} Annual Symposium on Foundations of Computer
  Science, {FOCS} 2023, Santa Cruz, CA, USA, November 6-9, 2023}, pages
  1745--1756. {IEEE}, 2023.

\bibitem[KS23b]{KarczmarzSmule23}
Adam Karczmarz and Marcin Smulewicz.
\newblock {On Fully Dynamic Strongly Connected Components}.
\newblock In Inge~Li G{\o}rtz, Martin Farach-Colton, Simon~J. Puglisi, and
  Grzegorz Herman, editors, {\em 31st Annual European Symposium on Algorithms
  (ESA 2023)}, volume 274 of {\em Leibniz International Proceedings in
  Informatics (LIPIcs)}, pages 68:1--68:15, Dagstuhl, Germany, 2023. Schloss
  Dagstuhl -- Leibniz-Zentrum f{\"u}r Informatik.

\bibitem[Law75]{Lawler75}
Eugene~L. Lawler.
\newblock Matroid intersection algorithms.
\newblock {\em Mathematical Programming}, 9(1):31--56, Dec 1975.

\bibitem[Law76]{Lawler76}
Eugene~L. Lawler.
\newblock {\em Combinatorial Optimization: Networks and Matroids}.
\newblock Holt, Rinehart and Winston, New York, 1976.

\bibitem[Lov80]{Lovasz80}
L~Lovász.
\newblock Matroid matching and some applications.
\newblock {\em Journal of Combinatorial Theory, Series B}, 28(2):208--236,
  1980.

\bibitem[LS22]{LongS22}
Yaowei Long and Thatchaphol Saranurak.
\newblock Near-optimal deterministic vertex-failure connectivity oracles.
\newblock In {\em 2022 IEEE 63rd Annual Symposium on Foundations of Computer
  Science (FOCS)}, pages 1002--1010, 2022.

\bibitem[LS24]{LiuS24}
Quanquan~C. Liu and Vaidehi Srinivas.
\newblock The predicted-updates dynamic model: Offline, incremental, and
  decremental to fully dynamic transformations.
\newblock In Shipra Agrawal and Aaron Roth, editors, {\em Proceedings of Thirty
  Seventh Conference on Learning Theory}, volume 247 of {\em Proceedings of
  Machine Learning Research}, pages 3582--3641. PMLR, 30 Jun--03 Jul 2024.

\bibitem[MMN{\etalchar{+}}25]{McCauleyMNNS25}
Samuel McCauley, Benjamin Moseley, Aidin Niaparast, Helia Niaparast, and Shikha
  Singh.
\newblock {Incremental Approximate Single-Source Shortest Paths with
  Predictions}.
\newblock In Keren Censor-Hillel, Fabrizio Grandoni, Jo\"{e}l Ouaknine, and
  Gabriele Puppis, editors, {\em 52nd International Colloquium on Automata,
  Languages, and Programming (ICALP 2025)}, volume 334 of {\em Leibniz
  International Proceedings in Informatics (LIPIcs)}, pages 117:1--117:20,
  Dagstuhl, Germany, 2025. Schloss Dagstuhl -- Leibniz-Zentrum f{\"u}r
  Informatik.

\bibitem[MMNS24]{McCauleyMNS24}
Samuel McCauley, Benjamin Moseley, Aidin Niaparast, and Shikha Singh.
\newblock Incremental topological ordering and cycle detection with
  predictions.
\newblock In {\em Proceedings of the 41st International Conference on Machine
  Learning}, ICML'24. JMLR.org, 2024.

\bibitem[NSV94]{NarayananSV94}
H.~Narayanan, Huzur Saran, and Vijay~V. Vazirani.
\newblock Randomized parallel algorithms for matroid union and intersection,
  with applications to arborescences and edge-disjoint spanning trees.
\newblock {\em SIAM Journal on Computing}, 23(2):387--397, 1994.

\bibitem[PQ82]{PicardQ82}
Jean{-}Claude Picard and Maurice Queyranne.
\newblock On the structure of all minimum cuts in a network and applications.
\newblock {\em Math. Program.}, 22(1):121, 1982.

\bibitem[PR23]{PengR23}
Binghui Peng and Aviad Rubinstein.
\newblock {\em Fully-dynamic-to-incremental reductions with known deletion
  order (e.g. sliding window)}, pages 261--271.
\newblock 2023.

\bibitem[PT07]{PatrascuT:07}
Mihai Patrascu and Mikkel Thorup.
\newblock Planning for fast connectivity updates.
\newblock In {\em 48th Annual {IEEE} Symposium on Foundations of Computer
  Science {(FOCS} 2007), October 20-23, 2007, Providence, RI, USA,
  Proceedings}, pages 263--271, 2007.

\bibitem[Qua24]{Quanrud24}
Kent Quanrud.
\newblock {\em Faster exact and approximation algorithms for packing and
  covering matroids via push-relabel}, pages 2305--2336.
\newblock 2024.

\bibitem[San04]{Sankowski04}
Piotr Sankowski.
\newblock Dynamic transitive closure via dynamic matrix inverse (extended
  abstract).
\newblock In {\em 45th Symposium on Foundations of Computer Science, {FOCS}
  2004, Rome, Italy, October 17-19, 2004, Proceedings}, pages 509--517. {IEEE}
  Computer Society, 2004.

\bibitem[San05]{Sankowski05}
Piotr Sankowski.
\newblock Shortest paths in matrix multiplication time.
\newblock In Gerth~St{\o}lting Brodal and Stefano Leonardi, editors, {\em
  Algorithms - {ESA} 2005, 13th Annual European Symposium, Palma de Mallorca,
  Spain, October 3-6, 2005, Proceedings}, volume 3669 of {\em Lecture Notes in
  Computer Science}, pages 770--778. Springer, 2005.

\bibitem[Sto15]{Storjohann15}
Arne Storjohann.
\newblock On the complexity of inverting integer and polynomial matrices.
\newblock {\em computational complexity}, 24(4):777--821, Dec 2015.

\bibitem[Ter25]{Terao25}
Tatsuya Terao.
\newblock Faster matroid partition algorithms.
\newblock {\em ACM Trans. Algorithms}, 21(2), January 2025.

\bibitem[Tho07]{Thorup07}
Mikkel Thorup.
\newblock Fully-dynamic min-cut*.
\newblock {\em Combinatorica}, 27(1):91--127, Feb 2007.

\bibitem[vdB21]{Brand21}
Jan van~den Brand.
\newblock {\em Unifying Matrix Data Structures: Simplifying and Speeding up
  Iterative Algorithms}, pages 1--13.
\newblock 2021.

\bibitem[vdBCK{\etalchar{+}}]{BrandCKLPPSS24}
Jan van~den Brand, Li~Chen, Rasmus Kyng, Yang~P. Liu, Richard Peng,
  Maximilian~Probst Gutenberg, Sushant Sachdeva, and Aaron Sidford.
\newblock {\em Incremental Approximate Maximum Flow on Undirected Graphs in
  Subpolynomial Update Time}, pages 2980--2998.

\bibitem[VDBCK{\etalchar{+}}24]{BrandCKLMGS24}
Jan Van Den~Brand, Li~Chen, Rasmus Kyng, Yang~P. Liu, Simon Meierhans,
  Maximilian~Probst Gutenberg, and Sushant Sachdeva.
\newblock { Almost-Linear Time Algorithms for Decremental Graphs: Min-Cost Flow
  and More via Duality }.
\newblock In {\em 2024 IEEE 65th Annual Symposium on Foundations of Computer
  Science (FOCS)}, pages 2010--2032, Los Alamitos, CA, USA, October 2024. IEEE
  Computer Society.

\bibitem[vdBFNP24]{BrandFNP24}
Jan van~den Brand, Sebastian Forster, Yasamin Nazari, and Adam Polak.
\newblock {\em On Dynamic Graph Algorithms with Predictions}, pages 3534--3557.
\newblock 2024.

\bibitem[vdBKZ26]{BrandKZ26}
Jan van~den Brand, Vishal Kumar, and Daniel~J. Zhang.
\newblock {Dynamic Rank, Basis, and Matching}.
\newblock In Sayan Bhattacharya, Danupon Nanongkai, Michael Benedikt, and
  Gabriele Puppis, editors, {\em 53rd International Colloquium on Automata,
  Languages, and Programming (ICALP 2026)}, volume 374 of {\em Leibniz
  International Proceedings in Informatics (LIPIcs)}, pages 45:1--45:26,
  Dagstuhl, Germany, 2026. Schloss Dagstuhl -- Leibniz-Zentrum f{\"u}r
  Informatik.

\bibitem[vdBNS19]{BrandNS19}
Jan van~den Brand, Danupon Nanongkai, and Thatchaphol Saranurak.
\newblock { Dynamic Matrix Inverse: Improved Algorithms and Matching
  Conditional Lower Bounds }.
\newblock In {\em 2019 IEEE 60th Annual Symposium on Foundations of Computer
  Science (FOCS)}, pages 456--480, Los Alamitos, CA, USA, November 2019. IEEE
  Computer Society.

\bibitem[vdBS19]{BrandS19}
Jan van~den Brand and Thatchaphol Saranurak.
\newblock Sensitive distance and reachability oracles for large batch updates.
\newblock In {\em 60th {IEEE} Annual Symposium on Foundations of Computer
  Science, {FOCS} 2019, Baltimore, Maryland, USA, November 9-12, 2019}, pages
  424--435, 2019.

\bibitem[WY13]{WeimannY13}
Oren Weimann and Raphael Yuster.
\newblock {Replacement Paths and Distance Sensitivity Oracles via Fast Matrix
  Multiplication}.
\newblock {\em ACM Transactions on Algorithms}, 9:14:1--14:13, 2013.

\end{thebibliography}
\end{document}